%% file: main.tex
\documentclass[ejsv2,preprint,noshowframe]{imsart}

\RequirePackage{natbib}
\RequirePackage{amsmath}
\RequirePackage{xcolor}
\RequirePackage{url}
\RequirePackage{multirow}
\RequirePackage{booktabs}
\RequirePackage{array}
\RequirePackage{flafter}
\RequirePackage{placeins}
\RequirePackage[hyperindex,breaklinks,colorlinks,citecolor=blue,urlcolor=blue]{hyperref}
\RequirePackage{comment}
\RequirePackage{graphicx}
\RequirePackage{subcaption}
\startlocaldefs
\DeclareMathOperator*{\argminA}{arg\,min}

\theoremstyle{definition}
\newtheorem{assu}{Assumption}
\newtheorem{sassu}{Structural Assumption}

\theoremstyle{plain}
\newtheorem{prop}{Proposition}

\newtheorem{theorem}{Theorem}
\newtheorem{lemma}{Lemma}

\theoremstyle{definition}

\newtheorem{remark}{Remark}
\endlocaldefs

\begin{document}

\begin{frontmatter}

\title{Joint-Sparse Transfer Learning for High-Dimensional Multi-Output Regression}
\runtitle{Joint-Sparse Multitask Transfer Learning}

\begin{aug}
\author[A]{\fnms{Sunwoo}~\snm{Lim}\ead[label=e1]{sunwoo.lim@marshall.usc.edu}}
\and
\author[A,B]{\fnms{Mladen}~\snm{Kolar}\ead[label=e2]{mkolar@marshall.usc.edu}}

\address[A]{Marshall School of Business, University of Southern California
\printead[presep={,\ }]{e1,e2}}
\address[B]{Mohamed bin Zayed University of Artificial Intelligence}
\runauthor{S. Lim and M. Kolar}
\end{aug}

\begin{abstract}
Multitask linear models can improve estimation and prediction by exploiting structure shared across responses, bridging taskwise fitting and complete pooling.
In many applications, however, the objective is estimation in a data-limited target domain, while data-rich but heterogeneous source domains are available.
Borrowing from these sources can improve efficiency but introduce bias.
We develop a joint-sparse transfer-learning framework for high-dimensional multi-output regression that combines shared predictor structure across responses with source--target similarity.
The framework yields two complementary estimators: a fused estimator that aggregates jointly fitted domain-specific coefficients and a target-based debiased estimator that adjusts for source-induced shifts.
Our error bounds show how transfer increases the available information and sharing predictors across responses reduces selection costs.
They also reveal a tradeoff: the fused estimator benefits from larger sources but may retain bias if source shifts point in similar directions,
whereas debiasing trades this bias
for additional estimation error governed by the smaller target sample.
Comparison with a minimax lower bound identifies regimes where 
the bounds match up to logarithmic factors, 
where matching remains unresolved,
and where projection onto a target-based convex set closes the gap.
Simulations and an analysis of single-cell RNA and surface-protein profiles across cell types support the theory.
\end{abstract}



\end{frontmatter}

\tableofcontents

\section{Introduction} \label{sec:intro}\label{sec:motivation}
Modern scientific studies increasingly collect high-dimensional predictors
and several related outcomes, creating a need for statistical methods that
estimate multiple regression models jointly. A prominent example arises in
single-cell multi-omics, where RNA expression and surface protein abundance are
measured in the same cells \citep{luecken2021sandbox}. In this setting, one may
use thousands of gene-expression features to predict a panel of protein markers,
yielding a multi-output regression problem whose coefficient matrix describes
RNA--protein associations. Because only a relatively small subset of genes
is expected to drive variation across many protein markers,
joint sparsity across the response tasks is a natural structural assumption.

Such settings also motivate learning across related datasets. Target-only
estimation can be difficult when the domain of primary interest contains
relatively few observations compared with the number of genes and protein tasks,
whereas related source domains may contain substantially more observations.
In a bone marrow mononuclear cell study with finely annotated cell types, a rare cell type
may serve as a data-poor target while more abundant cell types provide larger
collections of matched RNA--protein profiles \citep{luecken2021sandbox}. These
source cell types share the same measurement space, but differences in
transcriptional programs and surface-marker usage can induce cell-type-specific
regression patterns. Transfer learning is therefore a natural strategy for
borrowing strength from related sources while preserving target-specific
structure \citep{pan2009survey}.

Two obstacles make this problem challenging. First, standard multitask learning
methods are typically formulated for a single homogeneous population. When
applied to multidomain data, they either use only the limited target sample or
pool all domains as if the regression coefficients were the same. Existing
high-dimensional transfer linear models relax exact pooling, but they have
largely been developed for single-output regression and measure source--target
similarity task by task \citep{bastani2021predicting, li2022transfer,
tian2023transfer, li2024estimation, liu2024unified, he2024transfusion,
qu2025automatic}. This taskwise approach misses an important feature of
multi-output biomedical data: cell-type-specific shifts may vary across protein
markers, whereas the relevant genes are shared across many response tasks.
Second, even informative sources can differ from the target in systematic,
directionally aligned ways. In such cases, borrowing source information by
direct aggregation can reduce variance but may introduce bias. Conversely,
treating every domain separately avoids this source-induced bias but discards
the large source sample sizes. A useful method must therefore exploit multitask
structure and account for source-induced bias by respecting domain
heterogeneity.

We propose a joint-sparse transfer framework in which the same row-level
structure governs both the target coefficient matrix and its relation to the
sources. That matrix has one row for each of the $p$ predictors and one column
for each of the $K$ responses, and at most $s$ of its rows are nonzero. Each
difference between a source coefficient matrix and the target is in turn
controlled in a mixed row norm, so a gene is borrowed or shrunk across the whole
protein panel rather than marker by marker, while its effects are still allowed
to differ between cell types. This structure is well suited to multi-omics
applications, where a small set of genes can influence several protein markers
in related but domain-specific ways. Within this framework we study Transfer
Multitask Learning (TMTL), which combines the target observations with those of
$L$ source domains, $N$ observations in total, and yields two estimators.
TMTL(Fused) fits all
domains jointly and aggregates the fitted coefficient matrices, and
TMTL(Debiased) adds to that aggregate a correction estimated from the target
sample alone. These complementary estimators allow us to study when aggregation is
effective and when target correction is beneficial.

\subsection{Contributions}

Our first set of results gives estimation error bounds for TMTL(Fused) and
TMTL(Debiased), in Theorems~\ref{thm:fused} and~\ref{thm:debiased} respectively,
whose leading term is $\sqrt{s\psi_K/N}$ with $\psi_K=1+\log p/K$. Two effects
are visible in this single expression. Transfer supplies the combined sample
size $N$ in place of the target sample size alone, and sharing predictors across
responses divides the $\log p$ paid for locating the active rows among the $K$
responses. Once $K\gtrsim\log p$, the factor $\psi_K$ is of constant order and
the term reduces to $\sqrt{s/N}$, the rate available when the active set is
known. With a single response, $\psi_K$ reverts to $1+\log p$ and the
amortization disappears, so this is an effect that only a multitask analysis can
exhibit.

The bounds also separate two ways in which source heterogeneity is costly. Write
$\bar h$ for the weighted sum of the magnitudes of the individual source--target
contrasts, and $\epsilon_D$ for the magnitude of their weighted sum, so that
$\epsilon_D\leq\bar h$ by the triangle inequality. The first measures what it
costs to estimate the contrasts, the second how much of the source shift
survives aggregation. They coincide when the shifts point in a common direction
and separate when they cancel across domains, coefficient by coefficient. Only
$\bar h$ enters the bound for TMTL(Debiased), whereas the bound for TMTL(Fused)
carries $\epsilon_D$ as well; against that, the fused heterogeneity term is
divided by the source sample size and the debiased one by the smaller target
sample size. The fused bound benefits from the larger source sample size but
includes an aggregation-bias term. Its advantage therefore depends on whether
that bias is small enough to offset the cost of estimating a correction from the
smaller target sample. The distinction itself is inherited from the
scalar-response analysis of \citet{he2024transfusion}; what the multitask
setting adds is that both quantities are mixed row norms, so each aggregates a
whole row of response-specific shifts before the sum over predictors is taken.

Upper bounds alone cannot say how much of this cost is intrinsic to the problem.
Theorem~\ref{thm:multitask-transfer-lower-bound} therefore bounds below the risk
of every estimator with access to all target and source data, over a class in
which every source lies within a common radius $h$ of the target. It separates a
pooled estimation cost from a target--source separation cost whose form changes
across small, intermediate and large radii. For $K=1$ it recovers the
scalar-response benchmarks of \citet{li2022transfer} and
\citet{tian2023transfer}. Minimax lower bounds for multitask transfer are
available under spectral similarity between the source and target coefficient
matrices \citep{zhao2026smart}; to our knowledge none is available for row
sparsity with mixed-norm contrasts.

Comparing the upper and lower bounds identifies which parts of the rates are
unavoidable. The estimation term $\sqrt{s\psi_K/N}$ is optimal whenever $\psi_K$
is comparable with the corresponding factor in the lower bound; TMTL(Debiased)
has the optimal dependence on the radius and on the target sample size at
intermediate radii; and TMTL(Fused) is optimal at small radii when $L\lesssim
s$. At large radii neither available upper bound establishes the rate that the
target sample alone would give, which is the best any estimator can do once the
sources are arbitrarily far from the target.
Theorem~\ref{thm:safe-projection} closes that gap:
projecting either estimator onto a convex set built from the target sample
preserves the bound it already satisfies and caps its error at that target-only
rate, at the price of additional calibration conditions.

A two-block alternating direction method of multipliers (ADMM)
\citep{boyd2011distributed} solves the fused program, and the empirical work
traces the behavior the theory predicts. The simulations vary the magnitude and
the alignment of the source shifts, and a single-cell RNA--protein analysis
treats cell types as domains; both include settings in which borrowing from the
sources is worse than ignoring them.

\subsection{Related work}
\label{sec:related-work}
Our work connects high-dimensional multitask regression with statistical
transfer learning. The central question is when shared predictor structure
across responses and information from heterogeneous sources provide
complementary gains in estimating a target regression matrix.

Multitask regression exploits shared representations
\citep{argyriou2008convex, maurer2013sparse}, low-rank structure
\citep{ando2005framework, pong2010trace}, or joint sparsity
\citep{obozinski2010joint, lounici2009taking, liu2009blockwise}, including
models with task-specific deviations \citep{yang2009heterogeneous,
jalali2010dirty}. Joint variable selection through group penalties
\citep{yuan2006model} is especially relevant here. It uses shared predictor
relevance to improve estimation across responses, with guarantees for prediction
and support recovery \citep{lounici2011oracle, obozinski2011support}. These
results typically concern a single homogeneous population. Data-shared and
stratified regression also accommodate common and group-specific effects
\citep{gross2016data, ollier2017regression, asiaee2018high}. Our focus is
estimation in one data-limited target domain using related sources whose
regression effects may differ from those of the target.

High-dimensional transfer learning uses source samples to improve target
estimation \citep{pan2009survey, bastani2021predicting, li2022transfer,
tian2023transfer, li2024estimation}. Recent developments address source and
variable selection, co-regularization, and adaptive transfer
\citep{liu2024unified, liu2025coregularization, qu2025automatic,
rauschenberger2025estimating}. Much of this theory concerns scalar responses.
Applying such methods separately to multiple responses leaves shared predictor
structure unused. Our estimators build on the fused estimation and target
correction strategy of TransFusion \citep{he2024transfusion}. In the
multi-output setting, we study how the gains from sharing predictors interact
with source heterogeneity, and when joint transfer improves on target-only
multitask learning or compares favorably with response-by-response transfer.

Multi-output transfer methods also exploit low-rank coefficient structure
\citep{park2025transfer} or shared spectral subspaces \citep{zhao2026smart}. We
focus on shared predictor relevance and similarity of predictor effects across
domains, a structure motivated by the role of genes across multiple protein
responses. Row sparsity also implies a rank restriction, so these model classes
need not be disjoint. However, our formulation does not require the spectral
subspace relations used by \citet{zhao2026smart}, and our theory does not rely
on low-rank assumptions. The two lines of work also differ in what they
establish. \citet{park2025transfer} give upper bounds only, whereas
\citet{zhao2026smart} prove a minimax lower bound for their spectral class.
Ours is, to our knowledge, the first for row sparsity with mixed-norm
source--target contrasts, and the two benchmarks are not comparable, since
neither parameter class contains the other.

In single-cell multi-omics, technologies such as CITE-seq
\citep{stoeckius2017simultaneous} and REAP-seq
\citep{peterson2017multiplexed} have motivated protein prediction from RNA
using neural networks, latent-variable models, and ensembles
\citep{zhou2020surface, gayoso2021joint, xu2021ensemble,
lakkis2022multiuse, yang2023interpretable, chen2024imputing}. Our application
uses the bone marrow data of \citet{luecken2021sandbox} to examine borrowing
across cell types while retaining cell-type-specific regression effects and
sharing predictor information across protein responses.

\subsection{Organization and notation}
The remainder of the paper is organized as follows.
Section~\ref{sec:problem-setup-tmtl} introduces the model setting and the
joint-sparse transfer framework. Section~\ref{sec:methodology-tmtl} presents
Transfer Multitask Learning (TMTL), including its computation and practical
tuning. Section~\ref{sec:theory} establishes estimation error bounds for the
fused and target-corrected estimators, presents a minimax lower bound, and
compares the bounds. Sections~\ref{sec:simulation_p200}
and~\ref{sec:realdata} report the simulation study and the single-cell
multi-omics application, respectively. Section~\ref{sec:discussion} concludes
with a discussion. The Appendices contain proofs, computational
details, extensions, and additional simulations.
Code for the proposed method, simulations, and real-data analysis is available at
\url{https://github.com/Damelim/TMTL}.

\textbf{Notation.} For $n\in\mathbb N$, let $[n]=\{1,\ldots,n\}$.
For a matrix $\mathbf A$, let $A_{jk}$ denote its $(j,k)$ entry and
$\mathbf A_j$ its $j$th row. We write $\|\mathbf A\|_F$ for the Frobenius
norm and $\|\mathbf v\|_2$ for the Euclidean norm of a vector $\mathbf v$.
For matrices of the same dimensions, $\langle\mathbf A,\mathbf B\rangle
=\operatorname{tr}(\mathbf A^\top\mathbf B)=\sum_{j,k}A_{jk}B_{jk}$ is the
Frobenius inner product, so that
$\|\mathbf A\|_F^2=\langle\mathbf A,\mathbf A\rangle$. We write
$\mathbf I_d$ for the $d\times d$ identity matrix, $\mathbf 0$ for a zero
vector or matrix whose dimensions are clear from the context,
$\otimes$ for the Kronecker product, and $\mathbf 1\{\cdot\}$ for the
indicator of an event.
For $\mathbf A\in\mathbb R^{p\times K}$, define the mixed row norms
\[
\|\mathbf A\|_{2,q}
=\left(\sum_{j=1}^p\|\mathbf A_j\|_2^q\right)^{1/q},
\quad 1\leq q<\infty,
\qquad
\|\mathbf A\|_{2,\infty}=\max_{j\in[p]}\|\mathbf A_j\|_2.
\]
For a symmetric matrix $\mathbf A$, $\lambda_{\min}(\mathbf A)$ and
$\lambda_{\max}(\mathbf A)$ denote its smallest and largest eigenvalues.

We write $a\wedge b=\min\{a,b\}$ and $a\vee b=\max\{a,b\}$.
For sequences $a_m$ and $b_m$,
$a_m=O(b_m)$ or $a_m\lesssim b_m$ means
that $|a_m/b_m|$ is bounded for large $m$, whereas $a_m=o(b_m)$ means
that $a_m/b_m\to0$. We write $a_m\gtrsim b_m$ when $b_m\lesssim a_m$, and
$a_m\asymp b_m$ when both $a_m=O(b_m)$
and $b_m=O(a_m)$.

\section{Model and problem setup} 
\label{sec:problem-setup-tmtl} 

Throughout, a task is one of the $K$ responses, and a domain is a population or
dataset from which observations are drawn. We observe data from a target domain,
indexed by $\ell=0$, and $L\geq1$ related source domains. The main text treats
the common-design model
\begin{equation}
    \mathbf Y^\ell=\mathbf X^{\ell}\mathbf B^\ell+\mathbf E^\ell,
    \qquad \ell=0,\ldots,L,
    \label{transferMTL_dgp_samecovariate}
\end{equation}
where $\mathbf Y^\ell\in\mathbb R^{n_\ell\times K}$ is the response matrix,
$\mathbf X^{\ell}\in\mathbb R^{n_\ell\times p}$ is the design matrix
shared across responses within domain $\ell$,
$\mathbf B^\ell\in\mathbb R^{p\times K}$ is the coefficient matrix, and
$\mathbf E^\ell\in\mathbb R^{n_\ell\times K}$ is the mean-zero error matrix.
The target sample size is $n_0=n_T$. For the main theoretical analysis, each
source has sample size $n_\ell=n_S$, $\ell\in[L]$, and the total sample size is
$N=n_T+Ln_S$.

For observation $i\in[n_\ell]$, let $\mathbf x_i^\ell\in\mathbb R^p$, $\mathbf
y_i^\ell\in\mathbb R^K$, and $\boldsymbol\epsilon_i^\ell\in\mathbb R^K$ denote
the transposes of the $i$th rows of $\mathbf X^\ell$, $\mathbf Y^\ell$, and
$\mathbf E^\ell$, respectively. The columns of $\mathbf Y^\ell$, $\mathbf
B^\ell$, and $\mathbf E^\ell$ are written as $\mathbf y^{\ell k}$,
$\boldsymbol\beta^{\ell k}$, and $\boldsymbol\epsilon^{\ell k}$, respectively,
for $k\in[K]$, and $\epsilon_i^{\ell k}$ denotes the $(i,k)$ entry of
$\mathbf E^\ell$. The target coefficient matrix $\mathbf B^0$ is the primary
parameter of interest. Row $j$ collects the effects of feature $j$ across all
response tasks; thus, estimating $\mathbf B^0$ amounts to learning the target
multi-output regression map.

We write the regression models without explicit intercepts. This formulation
assumes that each predictor and response has been centered within its domain;
equivalently, it allows an unpenalized domain- and task-specific intercept that
has been removed before estimating the slope matrices. In data analyses, the
centering constants must be estimated from the training observations and then
applied unchanged to validation and test observations.

To encode structure shared across tasks, let
\[
\mathcal{S}(\mathbf{B}^0) := \{j \in [p]: \mathbf{B}_j^0 \neq 0\}
\]
denote the target row support, the union of the task-specific supports.
\begin{sassu}[Joint row sparsity]
\label{assumption_joint_sparsity}
For an integer $s\in\{0,\ldots,p\}$, the target coefficient matrix satisfies
\[
|\mathcal{S}(\mathbf{B}^0)| \leq s.
\]
\end{sassu}
This condition restricts the number of predictors relevant across all responses;
within an active row, coefficients may be nonzero for some tasks and zero for
others. In the single-cell application, rows are genes and columns are protein
markers, so row sparsity says that a limited set of genes drives variation
across the protein markers, while an active gene need not affect every marker.

The source domains are informative only insofar as their coefficient matrices
are related to the target coefficient matrix. Rather than require
$\mathbf{B}^\ell=\mathbf{B}^0$, we allow source-specific deviations while
controlling their aggregate row-wise magnitude.
\begin{sassu}[Row-wise source--target contrasts]
\label{assumption_source_target_contrast}
For $\boldsymbol{h}:=(h_1,\ldots,h_L)^\top\in[0,\infty)^L$, the source and
target coefficient matrices satisfy
\[
\lVert \mathbf{B}^\ell - \mathbf{B}^0 \rVert_{2,1} \leq h_\ell,
\quad \ell \in [L].
\]
\end{sassu}
This condition allows domain-specific regression effects while bounding the
total discrepancy across predictor rows. The contrasts need not be row sparse; 
small differences may occur across many predictors.

Together, Structural
Assumptions~\ref{assumption_joint_sparsity}
and~\ref{assumption_source_target_contrast} define the parameter space
\newpage
\[
\Theta(s,\boldsymbol h)
:=
\left\{
\begin{aligned}
&(\mathbf B^0,\ldots,\mathbf B^L)\in(\mathbb R^{p\times K})^{L+1}:\\
&|\mathcal S(\mathbf B^0)|\le s,\quad
\lVert\mathbf B^\ell-\mathbf B^0\rVert_{2,1}\le h_\ell,\quad
\ell\in[L]
\end{aligned}
\right\}.
\]

The rates below depend on the realized contrasts through two distinct
quantities. Define
\begin{align}
h_\ell^\star := \lVert \mathbf{B}^\ell - \mathbf{B}^0 \rVert_{2,1},
\qquad
\bar h := \frac{n_S}{N} \sum\limits_{\ell=1}^{L} h_\ell^\star,
\qquad
\epsilon_D := \left\lVert \frac{n_S}{N} \sum\limits_{\ell=1}^{L} (\mathbf{B}^\ell - \mathbf{B}^0) \right\rVert_{2,1}.
\label{eq:heterogeneity-quantities}
\end{align}
Structural Assumption~\ref{assumption_source_target_contrast} gives
$h_\ell^\star\le h_\ell$. The quantity $\bar h$ is the weighted sum of
contrast magnitudes, whereas $\epsilon_D$ is the magnitude of the weighted
aggregate contrast and measures the source shift that survives aggregation.
The triangle inequality gives $\epsilon_D\leq\bar h$. Cancellation among
source deviations can make $\epsilon_D$ much smaller than $\bar h$, or even
zero; directionally aligned deviations can make the two quantities comparable.

We next introduce TMTL,
which incorporates this source--target structure through row-wise regularization
and then corrects systematic source-induced bias.

\section{Transfer Multitask Learning} 
\label{sec:methodology-tmtl}

Transfer Multitask Learning (TMTL) fits the target and source domains jointly
under model \eqref{transferMTL_dgp_samecovariate} and produces two estimators
of $\mathbf B^0$: TMTL(Fused), which aggregates the fitted coefficient
matrices, and TMTL(Debiased), which applies a target-data correction to that
aggregate.

Let $\lambda_0>0$ be the baseline regularization parameter, let $a_\ell>0$ be
the relative contrast weight for source domain $\ell\in[L]$, and write
$\lambda_\ell=\lambda_0 a_\ell$. Following the joint-fitting strategy of
\citet{he2024transfusion}, we estimate domain-specific coefficient matrices
while encouraging shared predictor structure across responses
\citep{lounici2009taking, obozinski2010joint, obozinski2011support}.
Specifically, we solve
\begin{equation}
\begin{aligned}
\widehat{\mathbf{B}}^0,\ldots,\widehat{\mathbf{B}}^L
& \in \argminA_{\mathbf{B}^0,\ldots,\mathbf{B}^L}
\ \Biggl[\frac{1}{2NK}\sum_{\ell=0}^{L}
\lVert\mathbf Y^\ell-\mathbf X^\ell\mathbf B^\ell\rVert_F^2
+\lambda_0\lVert\mathbf B^0\rVert_{2,1}
+\lambda_0\sum_{\ell=1}^{L}a_\ell
\lVert\mathbf B^\ell-\mathbf B^0\rVert_{2,1}\Biggr],
\end{aligned}
\label{cotrain}
\end{equation}
and set
\begin{equation}
\widehat{\mathbf{W}}^F := \frac{n_S}{N} \sum\limits_{\ell=1}^{L} \widehat{\mathbf{B}}^\ell + \frac{n_T}{N} \widehat{\mathbf{B}}^0.
\label{fused}
\end{equation}
The penalty $\lVert \mathbf{B}^{0} \rVert_{2,1}$ imposes joint row sparsity on
the target coefficient matrix, whereas $\lVert
\mathbf{B}^\ell-\mathbf{B}^0\rVert_{2,1}$ shrinks each source coefficient matrix
toward the target at the level of feature rows. The source coefficients may
therefore differ from the target, although large row-wise discrepancies are
discouraged. The fitted target block $\widehat{\mathbf B}^0$ is one of the
optimization variables, whereas the reported estimator $\widehat{\mathbf W}^F$
averages all fitted blocks with sample-size weights. Our error bounds and
predictions concern $\widehat{\mathbf W}^F$.

The population counterpart of $\widehat{\mathbf W}^F$ is
\newpage
\[
\mathbf{W}^F := \frac{n_S}{N} \sum\limits_{\ell=1}^{L} \mathbf{B}^\ell + \frac{n_T}{N} \mathbf{B}^0.
\]
The directional bias in \eqref{eq:heterogeneity-quantities}, $\epsilon_D=\lVert
\mathbf{W}^{F} - \mathbf{B}^0 \rVert_{2,1}$, is the mixed-norm approximation bias
incurred by targeting $\mathbf W^F$ instead of $\mathbf B^0$. When the source
shifts cancel across domains, $\epsilon_D$ is small and the fused estimator can
benefit from the larger effective sample size. When they point in the same
direction, $\epsilon_D$ can be comparable to $\bar h$, and the aggregate may
retain a substantial source shift. For a correction penalty
$\widetilde\lambda>0$, TMTL therefore adds a target-data correction step:
\begin{align}
\widehat{\mathbf{D}} &\in \argminA_{\mathbf{D} \in \mathbb{R}^{p \times K}} \frac{1}{2n_T K} \lVert \mathbf{Y}^0 - \mathbf{X}^{0} (\widehat{\mathbf{W}}^{F} + \mathbf{D}) \rVert_F^2 + \widetilde\lambda \lVert \mathbf{D} \rVert_{2,1}, \label{debiasedobjective} \\
\widehat{\mathbf{W}}^{D} &:= \widehat{\mathbf{W}}^{F} + \widehat{\mathbf{D}}. \label{debiasedestimator}
\end{align}
With $\widehat{\mathbf W}^F$ held fixed, the correction step uses the target
loss and a row-wise penalty on $\mathbf D$. The parameter
$\widetilde\lambda$ controls the magnitude of the correction: larger values keep
$\widehat{\mathbf W}^D$ close to $\widehat{\mathbf W}^F$, whereas smaller values
permit greater target-specific adjustment. We call $\widehat{\mathbf W}^D$ the
debiased estimator. Here, debiasing refers to reducing aggregation bias in point
estimation.

Both estimators use the raw source datasets directly, and both extend beyond
the common-design model of Section~\ref{sec:problem-setup-tmtl}.
Section~\ref{sec:generalization} of the Appendices defines them for
task-specific covariates and unequal domain sample sizes, without either
restriction.

\subsection{Computation}
\label{sec:opt-tmtl}
The fused objective \eqref{cotrain} is convex, but its contrast penalties
couple the target and source coefficient matrices. To separate the quadratic
loss from the penalties, introduce auxiliary copies
$\boldsymbol\Gamma^{00},\ldots,\boldsymbol\Gamma^{0L}$ of $\mathbf B^0$ and
$\boldsymbol\Gamma^1,\ldots,\boldsymbol\Gamma^L$ of the source matrices.
The equivalent constrained problem is
\begin{equation}
\begin{aligned}
\underset{\mathbf B,\boldsymbol\Gamma}{\operatorname{minimize}}\quad
&\frac{1}{2NK}\sum_{\ell=0}^L
\|\mathbf Y^\ell-\mathbf X^\ell\mathbf B^\ell\|_F^2
 +\lambda_0\|\boldsymbol\Gamma^{00}\|_{2,1}
 +\sum_{\ell=1}^L\lambda_\ell
 \|\boldsymbol\Gamma^\ell-\boldsymbol\Gamma^{0\ell}\|_{2,1},\\
\text{subject to}\quad
&\boldsymbol\Gamma^{0\ell}=\mathbf B^0\quad(0\leq\ell\leq L),
\qquad \boldsymbol\Gamma^\ell=\mathbf B^\ell\quad(\ell\in[L]).
\end{aligned}
\label{eq:admm-main-constrained}
\end{equation}
We use two-block ADMM \citep{boyd2011distributed}: one block contains all
coefficient matrices and the other all auxiliary matrices. Updates across
domains or rows within either block are independent components of an exact
block minimization.

Let $\mathbf U^\ell$ and $\mathbf V^\ell$ be the scaled dual variables for the
target-copy and source-copy constraints, with positive penalties
$\rho_{0\ell}$ and $\rho_\ell$. Cache
$\mathbf G^\ell=(\mathbf X^\ell)^\top\mathbf X^\ell/(NK)$ and
$\mathbf C^\ell=(\mathbf X^\ell)^\top\mathbf Y^\ell/(NK)$, and write
\begin{equation}
S_a(\mathbf v):=
\begin{cases}
\mathbf 0,&\lVert\mathbf v\rVert_2\leq a,\\
\bigl(1-a/\lVert\mathbf v\rVert_2\bigr)\mathbf v,&\lVert\mathbf v\rVert_2>a,
\end{cases}
\qquad a>0,
\label{eq:group-soft-threshold}
\end{equation}
for group soft thresholding, the proximal operator of $a\lVert\cdot\rVert_2$
\citep{parikh2014proximal}; in particular $S_a(\mathbf 0)=\mathbf 0$. Initialize the penalties at one and the auxiliary
and dual variables at zero. Each iteration performs the three updates below;
iteration indices are suppressed and each step uses the latest available
values.

The coefficient block solves one ridge-type system per domain,
\begin{align*}
\Bigl(\mathbf G^0+\textstyle\sum_{\ell=0}^L\rho_{0\ell}\mathbf I_p\Bigr)\mathbf B^0
&=\mathbf C^0+\textstyle\sum_{\ell=0}^L\rho_{0\ell}
 (\boldsymbol\Gamma^{0\ell}-\mathbf U^\ell),\\
(\mathbf G^\ell+\rho_\ell\mathbf I_p)\mathbf B^\ell
&=\mathbf C^\ell+\rho_\ell(\boldsymbol\Gamma^\ell-\mathbf V^\ell),
\qquad \ell\in[L].
\end{align*}
These systems differ across iterations only in their right-hand sides and in a
scalar on the diagonal, so one eigendecomposition of each $\mathbf G^\ell$
serves every iteration, and the source systems can be solved in parallel.

The auxiliary block separates over feature rows. For $j\in[p]$,
\[
\boldsymbol\Gamma_j^{00}\leftarrow S_{\lambda_0/\rho_{00}}
(\mathbf B_j^0+\mathbf U_j^0),
\]
which is what sets rows of $\widehat{\mathbf B}^0$ exactly to zero. For each
source $\ell\in[L]$, set
$\mathbf q_j^{0\ell}=\mathbf B_j^0+\mathbf U_j^\ell$,
$\mathbf q_j^\ell=\mathbf B_j^\ell+\mathbf V_j^\ell$ and
$\mu_\ell=\rho_{0\ell}\rho_\ell/(\rho_{0\ell}+\rho_\ell)$. The contrast penalty
couples the two copies, but reparametrizing by their difference decouples the
row problem, leaving the single threshold
$\mathbf r_j^\ell=S_{\lambda_\ell/\mu_\ell}
(\mathbf q_j^\ell-\mathbf q_j^{0\ell})$ and the updates
\[
\boldsymbol\Gamma_j^{0\ell}\leftarrow
\frac{\rho_{0\ell}\mathbf q_j^{0\ell}+\rho_\ell\mathbf q_j^\ell
-\rho_\ell\mathbf r_j^\ell}{\rho_{0\ell}+\rho_\ell},
\qquad
\boldsymbol\Gamma_j^\ell\leftarrow
\frac{\rho_{0\ell}\mathbf q_j^{0\ell}+\rho_\ell\mathbf q_j^\ell
+\rho_{0\ell}\mathbf r_j^\ell}{\rho_{0\ell}+\rho_\ell}.
\]
Shrinking $\mathbf r_j^\ell$ is what pulls source row $j$ toward the target.
Finally, the dual step is multiplier ascent,
\[
\mathbf U^\ell\leftarrow\mathbf U^\ell+\mathbf B^0-\boldsymbol\Gamma^{0\ell},
\qquad
\mathbf V^\ell\leftarrow\mathbf V^\ell+\mathbf B^\ell-\boldsymbol\Gamma^\ell.
\]

We stop when the largest constraint-specific residual falls below a tolerance
or an iteration cap is reached, adapting each penalty from its own residual
pair during a fixed initial horizon. After that horizon the iteration is
standard two-block ADMM, so the objective converges to the optimum whenever a
primal--dual solution exists \citep[Section~3.2]{boyd2011distributed}. For
fixed $p$, the per-iteration cost is linear in $L$ and $K$, following a
one-time setup cubic in $p$. The group lasso correction
\eqref{debiasedobjective} uses the same algorithm.
Section~\ref{app:computation} of the Appendices gives the augmented
Lagrangian, the derivations, and the residual, penalty-adaptation, and cost
details.

\subsection{Hyperparameter selection}
\label{sec:hyper-param-tmtl}
TMTL uses separate regularization parameters for the fused fit and the target
correction. In \eqref{cotrain}, the baseline penalty $\lambda_0$ sets the
overall level of row-wise regularization and the source penalties are
$\lambda_\ell=a_\ell\lambda_0$; a larger $a_\ell$ penalizes the contrast with
source $\ell$ more heavily and so borrows more from it. In
\eqref{debiasedobjective}, $\widetilde\lambda$ sets the size of the target
correction. Tuning is based on held-out prediction performance in the target
domain, with $\widetilde\lambda$ selected after the fused fit.

Tuning $\lambda_0,\lambda_1,\ldots,\lambda_L$ independently is feasible only
for small $L$, because an unrestricted Cartesian grid grows exponentially in
$L$. For many source domains we instead restrict $(a_1,\ldots,a_L)$ to a small
number of prespecified weight vectors and tune $\lambda_0$ within each family,
so the number of fused fits is the number of $\lambda_0$ values times the
number of weight vectors. Two candidate families come from the theory: in the
balanced-source setting of Section~\ref{sec:theory}, the favorable scalings are
$a_\ell\asymp\sqrt{n_S/N}$ and $a_\ell\asymp n_S/N$, one for each
heterogeneity regime. These scalings depend only on the sample sizes, so they
can be used without knowing $s$, $\bar h$, or $\epsilon_D$. A weight vector
may be fixed in advance or selected by comparing candidate families on target
validation data.

After selecting the regularization parameters by held-out target prediction
error, we refit the estimators and evaluate them on separate test observations.
Sections~\ref{sec:simulation_p200} and~\ref{sec:realdata} specify the candidate
weights, tuning grids, and evaluation settings.

\section{Estimation error bounds and a minimax lower bound} \label{sec:theory}

This section states the technical conditions, gives an upper bound for each of
the two estimators over the parameter space $\Theta(s,\boldsymbol h)$, and
proves a minimax lower bound for the same estimation problem. The fused
estimator gains from the combined source--target sample size but pays the
directional bias $\epsilon_D/\sqrt K$; the target correction removes that term
under stronger sample-size conditions, while its rate continues to depend on
$\bar h$. The lower bound of Section~\ref{sec:lower-bound} provides a benchmark
for assessing which parts of these upper bounds are unavoidable and where gaps
remain. Proofs are in Sections~\ref{app:theory-proofs}
and~\ref{app:multitask-lower-bound} of the Appendices.

\subsection{Technical conditions}
\label{sec:technical-conditions-tmtl}
In addition to the structural assumptions in
Section~\ref{sec:problem-setup-tmtl}, we impose the following conditions on the
designs and errors.
\begin{assu}[Sub-Gaussian design]
For each domain, the rows of $\mathbf X^{\ell}$ are independent, mean-zero,
sub-Gaussian vectors with covariance $\boldsymbol{\Sigma}^{\ell}$. There is
a constant $\kappa_X<\infty$ such that
\[
\lVert u^\top\mathbf x_i^{\ell}\rVert_{\psi_2}
\leq
\kappa_X (u^\top\boldsymbol{\Sigma}^{\ell}u)^{1/2}
\quad\text{for all }u\in\mathbb R^p,\ 0\leq \ell\leq L.
\]
Here $\lVert Z\rVert_{\psi_2}:=\inf\{t>0:\mathbb E\exp(Z^2/t^2)\leq2\}$
is the sub-Gaussian Orlicz norm.
Moreover, the covariance eigenvalues are uniformly bounded: for some constant $C_X\geq1$,
\[
C_X^{-1}
\leq
\min_{0\leq \ell\leq L}\lambda_{\min}(\boldsymbol{\Sigma}^{\ell})
\leq
\max_{0\leq \ell\leq L}\lambda_{\max}(\boldsymbol{\Sigma}^{\ell})
\leq C_X.
\]
\label{assumption_design}
\end{assu}
\begin{assu}[Gaussian errors]
The errors $\epsilon_i^{\ell k}$ are independent across all observations,
tasks, and domains, with $\epsilon_i^{\ell k}\sim \mathcal N(0,\sigma_\ell^2)$ and
$\sigma_\ell^2\leq\sigma_{\max}^2$. They are jointly independent of the
design matrices $\mathbf X^0,\ldots,\mathbf X^L$.
\label{assumption_error}
\end{assu}
Assumption~\ref{assumption_design} imposes standard sub-Gaussian design and
well-conditioned covariance requirements uniformly across domains.
Assumption~\ref{assumption_error} requires independent Gaussian errors with
uniformly bounded variance; Gaussianity simplifies the concentration arguments
in the proofs.

All bounds are stated for the normalized Frobenius-norm error
\begin{align}
K^{-1/2}\lVert \widehat{\mathbf{W}}-\mathbf{B}^0\rVert_F
=
\left(
\frac1K\sum_{k=1}^K
\lVert\widehat{\boldsymbol w}^{k}-\boldsymbol\beta^{0k}\rVert_2^2
\right)^{1/2},
\label{eq:normalized-error}
\end{align}
where $\widehat{\boldsymbol w}^{k}$ and $\boldsymbol\beta^{0k}$ are the $k$-th
columns of $\widehat{\mathbf W}$ and $\mathbf B^0$. Thus
\eqref{eq:normalized-error} is the root mean squared coefficient error across
tasks, and normalizing by $K^{-1/2}$ makes the bounds comparable across
different numbers of tasks.

For later use, define the following quantity, which we call the task-complexity factor:
\begin{align}
\psi_K:=1+\frac{\log p}{K},
\label{eq:psi-K}
\end{align}
which quantifies the effective support-selection complexity when information is shared across $K$ tasks.

\subsection{Upper bound for TMTL(Fused)}
\label{sec:fused}
The error of $\widehat{\mathbf W}^F$ splits into estimation error around its
population counterpart $\mathbf W^F$ and the deterministic approximation bias
$\lVert \mathbf{W}^F-\mathbf{B}^0\rVert_{2,1}=\epsilon_D$.
\begin{theorem}[TMTL(Fused)]\label{thm:fused}
Consider a sequence of problem instances indexed by $m\to\infty$, along which
the dimensions and structural parameters may vary. Assume Structural
Assumptions~\ref{assumption_joint_sparsity}--\ref{assumption_source_target_contrast}
and Assumptions~\ref{assumption_design}--\ref{assumption_error}. Suppose that
$n_S \geq n_T$, $s\log p/n_S = o(1)$, and, for some fixed $M_F<\infty$,
\[
\frac{L\bar h\log p}{\sqrt{K n_S}}\le M_F
\]
for all sufficiently large $m$.
Fix a tail constant $c>8$. For a sufficiently large
constant $c_0>0$ depending only on the sub-Gaussian design constants, the
covariance eigenvalue bounds, $\sigma_{\max}$, and $c$, choose
\[
\lambda_0
=
c_0 N^{-1}\sqrt{\frac{(L+1)n_S}{K}\psi_K},
\qquad
\lambda_\ell
=
8\sqrt{\frac{n_S}{N}}\lambda_0,\quad \ell\in[L].
\]
Then, for all sufficiently large $m$, with probability at least
\[
1-C_1(L+1)Kp\exp(-C_2n_T)-C_3(L+1)p^{1-c},
\]
we have
\begin{align}
    \frac{1}{\sqrt{K}}
    \lVert \widehat{\mathbf{W}}^{F} - \mathbf{B}^0\rVert_F \lesssim 
    \sqrt{\frac{s\psi_K}{N}} +
    \sqrt{\frac{\bar h}{\sqrt{K}} \sqrt{\frac{\psi_K}{n_S}}} + \frac{\epsilon_D}{\sqrt{K}}.
    \label{rate_fused}
\end{align}
Here $C_1,C_2,C_3$ depend only on the fixed design and noise constants. The
hidden constant in $\lesssim$ may additionally depend on the fixed constants
$c$, $c_0$, and $M_F$, but not on $m$.
\end{theorem}
The bound in \eqref{rate_fused} has three components. The first,
\[
\sqrt{\frac{s\psi_K}{N}},
\]
is the joint-sparse multitask estimation error based on the total
source--target sample size $N$. The second,
\[
\sqrt{\frac{\bar h}{\sqrt K}
\sqrt{\frac{\psi_K}{n_S}}},
\]
is the cost of estimating source--target contrasts from the source samples. The
final term, $\epsilon_D/\sqrt K$, bounds the normalized Frobenius approximation
bias from using the population fused coefficient $\mathbf{W}^F$ in place of the
target coefficient $\mathbf{B}^0$. Because $\epsilon_D$ can be much smaller than
$\bar h$ when source deviations cancel, the theorem characterizes the regime
in which TMTL(Fused) is most attractive.

The probability bound combines a term of order $(L+1)Kp\exp(-C_2 n_T)$, arising from the restricted curvature and predictor column-norms, with row-wise empirical correlations between the design and noise (see Section~\ref{app:high-prob-events} of the Appendices).
It tends to one, for example, if
\[
\log\{(L+1)Kp\}=o(n_T)
\qquad\text{and}\qquad
(L+1)p^{1-c}=o(1).
\]
\begin{remark}[Comparison with target-only multitask learning]
Under the normalization and restricted-eigenvalue conditions of Corollary~4.1
of \citet{lounici2011oracle}, the target-only multitask rate is
\begin{align}
    \sqrt{\frac{s}{n_T}\left(1+\frac{A\log p}{K}\right)},
    \qquad A>5/2.
    \label{rate_lounici}
\end{align}
Writing $r_T$ for the rate in \eqref{rate_lounici}, the leading term in
\eqref{rate_fused} is $o(r_T)$ when
\[
\frac{n_T\psi_K}{N\{1+A\log p/K\}}=o(1).
\]
The full right-hand side of \eqref{rate_fused} is $o(r_T)$ if, in addition,
\[
\sqrt{\frac{\bar h}{\sqrt K}\sqrt{\frac{\psi_K}{n_S}}}
+\frac{\epsilon_D}{\sqrt K}=o(r_T).
\]
Thus, a larger total sample size alone does not imply improvement, and this
comparison does not establish uniform dominance.
\end{remark}
\begin{remark}[Comparison with single-output transfer learning]
\label{rem:single-output-comparison}
Joint row sparsity allows the responses to share information about which
predictors are active. The leading estimation term in \eqref{rate_fused} is
\[
\sqrt{\frac{s\psi_K}{N}}
=
\sqrt{\frac{s}{N}+\frac{s\log p}{NK}},
\]
compared with a leading term of order $\sqrt{s\log p/N}$ for separate
scalar-response transfer fits when the individual task sparsities are
comparable to the union support size $s$. Thus, transfer increases the
available sample size, while multitask structure reduces the support-selection
component of the estimation bound. Under the conditions of
Theorem~\ref{thm:fused}, when $K$ is of order $\log p$ or larger, the leading
term is of order $\sqrt{s/N}$.

Section~\ref{app:single-output-comparison} of the Appendices makes this
comparison precise, matching \eqref{rate_fused} term by term against the
single-output transfer rate of \citet{he2024transfusion}: the mixed row norms
replace the $\ell_1$ norms, and the support-selection logarithm is divided
among the responses.

The above comparison is most informative when the responses share much of their
support. If individual tasks involve substantially fewer predictors than the
union, the corresponding bounds for separate fits can be smaller than the
benchmark above. The overall benefit also depends on how $\bar h$ and
$\epsilon_D$ scale with $K$. Without additional scaling assumptions, the
comparison does not imply uniform dominance over single-output transfer
learning.

\end{remark}
\subsection{Upper bound for TMTL(Debiased)}
\label{sec:debiased}
When the source--target shifts are directionally aligned, $\epsilon_D$ is
comparable to $\bar h$, so the directional-bias term in \eqref{rate_fused} can
be substantial.
The row-sparse correction in
\eqref{debiasedobjective}--\eqref{debiasedestimator} removes that term, at the
cost of stronger sample-size conditions, because the correction is estimated
from the target sample alone.
\begin{theorem}[TMTL(Debiased)]\label{thm:debiased}
Fix $c>8$ and assume the structural, design, error, sample-size, and sequence
conditions of Theorem~\ref{thm:fused}, but not its tuning specification. In
addition, suppose that, for some fixed $M_T<\infty$,
$s\log p/n_T\le M_T$ for all sufficiently large $m$, and
\[
\frac{\bar h\log p}{\sqrt{K n_T}}+\frac{Ls\log p}{n_S}=o(1).
\]
Define regime~A by
\[
\frac{\bar h}{\sqrt K}
\sqrt{\frac{\psi_K}{n_T}}
\leq
\frac{s\psi_K}{n_S}.
\]
Choose $\lambda_\ell=a_\ell\lambda_0$ with
\[
(\lambda_0,a_\ell)=
\begin{cases}
\left(
\dfrac{c_0}{N} \sqrt{\dfrac{(L+1)n_S}{K}\psi_K},
8\sqrt{n_S/N}
\right), & \text{in regime A},\\[3mm]
\left(
c_0\sqrt{\dfrac{L+1}{KN}\psi_K},
8n_S/N
\right), & \text{otherwise},
\end{cases}
\]
where $c_0$ is chosen sufficiently large as in Theorem~\ref{thm:fused}, and set
\[
\widetilde\lambda
=
c_0\sqrt{\frac{\psi_K}{K n_T}}.
\]
Then, for all sufficiently large $m$, with probability at least
\[
1-C_1(L+1)Kp\exp(-C_2n_T)-C_3(L+1)p^{1-c},
\]
\begin{align}
\begin{split}
    \frac{1}{\sqrt{K}} \lVert \widehat{\mathbf{W}}^D - \mathbf{B}^0 \rVert_F
    \lesssim 
    \sqrt{\frac{s\psi_K}{N}} +
    \sqrt{\frac{\bar h}{\sqrt{K}} \sqrt{\frac{\psi_K}{n_T}}}.
\end{split}
\label{rate_debiased}
\end{align}
Here $C_1,C_2,C_3$ depend only on the fixed design and noise constants. The
hidden constant in $\lesssim$ may additionally depend on the fixed constants
$c$, $c_0$, $M_F$, and $M_T$, but not on $m$.
\end{theorem}
\begin{remark}[Oracle tuning]
The theorem chooses between two regularization scalings according to whether
regime~A holds, which depends on the unknown sparsity and heterogeneity levels
$s$ and $\bar h$. This oracle choice characterizes the achievable rate in the
two source-heterogeneity regimes. In practice, validation can compare candidates
under both scalings, as described in Section~\ref{sec:hyper-param-tmtl}.
\label{oracle_tuning}
\end{remark}
The first term in \eqref{rate_debiased} is the joint-sparse estimation term from
the fused bound. The second term depends on $n_T$ rather than $n_S$ because the
correction is estimated from the target data. Unlike \eqref{rate_fused}, the
debiased rate contains no explicit directional bias term
$\epsilon_D/\sqrt K$; instead, $\bar h$ captures the remaining cost of
heterogeneity. This benefit comes at the cost of the stronger sample-size
conditions above and potentially greater variance when the target sample is
small. Relative to the single-output debiased transfer rate of
\citet{he2024transfusion}, the
multi-output rate reflects row-sparse sharing across tasks through $\psi_K$.
\subsection{A minimax lower bound}
\label{sec:lower-bound}

Theorems~\ref{thm:fused} and~\ref{thm:debiased} say what the two estimators
achieve. To judge which parts of those rates are intrinsic to the problem, we
now bound below the error of every estimator that uses all of the target and
source data. We work in a homogeneous Gaussian submodel, where the pairwise
Kullback--Leibler divergence between data distributions can be computed exactly,
and over a slice of $\Theta(s,\boldsymbol h)$ in which all sources share one
contrast radius. Because the submodel is contained in the model of
Section~\ref{sec:problem-setup-tmtl}, the resulting lower bound established for
the submodel is also a valid lower bound for the broader model. However, it does
not measure any additional difficulty caused by heterogeneous designs.

Let $n_0=n_T$, let $n_\ell=n_S$ for $\ell\in[L]$, and let $N=n_T+Ln_S$. For a
fixed $0<\sigma\leq\sigma_{\max}$, suppose that all design vectors and errors
are mutually independent and satisfy
\begin{equation}
\begin{aligned}
\mathbf x_i^\ell&\sim \mathcal N_p(\mathbf0,\mathbf I_p),\qquad
\boldsymbol\epsilon_i^\ell\sim \mathcal N_K(\mathbf0,\sigma^2\mathbf I_K),\\
\mathbf y_i^\ell&=(\mathbf B^\ell)^\top\mathbf x_i^\ell+
\boldsymbol\epsilon_i^\ell,\qquad i\in[n_\ell],\quad 0\leq\ell\leq L.
\end{aligned}
\label{eq:lower-gaussian-model}
\end{equation}
Write $\mathbb P_{\mathbf B}$ and $\mathbb E_{\mathbf B}$ for probability
and expectation under the coefficient tuple
$\mathbf B=(\mathbf B^0,\ldots,\mathbf B^L)$ in
\eqref{eq:lower-gaussian-model}.
For $h\geq0$, define
\begin{equation}
\begin{aligned}
\Theta_{\mathrm{MT}}(s,h):=\Theta\bigl(s,(h,\ldots,h)^\top\bigr)
=\left\{\begin{aligned}
&(\mathbf B^0,\ldots,\mathbf B^L)\in(\mathbb R^{p\times K})^{L+1}:\\
&|\mathcal S(\mathbf B^0)|\leq s,\quad
\max_{\ell\in[L]}\lVert\mathbf B^\ell-\mathbf B^0\rVert_{2,1}\leq h
\end{aligned}\right\}.
\end{aligned}
\label{eq:lower-parameter-class}
\end{equation}
The following theorem gives a minimax lower bound over this parameter space.

\begin{theorem}[Multitask transfer lower bound]
\label{thm:multitask-transfer-lower-bound}
Suppose $p,K,L,n_T,n_S$ are positive integers with $p\geq8$, and let
$s\in[p]$ satisfy $s\leq p/8$. Define
\begin{equation}
\omega_{q,K}:=1+\frac{1}{K}\log\left(\frac{\mathrm e p}{q}\right),\qquad 1\leq q\leq p,
\label{eq:lower-complexity}
\end{equation}
and, for $h\geq0$, set
\begin{equation}
r_{\mathrm{pool}}:=\frac{\sigma^2s\omega_{s,K}}{N},\qquad r_{\mathrm{shift}}:=\frac{\sigma h}{\sqrt K}\sqrt{\frac{\omega_{s,K}}{n_T}}\wedge\frac{h^2}{K}\wedge\frac{\sigma^2s\omega_{s,K}}{n_T}.
\label{eq:lower-rates}
\end{equation}
Under \eqref{eq:lower-gaussian-model}, there are universal constants
$c_{\mathrm{lb}},c_1>0$ such that, for every $h\geq0$,
\begin{equation}
\inf_{\widetilde{\mathbf B}}\sup_{\mathbf B\in\Theta_{\mathrm{MT}}(s,h)}\mathbb P_{\mathbf B}\left\{\frac{1}{K}\lVert\widetilde{\mathbf B}-\mathbf B^0\rVert_F^2\geq c_{\mathrm{lb}}\bigl(r_{\mathrm{pool}}+r_{\mathrm{shift}}\bigr)\right\}\geq c_1,
\label{eq:lower-probability}
\end{equation}
where the infimum is over all measurable estimators based on the complete target and
source samples. Consequently,
\begin{equation}
\inf_{\widetilde{\mathbf B}}\sup_{\mathbf B\in\Theta_{\mathrm{MT}}(s,h)}\mathbb E_{\mathbf B}\left[\frac{1}{K}\lVert\widetilde{\mathbf B}-\mathbf B^0\rVert_F^2\right]\geq c_{\mathrm{lb}}c_1\bigl(r_{\mathrm{pool}}+r_{\mathrm{shift}}\bigr).
\label{eq:lower-expectation}
\end{equation}
\end{theorem}

Even when all domains share the same coefficient matrix, estimating its
nonzero rows requires the pooled rate $r_{\mathrm{pool}}$, based on all $N$
observations. This is the only component that remains when $h=0$.

As the source--target differences grow, $r_{\mathrm{shift}}$ describes the
additional difficulty of estimating the target. For small $h$, this cost is
$h^2/K$. In the intermediate range of $h$, it grows linearly with $h$, at rate
$\sigma hK^{-1/2}(\omega_{s,K}/n_T)^{1/2}$. For large $h$, it reaches the
target-only rate $\sigma^2s\omega_{s,K}/n_T$: even an estimator with access
to every source sample cannot improve on that order uniformly over the class.

The group-sparse complexity $\omega_{s,K}$ in
\eqref{eq:lower-complexity} reflects both the $K$ coefficients within an
active row and the cost of selecting the active rows.

Setting $K=1$ recovers the scalar-response problem, and
$r_{\mathrm{pool}}+r_{\mathrm{shift}}$ then reduces, up to constants, to
\begin{equation}
\frac{\sigma^2s\log(\mathrm e p/s)}{N}+\left\{\sigma h\sqrt{\frac{\log(\mathrm e p/s)}{n_T}}\wedge h^2\wedge\frac{\sigma^2s\log(\mathrm e p/s)}{n_T}\right\},
\label{eq:lower-scalar-reduction}
\end{equation}
the benchmark established for homogeneous designs by \citet{li2022transfer}
and \citet{tian2023transfer}. Theorem~\ref{thm:multitask-transfer-lower-bound}
therefore extends those results rather than competing with them: the effect of
multiple responses is confined to the factor $\omega_{s,K}$, which divides the
support-selection logarithm among the $K$ tasks. The proof is given in
Section~\ref{app:multitask-lower-bound} of the Appendices.

\subsection{Matching regimes and the large-radius gap}
\label{sec:bound-comparison}
On $\Theta_{\mathrm{MT}}(s,h)$, we have $\bar h\leq h$ and $\epsilon_D\leq h$.
Under their respective assumptions, Theorems~\ref{thm:fused}
and~\ref{thm:debiased} therefore give the squared-error bounds
\[
\begin{aligned}
\frac{1}{K}\lVert\widehat{\mathbf W}^F-\mathbf B^0\rVert_F^2
&\lesssim\frac{s\psi_K}{N}
+\frac{h}{\sqrt K}\sqrt{\frac{\psi_K}{n_S}}+\frac{h^2}{K},\\
\frac{1}{K}\lVert\widehat{\mathbf W}^D-\mathbf B^0\rVert_F^2
&\lesssim\frac{s\psi_K}{N}
+\frac{h}{\sqrt K}\sqrt{\frac{\psi_K}{n_T}},
\end{aligned}
\]
with high probability. Throughout this subsection $\sigma$ is fixed and its
constant factors are suppressed.

The two task-complexity factors differ only in the logarithm they carry:
$\psi_K=1+\log p/K$ in the upper bounds and
$\omega_{s,K}=1+\log(\mathrm ep/s)/K$ in the lower bound. If
$s\leq p^{1-\delta}$ for a fixed $\delta>0$, then
$1+\delta\log p\leq\log(\mathrm ep/s)\leq1+\log p$, so the two factors agree up
to a constant depending only on $\delta$. The same holds when
$K\gtrsim\log p$, since both are then of constant order.  We
assume $\psi_K\asymp\omega_{s,K}$ in what follows.

Which component of $r_{\mathrm{shift}}$ binds depends on the radius. Writing
$\tau:=\sigma(K\omega_{s,K}/n_T)^{1/2}$, the three components of
\eqref{eq:lower-rates} are ordered so that $h^2/K$ is smallest for
$h\leq\tau$, the term $\sigma hK^{-1/2}(\omega_{s,K}/n_T)^{1/2}$ for
$\tau\leq h\leq s\tau$, and the target-only cap
$\sigma^2s\omega_{s,K}/n_T$ for $h\geq s\tau$. We refer to these as the small,
intermediate and large radius regimes.

The estimation term $s\psi_K/N$ in both upper bounds matches
$r_{\mathrm{pool}}$ up to constants, so the cost of estimating the shared row
structure from the combined sample is of optimal order in every regime. The
heterogeneity term $hK^{-1/2}\sqrt{\psi_K/n_T}$ in the bound for
TMTL(Debiased) likewise matches the binding component of $r_{\mathrm{shift}}$
at intermediate radii.

At small radii TMTL(Fused) matches the lower bound when $L\lesssim s$, since
Young's inequality gives
\newpage
\[
\frac{h}{\sqrt K}\sqrt{\frac{\psi_K}{n_S}}
\leq\frac12\left(\frac{h^2}{K}+\frac{\psi_K}{n_S}\right),
\]
and $\psi_K/n_S\lesssim s\psi_K/N$ under that condition and $n_S\geq n_T$. Then
squaring \eqref{rate_fused} gives $s\psi_K/N+h^2/K$. The bound for
TMTL(Debiased) does not match throughout this regime. It does match at the lower
end: at $h=0$ both reduce to the pooled rate, and they continue to agree for
small positive $h$, as long as $s\psi_K/N$ dominates the debiased heterogeneity
term $(h/\sqrt K)\sqrt{\psi_K/n_T}$. Away from the pooled-rate region, the
debiased bound can be larger in order than the lower bound; it therefore does
not establish optimality uniformly throughout the small-radius regime.

At large radii neither available upper bound establishes the target-only rate
$s\psi_K/n_T$: the bound for TMTL(Fused) contains $h^2/K$ and that for
TMTL(Debiased) a term linear in $h$, and both grow without limit as the sources
move away from the target. 

\paragraph{Closing the large-radius gap.}
\label{sec:projection}
A projection closes that gap, by supplying the missing target-only branch. Its
role here is to say what is achievable, not to propose a third procedure. The
calibration below uses the unknown sparsity level $s$ and the constant
$\alpha_0$ in the restricted lower curvature bound for the target design (see
Lemma~\ref{lemma_rscrsm}), making this an oracle procedure. The projected
estimators are theoretical constructions introduced to establish what is
achievable at large radii. They are not implemented or evaluated in
Sections~\ref{sec:simulation_p200} and~\ref{sec:realdata}, which report
TMTL(Fused) and TMTL(Debiased) only.

The construction uses the target sample to decide which coefficient matrices
are compatible with its noise level and sparsity. For $\lambda_T>0$, let the
target-only group lasso pilot be
\begin{equation}
\widehat{\mathbf B}^{\mathrm{tar}}\in\argminA_{\mathbf B\in\mathbb R^{p\times K}}\left\{\frac{1}{2n_TK}\lVert\mathbf Y^0-\mathbf X^{0}\mathbf B\rVert_F^2+\lambda_T\lVert\mathbf B\rVert_{2,1}\right\},
\label{eq:safe-target-pilot}
\end{equation}
Given $\kappa_T\geq0$, the mixed-norm target confidence set is
\begin{equation}
\mathcal C_T(\lambda_T,\kappa_T):=
\left\{\mathbf B\in\mathbb R^{p\times K}:
\begin{aligned}
&\left\lVert\frac{1}{n_TK}(\mathbf X^{0})^\top(\mathbf X^{0}\mathbf B-\mathbf Y^0)\right\rVert_{2,\infty}\leq\lambda_T,\\
&\lVert\mathbf B\rVert_{2,1}\leq\lVert\widehat{\mathbf B}^{\mathrm{tar}}\rVert_{2,1}+\kappa_T
\end{aligned}
\right\}.
\label{eq:safe-confidence-set}
\end{equation}
Row $j$ of the matrix in the first restriction collects the empirical
correlations between predictor $j$ and the target residuals across tasks, so
bounding its largest row norm is a group Dantzig constraint keeping $\mathbf B$
compatible with the target sample. The second restriction bounds the total row
norm, using the pilot as a reference, because small residual correlations alone
do not pin down a coefficient matrix in high dimensions.
For $g\in\{F,D\}$, target-certified TMTL projects the reported estimator onto
this set:
\begin{equation}
\widehat{\mathbf W}^{g,\mathrm{safe}}:=\mathop{\arg\min}_{\mathbf B\in\mathcal C_T(\lambda_T,\kappa_T)}\frac12\lVert\mathbf B-\widehat{\mathbf W}^g\rVert_F^2.
\label{eq:safe-projection}
\end{equation}
The set in \eqref{eq:safe-confidence-set} is nonempty, closed and convex, so
the projection exists and is unique, and a candidate already in the set is
unchanged.

Two properties of this set drive the result, and both hold on one
high-probability event. First, the set contains the truth $\mathbf B^0$: the
calibration of $\kappa_T$ in \eqref{eq:safe-tuning} is chosen to be large
enough to cover the pilot's error in total row norm. Because $\mathbf B^0$ is
feasible, projecting onto a convex set cannot move a candidate further from it,
so the candidate's transfer guarantee survives projection. Second, on the same
event every matrix in the set is within the target-only rate of
$\mathbf B^0$; the two restrictions in \eqref{eq:safe-confidence-set}, combined
with restricted curvature of the target design, bound the radius of the whole
set, not merely the distance to one matrix fixed in advance. The first property
is what preserves the transfer branch, and the second is what supplies the cap
that takes over when the transfer branch is large. Constraining the fitted
target block $\widehat{\mathbf B}^0$ inside \eqref{cotrain} would not serve the
same purpose, because neither reported estimator need lie in a set containing
it. Section~\ref{app:safe-tmtl} of the Appendices explains this and proves the
two properties.

\begin{theorem}[Target-certified projection]
\emergencystretch=1.5em
\label{thm:safe-projection}
Assume Structural Assumption~\ref{assumption_joint_sparsity} and
Assumptions~\ref{assumption_design}--\ref{assumption_error}. Fix $c>2$, choose
\begin{equation}
\lambda_T=c_T\sqrt{\frac{\psi_K}{Kn_T}},\qquad \kappa_T=\frac{24Ks\lambda_T}{\alpha_0},
\label{eq:safe-tuning}
\end{equation}
where $c_T$ is sufficiently large and $\alpha_0$ is the target curvature
constant of Lemma~\ref{lemma_rscrsm}, and suppose that
$s\log p/n_T\leq c_{\mathrm{sp}}$ for a sufficiently small constant
$c_{\mathrm{sp}}>0$. For any random matrix
$\widehat{\mathbf W}\in\mathbb R^{p\times K}$, let
$\widehat{\mathbf W}^{\mathrm{safe}}$ be its projection onto
$\mathcal C_T(\lambda_T,\kappa_T)$. Then, with probability at least
\begin{equation}
1-C_1Kp\exp(-C_2n_T)-C_3p^{1-c},
\label{eq:safe-probability}
\end{equation}
we have
\begin{equation}
\frac{1}{\sqrt K}\lVert\widehat{\mathbf W}^{\mathrm{safe}}-\mathbf B^0\rVert_F\leq\left\{\frac{1}{\sqrt K}\lVert\widehat{\mathbf W}-\mathbf B^0\rVert_F\right\}\wedge C\sqrt{\frac{s\psi_K}{n_T}}.
\label{eq:safe-oracle-bound}
\end{equation}
The constants depend only on the constants in the design and noise assumptions,
the curvature constants, and the fixed tail constant. In particular, the result
permits arbitrary source--target contrasts and does not require the candidate
to be independent of the target sample.
\end{theorem}

Applying Theorem~\ref{thm:safe-projection} to the two TMTL candidates produces
the following bounds. The transfer branch still requires the conditions and
tuning of the corresponding theorem above; projection supplies the additional
target-only branch. Here $m$ is the problem-sequence index used in
Theorems~\ref{thm:fused} and~\ref{thm:debiased}, so the large-$m$
qualification is inherited from those results.

\begin{prop}[Target-certified TMTL rates]
\label{prop:safe-tmtl-rates}
Suppose the conditions of Theorem~\ref{thm:safe-projection} hold with $c>8$.

If the conditions and tuning of Theorem~\ref{thm:fused} also hold with the same
$c$, then, for all sufficiently large $m$, with
probability at least
\begin{equation}
1-C_1(L+1)Kp\exp(-C_2n_T)-C_3(L+1)p^{1-c},
\label{eq:safe-fused-probability}
\end{equation}
\begin{equation}
\frac{1}{\sqrt K}\lVert\widehat{\mathbf W}^{F,\mathrm{safe}}-\mathbf B^0\rVert_F\lesssim\left\{\sqrt{\frac{s\psi_K}{N}}+\sqrt{\frac{\bar h}{\sqrt K}\sqrt{\frac{\psi_K}{n_S}}}+\frac{\epsilon_D}{\sqrt K}\right\}\wedge\sqrt{\frac{s\psi_K}{n_T}}.
\label{eq:safe-fused-rate}
\end{equation}

If the conditions and tuning of Theorem~\ref{thm:debiased} also hold with the
same $c$, then, for all sufficiently large $m$,
with the probability in \eqref{eq:safe-fused-probability} up to a change in its
constants,
\begin{equation}
\frac{1}{\sqrt K}\lVert\widehat{\mathbf W}^{D,\mathrm{safe}}-\mathbf B^0\rVert_F\lesssim\left\{\sqrt{\frac{s\psi_K}{N}}+\sqrt{\frac{\bar h}{\sqrt K}\sqrt{\frac{\psi_K}{n_T}}}\right\}\wedge\sqrt{\frac{s\psi_K}{n_T}}.
\label{eq:safe-debiased-rate}
\end{equation}
\end{prop}

Each of \eqref{eq:safe-fused-rate} and \eqref{eq:safe-debiased-rate} is a
minimum of two terms: a transfer branch, which retains the behavior of the
underlying candidate, and a target-only branch $\sqrt{s\psi_K/n_T}$, which
holds with high probability under the stated sparsity and curvature
conditions. We refer to the second branch as the target-only cap, since it
bounds the rate however far the sources lie from the target. The cap is a
statement about rates. It does not assert pointwise dominance in realized
prediction loss over every target-only procedure, and because the unprojected
bounds \eqref{rate_fused} and \eqref{rate_debiased} have no such branch, it
describes the projected estimators only.

Replacing the unknown $s$ and $\alpha_0$ in \eqref{eq:safe-tuning} by a
data-driven radius is the natural route to a usable procedure, and
Section~\ref{app:safe-plugin} of the Appendices describes one such plug-in,
following \citet{he2026parameter}. Nothing above covers it: the argument rests
on $\mathbf B^0$ remaining feasible, and no result here shows that an
estimated radius preserves that. We leave this direction
for future research. 

\begin{table}[!t]
\caption{Total squared error $K^{-1}\lVert\widehat{\mathbf B}-\mathbf B^0\rVert_F^2$
on $\Theta_{\mathrm{MT}}(s,h)$, up to constants. The first row gives the lower
bound $r_{\mathrm{pool}}+r_{\mathrm{shift}}$ of
Theorem~\ref{thm:multitask-transfer-lower-bound} in each radius regime, and the
remaining rows record whether the corresponding upper bound matches it.
Entries assume $\psi_K\asymp\omega_{s,K}$ and hold under the conditions of the
theorem concerned, including $n_S\geq n_T$ throughout, the oracle regime choice
of Theorem~\ref{thm:debiased}, and the target sparsity condition of
Theorem~\ref{thm:safe-projection} for the projected estimators. A dash means that the available bounds do not
establish a match throughout that regime.}
\label{tab:rate-comparison}
\centering
\small
\renewcommand{\arraystretch}{1.4}
\setlength{\tabcolsep}{5pt}
\begin{tabular}{lccc}
\toprule
& small $h$ & intermediate $h$ & large $h$\\
& ($h\leq\tau$) & ($\tau\leq h\leq s\tau$) & ($h\geq s\tau$)\\
\midrule
Lower bound
 & $\dfrac{s\omega_{s,K}}{N}+\dfrac{h^2}{K}$
 & $\dfrac{s\omega_{s,K}}{N}+\dfrac{h}{\sqrt K}\sqrt{\dfrac{\omega_{s,K}}{n_T}}$
 & $\dfrac{s\omega_{s,K}}{n_T}$\\
\midrule
TMTL(Fused) & yes, if $L\lesssim s$ & --- & ---\\
TMTL(Debiased) & --- & yes & ---\\
TMTL(Fused), projected & yes, if $L\lesssim s$ & --- & yes\\
TMTL(Debiased), projected & --- & yes & yes\\
\bottomrule
\end{tabular}
\end{table}

We can now complete the comparison with the lower bound. On
$\Theta_{\mathrm{MT}}(s,h)$ we again have $\bar h\leq h$ and
$\epsilon_D\leq h$, so replacing $\bar h$ and $\epsilon_D$ in
\eqref{eq:safe-fused-rate} and \eqref{eq:safe-debiased-rate} by $h$ and
squaring gives the orders
\begin{equation}
R_{F,\mathrm{safe}}(s,h):=\left\{\frac{s\psi_K}{N}+\frac{h}{\sqrt K}\sqrt{\frac{\psi_K}{n_S}}+\frac{h^2}{K}\right\}\wedge\frac{s\psi_K}{n_T},
\label{eq:lower-upper-comparison}
\end{equation}
\begin{equation}
R_{D,\mathrm{safe}}(s,h):=\left\{\frac{s\psi_K}{N}+\frac{h}{\sqrt K}\sqrt{\frac{\psi_K}{n_T}}\right\}\wedge\frac{s\psi_K}{n_T}.
\label{eq:safe-debiased-squared-rate}
\end{equation}
As in Section~\ref{sec:bound-comparison}, we take $\psi_K\asymp\omega_{s,K}$.
At $h=0$ both reduce to $s\psi_K/N$, since $N\geq n_T$.

Table~\ref{tab:rate-comparison} assembles the resulting comparison. Each row
records, for one estimator and one radius regime, whether the available upper
bound matches the lower bound of
Theorem~\ref{thm:multitask-transfer-lower-bound}: the first two rows restate
Section~\ref{sec:bound-comparison} for $\widehat{\mathbf W}^F$ and
$\widehat{\mathbf W}^D$, and the last two record what
\eqref{eq:lower-upper-comparison} and \eqref{eq:safe-debiased-squared-rate}
give for their projections.

Because each projected rate is a minimum, the cap can only lower a bound, so
every match already established survives projection. What it adds differs
between the two estimators, and we take them in turn.

For TMTL(Debiased), the cap is inactive in both of the first two regimes. At
$h\leq s\tau$ the heterogeneity term satisfies
$(h/\sqrt K)\sqrt{\psi_K/n_T}\lesssim s\psi_K/n_T$, and
$s\psi_K/N\leq s\psi_K/n_T$ throughout, so
\eqref{eq:safe-debiased-squared-rate} has the order of its unprojected bound
and the intermediate-radius match carries over unchanged.

For TMTL(Fused), the cap is inactive at small radii, where $h\leq\tau$ gives
$h^2/K\leq\tau^2/K\lesssim s\psi_K/n_T$, so the small-radius match under
$L\lesssim s$ also carries over. Inside the intermediate regime, however, the
cap can bind and strictly improve the bound. Taking $h=s^{3/4}\tau$, which
satisfies $\tau\leq h\leq s\tau$ for $s\geq1$, gives
$h^2/K=\sigma^2s^{3/2}\omega_{s,K}/n_T$, larger than the cap
$s\psi_K/n_T$ by a factor of order $\sqrt s$. Projection removes that excess.
It does not produce a match throughout the regime: for $h$ in its interior the
cap exceeds the lower rate
$s\omega_{s,K}/N+(h/\sqrt K)\sqrt{\omega_{s,K}/n_T}$, so the corresponding
entry of Table~\ref{tab:rate-comparison} remains a dash. But the
projected fused bound is smaller than the unprojected one over part of the
intermediate regime.

At large radii the cap binds for both. There $h\geq s\tau$ forces both
$\sigma hK^{-1/2}(\omega_{s,K}/n_T)^{1/2}$ and $h^2/K$ above
$\sigma^2s\omega_{s,K}/n_T$, so the lower rate is of the target-only order,
whereas the unprojected bounds have no such branch and keep growing with $h$.
Both projected rates are bounded by $s\psi_K/n_T$, so both attain that order
and the large-radius gap closes.

\section{Simulation study} \label{sec:simulation_p200}

We study how the magnitude and alignment of source--target differences affect
the benefits of transfer and target correction. The simulations include source
shifts that cancel and shifts that reinforce one another. We compare
TMTL(Fused) and TMTL(Debiased) with target-only multitask learning, pooled
multitask learning, and separate single-output transfer fits, using
coefficient-estimation error, prediction error, and the frequency of negative
transfer.

\subsection{Simulation design} \label{sec:simulation-design}
The simulation design is similar to the designs of
\citet{li2022transfer, he2024transfusion}.
We generate data from the common-design model in
\eqref{transferMTL_dgp_samecovariate}. The target domain has sample size
$n_T=100$, each source domain has sample size $n_S=500$, the predictor dimension
is $p=200$, and the number of tasks is $K\in\{10,20,30\}$. Thus, the source sample
sizes are balanced across domains in every multisource setting.
The target tasks share a design matrix whose entries are generated independently
from $\mathcal N(0,1)$. The first $s=10$ rows of $\mathbf B^0$ are active, with
entries generated independently from $\mathcal N(0.2,0.2^2)$; the remaining
rows are zero. Source and target errors are generated independently from
$\mathcal N(0,0.25^2)$.

We consider $L=3$ source domains under mild covariate shift. For each
replication and source domain, the entries of the design matrix
$\mathbf{X}^{\ell}$ are generated independently from
$\mathcal N(0,\sigma_x^2)$, where, independently for each entry, $\sigma_x$
equals $0.9$ or $1.1$ with equal probability.

To generate coefficient heterogeneity, for each $\alpha$, define
$\Delta_\alpha\in\mathbb R^{p\times K}$ by drawing entries in row $j\le 15$
independently from
\[
\mathcal N\left(0.2\alpha/\sqrt j,\; (0.2\alpha/\sqrt j)^2\right),
\]
and setting rows $j>15$ to zero. The parameter $\alpha$ controls the magnitude
of the source--target coefficient shift. We generate the source coefficient
matrices as follows:
\begin{itemize}
    \item Balanced shifts:
    $\mathbf B^1=\mathbf B^2=\mathbf B^0+(3/4)\Delta_{1/2}$ and
    $\mathbf B^3=\mathbf B^0-(3/2)\Delta_{1/2}$.
    \item Aligned shifts: $\mathbf B^1=\mathbf B^2=\mathbf B^3=\mathbf B^0+\Delta_\alpha$,
    with $\alpha\in\{1/4,1/2,1\}$.
\end{itemize}
In the balanced-shift scenarios, the source deviations cancel in the weighted
average, so $\epsilon_D=0$ exactly, although the individual contrasts
$h_\ell^\star$, and hence $\bar h$, are large. In the aligned-shift scenarios,
the deviations point in the same direction and do not cancel, so $\epsilon_D$
is comparable to $\bar h$. We generate 100 Monte Carlo replications.
The target and source coefficient matrices remain fixed across replications,
whereas the covariates and responses are regenerated in each replication.

Because $\mathbf B^0$ and $\Delta_\alpha$ are drawn separately for each $K$, the
three values of $K$ are distinct realizations of the same generating mechanism
rather than nested versions of a single instance. We therefore read them as
repeated instances of the same comparison: the ordering of the methods is
assessed within each $K$, and its stability across $K$ indicates that the
conclusions are not an artifact of one draw of the coefficient matrices. The
numerical differences between values of $K$ are not interpreted as a
task-dimension asymptotic.

\subsection{Competing methods and tuning}
\label{sec:simulation-methods}
We compare the proposed estimators with single-output transfer and multitask
baselines:
\begin{itemize}
    \item \textbf{TMTL(Fused)}: the proposed fused multitask transfer estimator $\widehat{\mathbf W}^F$.
    \item \textbf{TMTL(Debiased)}: the proposed target-corrected multitask transfer estimator $\widehat{\mathbf W}^D$.
    \item \textbf{TSTL(Fused)}: the fused high-dimensional transfer estimator of \citet{he2024transfusion}.
    \item \textbf{TSTL(Debiased)}: the debiased high-dimensional transfer estimator of \citet{he2024transfusion}.
    \item \textbf{MTL(Target)}: an $\ell_{2,1}$-regularized multitask estimator
    fitted using only the target data.
    \item \textbf{MTL(Full)}: an $\ell_{2,1}$-regularized multitask estimator
    fitted after naively pooling the source and target observations and ignoring
    domain labels, thereby imposing a common coefficient matrix across domains.
\end{itemize}
We select the tuning parameters, including $\lambda_0,\ldots,\lambda_L$ and
$\widetilde\lambda$, by cross-validation to minimize the validation mean-squared prediction
error. Except for the taskwise correction, the methods use a logarithmically
spaced grid of 10 values, from $10^{-5}$ to $10^{-1}$. Cross-validating
$\lambda_0,\cdots,\lambda_3$ over a four-dimensional grid is computationally prohibitive at $p=200$. For
\textbf{TMTL(Fused)} and \textbf{TSTL(Fused)}, we therefore fix the
theory-guided ratios $a_\ell = 8\sqrt{n_S/N}$, so that only $\lambda_0$ is
cross-validated. For \textbf{TMTL(Debiased)}, $\widetilde \lambda$ is
cross-validated over the same grid. The single-task baselines are fitted separately
for each response and then aggregated across responses.
Section~\ref{app:simulation-solver} of the Appendices summarizes
method-specific tuning details.

\subsection{Evaluation metrics}
\label{sec:simulation-metric}
We report three quantities. Coefficient-estimation error for an estimator
$\widehat{\mathbf{B}}$ of the target coefficient matrix is measured by the
squared Frobenius norm normalized by the number of tasks $K$, on the log scale.
In replication $r \in [R]$, $R = 100$, we record
\begin{equation}
  e^{(r)} = \log\!\left(\|\widehat{\mathbf{B}}^{(r)}-\mathbf{B}^0\|_F^2/K\right),
  \label{eq:coeferr}
\end{equation}
where $\log$ is the natural logarithm, so that $\exp(e^{(r)}/2)$ is the root
mean squared coefficient error across tasks appearing in
\eqref{eq:normalized-error}. Table~\ref{tab:simulationresult_frobnormerror}
reports the mean and standard deviation of the log error $e^{(r)}$ over the 100
Monte Carlo replications, so the logarithm is taken within each replication and
averaged afterwards; smaller values indicate more accurate estimation.

Out-of-sample prediction error is computed on an independent target test design
$\mathbf{X}^{0,\text{test}} \in \mathbb{R}^{n_\text{test} \times p}$ by
\begin{equation}
  \mathrm{RMSE}^{(r)} = \left\{\frac{1}{n_\text{test}K}
  \bigl\|\mathbf{X}^{0,\text{test}}(\widehat{\mathbf{B}}^{(r)}-\mathbf{B}^0)\bigr\|_F^2\right\}^{1/2},
  \label{eq:predrmse}
\end{equation}
which measures error against the true regression function
$\mathbf X^{0,\text{test}}\mathbf B^0$ and therefore excludes the observation
noise. The expected squared prediction error for a new response additionally
includes the noise variance. 

Because all methods are fitted to the same data within a replication, their
coefficient errors can be compared replication by replication. Taking the
target-only estimator as the reference, we record the negative-transfer
frequency
\begin{equation}
  \mathrm{NT} = \frac{1}{R}\sum_{r=1}^{R}
  \mathbf{1}\bigl\{\|\widehat{\mathbf{B}}^{(r)}-\mathbf{B}^0\|_F >
  \|\widehat{\mathbf{B}}^{(r)}_{\mathrm{MTL(Target)}}-\mathbf{B}^0\|_F\bigr\},
  \label{eq:negative_transfer}
\end{equation}
the proportion of replications in which using the source domains is worse than target-only
multitask learning. We report $\mathrm{NT}$ for the five estimators that use source data;
MTL(Target) is the reference and is omitted.
Table~\ref{tab:simulationresult_frobnormerror} summarizes the log
coefficient-estimation errors. Tables~\ref{tab:simulationresult_predictionerror}
and~\ref{tab:simulationresult_negativetransfer} report the prediction errors
and negative-transfer frequencies for the same settings.
\begin{table}[!htbp]
\caption{Log coefficient-estimation errors $e^{(r)}$ in \eqref{eq:coeferr} over 100 Monte
Carlo replications for $K\in\{10,20,30\}$. Entries report the mean (standard deviation); smaller values indicate more accurate
estimation. Boldface marks the smallest displayed mean in each row, including ties.}
\label{tab:simulationresult_frobnormerror}
\centering
\small
\renewcommand{\arraystretch}{1.15}
\setlength{\tabcolsep}{2.4pt}
\resizebox{\linewidth}{!}{%
\begin{tabular}{llcccccc}
\toprule
\multirow{2}{*}{\textbf{K}} &
\multirow{2}{*}{\textbf{Shifts}} &
\multicolumn{2}{c}{\textbf{TMTL}} &
\multicolumn{2}{c}{\textbf{MTL}} &
\multicolumn{2}{c}{\textbf{TSTL}} \\
\cmidrule(lr){3-4}\cmidrule(lr){5-6}\cmidrule(lr){7-8}
& &
\textbf{Fused} & \textbf{Debiased} &
\textbf{Target} & \textbf{Full} &
\textbf{Fused} & \textbf{Debiased} \\
\midrule
\multirow{4}{*}{10}
& Balanced             & -5.84 (0.18) & \textbf{-5.85 (0.18)} & -3.76 (0.19) & -5.83 (0.18) & -5.60 (0.17) & -5.40 (0.23) \\
& Aligned $\alpha=1/4$ & -4.74 (0.05) & \textbf{-5.06 (0.15)} & -3.76 (0.19) & -4.41 (0.04) & -4.46 (0.09) & -4.49 (0.14) \\
& Aligned $\alpha=1/2$ & -3.53 (0.03) & \textbf{-4.32 (0.14)} & -3.76 (0.19) & -2.98 (0.02) & -3.36 (0.05) & -3.88 (0.11) \\
& Aligned $\alpha=1$   & -2.30 (0.04) & \textbf{-3.95 (0.16)} & -3.76 (0.19) & -1.56 (0.04) & -2.06 (0.05) & -3.45 (0.12) \\
\midrule
\multirow{4}{*}{20}
& Balanced             & \textbf{-6.21 (0.18)} & \textbf{-6.21 (0.18)} & -4.08 (0.17) & -5.50 (0.30) & -5.60 (0.15) & -5.38 (0.18) \\
& Aligned $\alpha=1/4$ & -4.60 (0.03) & \textbf{-4.79 (0.17)} & -4.08 (0.17) & -4.35 (0.03) & -4.30 (0.06) & -4.36 (0.09) \\
& Aligned $\alpha=1/2$ & -3.32 (0.03) & \textbf{-4.57 (0.10)} & -4.08 (0.17) & -2.87 (0.02) & -3.18 (0.03) & -3.82 (0.11) \\
& Aligned $\alpha=1$   & -1.96 (0.03) & \textbf{-4.21 (0.12)} & -4.08 (0.17) & -1.43 (0.01) & -2.02 (0.04) & -3.45 (0.10) \\
\midrule
\multirow{4}{*}{30}
& Balanced             & \textbf{-6.27 (0.18)} & \textbf{-6.27 (0.18)} & -4.10 (0.15) & -6.11 (0.20) & -5.56 (0.11) & -5.38 (0.12) \\
& Aligned $\alpha=1/4$ & -4.61 (0.04) & \textbf{-5.11 (0.13)} & -4.10 (0.15) & -4.45 (0.02) & -4.35 (0.05) & -4.37 (0.06) \\
& Aligned $\alpha=1/2$ & -3.24 (0.03) & \textbf{-4.57 (0.10)} & -4.10 (0.15) & -2.96 (0.02) & -3.26 (0.03) & -3.80 (0.07) \\
& Aligned $\alpha=1$   & -2.10 (0.02) & \textbf{-4.36 (0.10)} & -4.10 (0.15) & -1.50 (0.01) & -1.99 (0.03) & -3.39 (0.08) \\
\bottomrule
\end{tabular}}

\end{table}

\subsection{Results} \label{sec:simulation-results}

{\emergencystretch=1.5em Fusion performs well when source shifts cancel, while
target correction becomes important when they align
(Table~\ref{tab:simulationresult_frobnormerror}). In the balanced-shift
settings, \textbf{TMTL(Fused)} and \textbf{TMTL(Debiased)} have nearly identical
mean log coefficient-estimation errors and both improve on target-only multitask
learning. Their mean test RMSE is about one third of that of
\textbf{MTL(Target)} (Table~\ref{tab:simulationresult_predictionerror}). Simple pooling is also competitive here: although individual
sources differ from the target, their shifts cancel in the population pooled
coefficient. \par}

When the shifts align, this cancellation disappears. As their magnitude grows,
the advantage of \textbf{TMTL(Debiased)} over \textbf{TMTL(Fused)} widens,
while simple pooling deteriorates sharply. \textbf{TMTL(Debiased)} has the
smallest mean log coefficient-estimation error in every aligned-shift setting
and a smaller mean log error than \textbf{MTL(Target)} in every setting.
The negative-transfer frequencies in
Table~\ref{tab:simulationresult_negativetransfer} make the benefit of
correction particularly clear. With balanced shifts, no source-based method
performs worse than \textbf{MTL(Target)} in any replication. With aligned
shifts and $\alpha\geq1/2$, however, \textbf{TMTL(Fused)} has larger
coefficient error than \textbf{MTL(Target)} in $86$--$100\%$ of replications.
Correction reduces that frequency to zero at $\alpha=1/2$ and to $3$--$19\%$
at $\alpha=1$: it substantially reduces negative transfer in these settings,
without eliminating it.

Sharing information across responses provides an additional benefit. The TMTL
estimators generally have smaller mean log errors than their taskwise TSTL
counterparts. The exceptions are two aligned-shift settings, $K=20$ with
$\alpha=1$ and $K=30$ with $\alpha=1/2$, where \textbf{TSTL(Fused)} has
slightly smaller mean log error than \textbf{TMTL(Fused)}. Overall, these
comparisons support combining transfer across domains with shared predictor
structure across responses.

\begin{table}[!t]
\caption{Target test RMSE \eqref{eq:predrmse} for the main study.
Entries give the mean (standard deviation) over 100 replications, rounded to
two decimal places. Smaller values indicate better prediction; boldface
marks the smallest mean before rounding, including ties.}
\label{tab:simulationresult_predictionerror}
\centering
\small
\renewcommand{\arraystretch}{1.15}
\setlength{\tabcolsep}{3pt}
\begin{tabular}{llcccccc}
\toprule
\multirow{2}{*}{\textbf{K}} &
\multirow{2}{*}{\textbf{Shifts}} &
\multicolumn{2}{c}{\textbf{TMTL}} &
\multicolumn{2}{c}{\textbf{MTL}} &
\multicolumn{2}{c}{\textbf{TSTL}} \\
\cmidrule(lr){3-4}\cmidrule(lr){5-6}\cmidrule(lr){7-8}
& &
\textbf{Fused} & \textbf{Debiased} &
\textbf{Target} & \textbf{Full} &
\textbf{Fused} & \textbf{Debiased} \\
\midrule
\multirow{4}{*}{10}
& Balanced             & \textbf{0.05 (0.01)} & \textbf{0.05 (0.01)} & 0.15 (0.02) & \textbf{0.05 (0.01)} & 0.06 (0.01) & 0.07 (0.01) \\
& Aligned $\alpha=1/4$ & 0.09 (0.01) & \textbf{0.08 (0.01)} & 0.15 (0.02) & 0.11 (0.01) & 0.11 (0.01) & 0.11 (0.01) \\
& Aligned $\alpha=1/2$ & 0.17 (0.01) & \textbf{0.12 (0.01)} & 0.15 (0.02) & 0.23 (0.01) & 0.19 (0.01) & 0.14 (0.01) \\
& Aligned $\alpha=1$   & 0.32 (0.01) & \textbf{0.14 (0.01)} & 0.15 (0.02) & 0.46 (0.03) & 0.36 (0.02) & 0.18 (0.01) \\
\midrule
\multirow{4}{*}{20}
& Balanced             & \textbf{0.05 (0.00)} & \textbf{0.05 (0.00)} & 0.13 (0.01) & 0.07 (0.01) & 0.06 (0.01) & 0.07 (0.01) \\
& Aligned $\alpha=1/4$ & 0.10 (0.00) & \textbf{0.09 (0.01)} & 0.13 (0.01) & 0.11 (0.01) & 0.12 (0.01) & 0.11 (0.01) \\
& Aligned $\alpha=1/2$ & 0.19 (0.01) & \textbf{0.10 (0.01)} & 0.13 (0.01) & 0.24 (0.01) & 0.20 (0.01) & 0.15 (0.01) \\
& Aligned $\alpha=1$   & 0.38 (0.01) & \textbf{0.12 (0.01)} & 0.13 (0.01) & 0.49 (0.02) & 0.36 (0.02) & 0.18 (0.01) \\
\midrule
\multirow{4}{*}{30}
& Balanced             & \textbf{0.04 (0.00)} & \textbf{0.04 (0.00)} & 0.13 (0.01) & 0.05 (0.01) & 0.06 (0.00) & 0.07 (0.01) \\
& Aligned $\alpha=1/4$ & 0.10 (0.00) & \textbf{0.08 (0.01)} & 0.13 (0.01) & 0.11 (0.00) & 0.11 (0.00) & 0.11 (0.00) \\
& Aligned $\alpha=1/2$ & 0.20 (0.01) & \textbf{0.10 (0.01)} & 0.13 (0.01) & 0.23 (0.01) & 0.20 (0.01) & 0.15 (0.01) \\
& Aligned $\alpha=1$   & 0.35 (0.01) & \textbf{0.11 (0.01)} & 0.13 (0.01) & 0.47 (0.02) & 0.37 (0.01) & 0.18 (0.01) \\
\bottomrule
\end{tabular}

\end{table}
\begin{table}[!t]
\caption{Negative-transfer frequency \eqref{eq:negative_transfer} for the main
study: the proportion of 100 replications in which coefficient error exceeds
that of \textbf{MTL(Target)}. Boldface marks the smallest frequency, except
in rows where all methods tie. The target-only reference is omitted.}
\label{tab:simulationresult_negativetransfer}
\centering
\small
\renewcommand{\arraystretch}{1.15}
\setlength{\tabcolsep}{5pt}
\begin{tabular}{llccccc}
\toprule
\multirow{2}{*}{\textbf{K}} &
\multirow{2}{*}{\textbf{Shifts}} &
\multicolumn{2}{c}{\textbf{TMTL}} &
\multicolumn{1}{c}{\textbf{MTL}} &
\multicolumn{2}{c}{\textbf{TSTL}} \\
\cmidrule(lr){3-4}\cmidrule(lr){5-5}\cmidrule(lr){6-7}
& &
\textbf{Fused} & \textbf{Debiased} &
\textbf{Full} &
\textbf{Fused} & \textbf{Debiased} \\
\midrule
\multirow{4}{*}{10}
& Balanced             & 0.00 & 0.00 & 0.00 & 0.00 & 0.00 \\
& Aligned $\alpha=1/4$ & 0.00 & 0.00 & 0.00 & 0.00 & 0.00 \\
& Aligned $\alpha=1/2$ & 0.86 & \textbf{0.00} & 1.00 & 0.99 & 0.31 \\
& Aligned $\alpha=1$   & 1.00 & \textbf{0.18} & 1.00 & 1.00 & 0.96 \\
\midrule
\multirow{4}{*}{20}
& Balanced             & 0.00 & 0.00 & 0.00 & 0.00 & 0.00 \\
& Aligned $\alpha=1/4$ & \textbf{0.00} & \textbf{0.00} & 0.03 & 0.09 & 0.04 \\
& Aligned $\alpha=1/2$ & 1.00 & \textbf{0.00} & 1.00 & 1.00 & 0.95 \\
& Aligned $\alpha=1$   & 1.00 & \textbf{0.19} & 1.00 & 1.00 & 0.99 \\
\midrule
\multirow{4}{*}{30}
& Balanced             & 0.00 & 0.00 & 0.00 & 0.00 & 0.00 \\
& Aligned $\alpha=1/4$ & \textbf{0.00} & \textbf{0.00} & \textbf{0.00} & 0.03 & 0.02 \\
& Aligned $\alpha=1/2$ & 1.00 & \textbf{0.00} & 1.00 & 1.00 & 0.97 \\
& Aligned $\alpha=1$   & 1.00 & \textbf{0.03} & 1.00 & 1.00 & 1.00 \\
\bottomrule
\end{tabular}

\end{table}

A supplementary simulation study in Section~\ref{app:simulation_p100} of the Appendices varies the number of source domains and uses full-grid penalty parameter tuning. It exhibits broadly similar relative performance patterns for \textbf{TMTL(Fused)} and \textbf{TMTL(Debiased)}. That section also reports the ADMM iteration counts of \textbf{TMTL(Fused)}.

\section{Application to single-cell multi-omics data} \label{sec:realdata}
We study whether borrowing information across related cell types improves
prediction of surface-protein abundance from gene expression. Cell types share
molecular features, but their RNA--protein relationships may differ, making
this a setting in which transfer can help while simple pooling can introduce
bias. We examine whether target correction improves prediction and whether
the gains from transfer are greater for smaller target populations.

\subsection{Dataset and preprocessing} \label{sec:realdata-data}
We analyze the CITE-seq benchmark of \citet{luecken2021sandbox}, comprising
26985 bone marrow T-cells with paired measurements of 13953 genes and 134
surface-protein tags. The cells come from nine donors and four sequencing
sites. We use the fine-grained T-cell subtype annotations as domain labels
and analyze each site separately to account for between-site variation.
Within each site, cells are pooled across donors; this defines domains by
cell type but does not remove possible donor effects.

Detailed preprocessing procedures, including filtering, normalization, and
mean-centering, are given in
Section~\ref{app:realdata-preprocessing} of the Appendices.
Table~\ref{tab:realdata_sites_dims} gives the resulting dimensions.

\begin{table}[!t]
\caption{Dimensions after preprocessing separately within each sequencing
site. Each cell type serves in turn as the target, with the remaining cell
types at that site supplying source data.}
\label{tab:realdata_sites_dims}
\centering
\small
\setlength{\tabcolsep}{8pt}
\begin{tabular}{lcccc}
\toprule
& Site 1 & Site 2 & Site 3 & Site 4 \\
\midrule
Total number of cells   & 6447 & 6036 & 5902 & 6723 \\
Number of cell types    & 12   & 10   & 12   & 14   \\
Number of predictors ($p$)          & 1554 & 641  & 446  & 1269 \\
Number of protein responses ($K$)   & 77   & 51   & 68   & 105  \\
\bottomrule
\end{tabular}

\end{table}

\subsection{Transfer-learning setup} \label{sec:realdata-transfer-setup}
Each cell type at a site serves in turn as the target, with the remaining
cell types at that site supplying source data. Because the target changes from
one analysis to the next, $\ell$ indexes cell types in this section and labels
whichever domain is being reported, rather than following the convention of
Section~\ref{sec:problem-setup-tmtl}, under which $\ell=0$ is always the
target and $\ell\in[L]$ the sources. Within each analysis that convention
still applies, after relabeling the current target as domain $0$. 
This rotating-target scheme is adapted from \citet{li2022transfer}, who treat each tissue in turn as the target and the remaining tissues as sources. 
In our analysis, target populations vary substantially in sample size, allowing us to examine how the benefit of transfer varies with the number of target cells.

For each target, we repeat the analysis over 20 independent random splits.
Within each domain, half of the cells, rounded down, are allocated to training,
one quarter, rounded down, to validation, and the remainder to testing.
The fused estimators use training cells from the target and source domains;
validation and prediction assessment use the target domain.

We evaluate prediction by test mean-squared error relative to an intercept-only
model. For target domain $\ell$, and split $r$, define
\begin{equation}
Q_{\text{method},r}^{(\ell)}
=\frac{\operatorname{MSE}_{\text{method},r}^{(\ell)}}
{\operatorname{MSE}_{\mathrm{null},r}^{(\ell)}}.
\label{eq:realdata-relative-mse}
\end{equation}
Both MSEs average squared errors over all proteins and the same target test
cells. The null model predicts each protein by its target training-sample
mean, so values below one indicate improvement over prediction without RNA features. 
We report the mean and standard deviation of this ratio over the 20 splits for each method and target domain.

\subsection{Competing methods and tuning} \label{sec:realdata-methods}
We compare the six methods introduced in Section~\ref{sec:simulation_p200},
using target-domain prediction error for parameter selection. The grid contains 15
logarithmically spaced values in $[10^{-6},1]$ for the fused penalty
$\lambda_0$, the TMTL correction penalty $\widetilde\lambda$, and the
penalties of \textbf{MTL(Target)} and \textbf{MTL(Full)}.
For \textbf{TMTL(Fused)} and \textbf{TSTL(Fused)}, we set
$\lambda_\ell=a_\ell\lambda_0$ with
$a_\ell=8\sqrt{n_\ell^{\mathrm{train}}/N_{\mathrm{train}}}$,
where $n_\ell^{\mathrm{train}}$ is the number of training cells in domain
$\ell$ and $N_{\mathrm{train}}$ is the total across domains.
This uses the formulation for unequal domain sizes in
Section~\ref{sec:generalization} of the Appendices.
Detailed explanations of the implementation settings are given in Section~\ref{app:realdata-implementation} of the
Appendices.

\subsection{Prediction results} \label{sec:realdata-prediction}

\textbf{TMTL(Debiased)} has the smallest mean relative MSE in 35 of the 48
target domains, including ties. Compared with \textbf{MTL(Target)}, it improves
prediction in 43 domains. The median reduction in mean relative MSE is $0.029$
among the 31 domains with fewer than 400 cells, compared with $0.007$ among the
other 17 domains with larger sample sizes.
Figure~\ref{fig:realdata-transfer-gain} shows the larger predictive gains
achieved for the majority of small target populations, as well as the five
domains in which \textbf{TMTL(Debiased)} performs worse than
\textbf{MTL(Target)}. Thus, borrowing information across cell types is
particularly useful when target data are limited, although it does not improve
prediction in every domain.

\textbf{TMTL(Debiased)} improves on \textbf{TMTL(Fused)} in 42 domains and
matches it in the remaining six.
Figure~\ref{fig:realdata-correction-gain} shows that these improvements occur
across a wide range of target sample sizes, although their magnitude
varies by site.

\textbf{TMTL(Debiased)} also has smaller mean relative MSE than
\textbf{MTL(Full)} in 40 domains and \textbf{TSTL(Debiased)} in 42 domains.
These advantages do not hold in every domain, but the overall results support
accounting for differences between cell types, while sharing information across
protein responses. Table~\ref{tab:realdatarmse} of the Appendices reports the
complete results.

\begin{figure}[!htbp]
\centering
\begin{subfigure}{\textwidth}
    \centering
    \includegraphics[width=.93\linewidth]{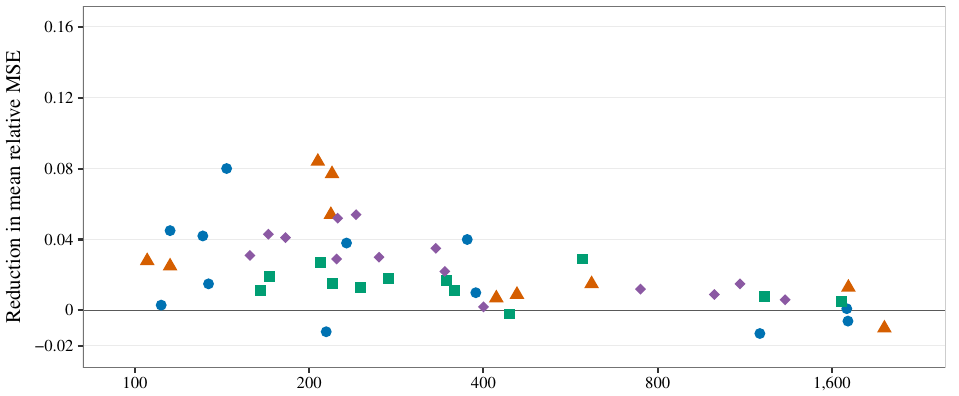}
    \caption{Benefit of transfer: comparison with \textbf{MTL(Target)}.}
    \label{fig:realdata-transfer-gain}
\end{subfigure}

\medskip

\begin{subfigure}{\textwidth}
    \centering
    \includegraphics[width=.93\linewidth]{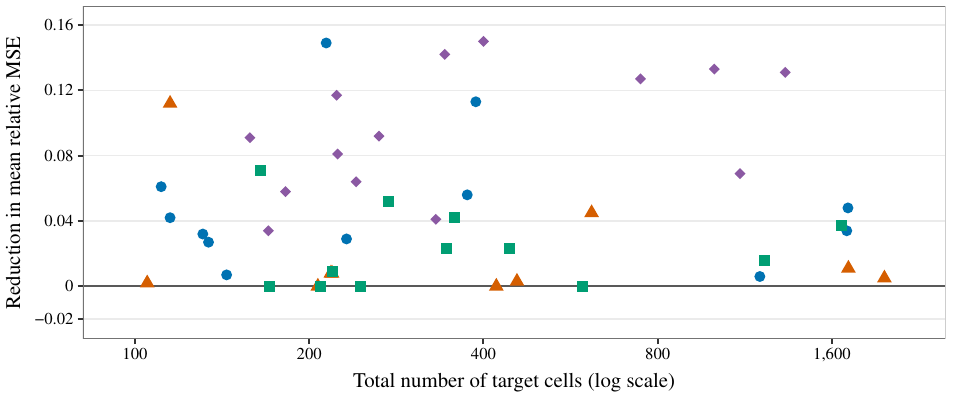}
    \caption{Benefit of target correction: comparison with \textbf{TMTL(Fused)}.}
    \label{fig:realdata-correction-gain}
\end{subfigure}

\medskip
{\small
\textcolor[HTML]{0072B2}{$\bullet$}\ Site 1\qquad
\textcolor[HTML]{D55E00}{$\blacktriangle$}\ Site 2\qquad
\textcolor[HTML]{009E73}{$\blacksquare$}\ Site 3\qquad
\textcolor[HTML]{8B59A3}{$\blacklozenge$}\ Site 4
\par}

\caption{Prediction gains from transfer and target correction across the 48
target domains. Each point represents one cell type at one site. The vertical
axis is the mean relative MSE of the comparison method minus that of
\textbf{TMTL(Debiased)}, so positive values favor \textbf{TMTL(Debiased)}.
The horizontal axis gives the total number of target cells.}
\label{fig:realdata-benefits}
\end{figure}

\clearpage
\section{Discussion} \label{sec:discussion}
This paper develops a transfer-learning framework for multitask regression
with multiple source domains and a data-limited target domain. It combines
shared predictor structure across responses with similarity between source
and target regression matrices. The two estimators offer complementary ways
to use this information: fusion borrows strength across domains, while target
correction adjusts for systematic source-induced shifts. A central message is
that the usefulness of correction depends on how these shifts combine.
Fusion performs well when they are small or tend to cancel, whereas correction
becomes more valuable when they align.

The theoretical and empirical results support this perspective. The error
bounds identify regimes in which transfer across domains and sharing across
responses jointly improve estimation. In the simulations, correction adds
little when source shifts cancel but improves accuracy and reduces negative
transfer when the shifts align. The single-cell analysis illustrates the value
of accounting for cell-type heterogeneity, with the clearest prediction gains
for smaller cell types. These benefits are not uniform: the simulations also
include settings in which using source data is worse than relying on the
target sample alone.

Several directions merit further study. Adaptive tuning procedures with
provable guarantees would connect the theoretical penalty choices to practical
parameter selection. Valid confidence intervals and hypothesis tests would
extend the framework beyond point estimation, while support-recovery guarantees
would clarify when relevant predictors can be identified. The present theory
already allows domain-specific, well-conditioned sub-Gaussian designs;
extensions to heavy-tailed covariates and more severe differences between
designs would broaden its applicability.

\section*{Appendices}
\input{appendix_overview}

\begin{appendix}
\makeatletter
\@addtoreset{theorem}{section}
\@addtoreset{assu}{section}
\@addtoreset{step}{section}
\let\c@lemma\c@theorem
\let\c@prop\c@theorem
\let\c@exam\c@theorem
\let\c@remark\c@theorem
\makeatother
\renewcommand{\thetheorem}{\thesection.\arabic{theorem}}
\renewcommand{\thelemma}{\thetheorem}
\renewcommand{\theprop}{\thetheorem}
\renewcommand{\theexam}{\thetheorem}
\renewcommand{\theremark}{\thetheorem}
\renewcommand{\theassu}{\thesection.\arabic{assu}}
\renewcommand{\thestep}{\thesection.\arabic{step}}
\providecommand{\theHtheorem}{}
\renewcommand{\theHtheorem}{appendix.\thesection.\arabic{theorem}}
\providecommand{\theHlemma}{}
\renewcommand{\theHlemma}{appendix.\thesection.\arabic{theorem}}
\providecommand{\theHprop}{}
\renewcommand{\theHprop}{appendix.\thesection.\arabic{theorem}}
\providecommand{\theHexam}{}
\renewcommand{\theHexam}{appendix.\thesection.\arabic{theorem}}
\providecommand{\theHremark}{}
\renewcommand{\theHremark}{appendix.\thesection.\arabic{theorem}}
\providecommand{\theHassu}{}
\renewcommand{\theHassu}{appendix.\thesection.\arabic{assu}}
\providecommand{\theHstep}{}
\renewcommand{\theHstep}{appendix.\thesection.\arabic{step}}

\input{supp_computation}
\input{supp_generalization}
\input{supp_proofs_upper}

\input{appendix_d_lower_bound}

\input{appendix_c_safe_tmtl}
\input{supp_simulation}
\FloatBarrier
\input{supp_realdata}
\FloatBarrier
\end{appendix}

\bibliographystyle{imsart-nameyear}
\bibliography{references}

\end{document}

%% file: appendix_overview.tex
These appendices provide computational details, theoretical developments, and
additional empirical results. Sections~\ref{app:computation}--\ref{sec:generalization}
describe the optimization algorithm and extensions to more general designs.
Sections~\ref{app:theory-proofs}--\ref{app:safe-tmtl} contain proofs and further
theoretical results, including projected estimators and comparisons with the
minimax benchmark. Sections~\ref{app:simulation_p100}--\ref{app:realdata}
provide additional simulation results and details of the single-cell
application. Structural Assumptions~\ref{assumption_joint_sparsity}
and~\ref{assumption_source_target_contrast} are those of the main text.
The appendices are lettered A--G; equation and table numbering continues from the main text, and formal results are numbered by appendix.

%% file: supp_computation.tex
\section{Computational details for TMTL}
\label{app:computation}
This section gives the updates, diagnostics, and costs of the two-block ADMM
algorithm described in Section~\ref{sec:opt-tmtl}. Throughout, the notation is
that of \eqref{eq:admm-main-constrained}: the first block is
$(\mathbf B^0,\ldots,\mathbf B^L)$ and the second is all auxiliary matrices
$\boldsymbol\Gamma$, the parallel updates within either block being components
of one exact block minimization rather than separate multiblock steps.

\subsection{Augmented Lagrangian}
\label{app:computation-lagrangian}
The augmented Lagrangian of \eqref{eq:admm-main-constrained} is
\begin{align*}
\begin{split}
&\frac{1}{2NK} \sum\limits_{\ell=0}^{L} \lVert \mathbf{Y}^\ell - \mathbf{X}^{\ell} \mathbf{B}^\ell \rVert_F^2 + \lambda_0 \lVert \boldsymbol{\Gamma}^{00} \rVert_{2,1} + \sum\limits_{\ell=1}^{L} \lambda_\ell \lVert \boldsymbol{\Gamma}^\ell - \boldsymbol{\Gamma}^{0\ell} \rVert_{2,1} \\
&\quad + \sum\limits_{\ell=0}^{L} \frac{\rho_{0 \ell}}{2}  \lVert  \mathbf{B}^0 - \boldsymbol{\Gamma}^{0\ell} + \mathbf{U}^\ell \rVert_F^2 - \sum\limits_{\ell=0}^{L} \frac{\rho_{0\ell}}{2} \lVert \mathbf{U}^\ell \rVert_F^2 \\
&\quad + \sum\limits_{\ell=1}^{L} \frac{\rho_\ell}{2} \lVert  \mathbf{B}^\ell - \boldsymbol{\Gamma}^\ell + \mathbf{V}^\ell \rVert_F^2 - \sum\limits_{\ell=1}^{L} \frac{\rho_\ell}{2} \lVert \mathbf{V}^\ell \rVert_F^2,
\end{split}
\end{align*}
where $\rho_{00},\rho_{01},\ldots,\rho_{0L},\rho_1,\ldots,\rho_L$ are positive
penalty parameters and $\mathbf U^\ell,\mathbf V^\ell$ are scaled dual
variables. We initialize $\rho_{0\ell}^{(0)}=1$ for $0\leq\ell\leq L$,
$\rho_\ell^{(0)}=1$ for $\ell\in[L]$, and all auxiliary matrices and scaled
dual variables at zero. Each iteration updates the $\mathbf B$ block, then the
$\boldsymbol\Gamma$ block, and finally the scaled dual variables.

\subsection{Block updates}
\label{app:computation-updates}
The $\mathbf B$ updates of Section~\ref{sec:opt-tmtl} admit the closed forms
\begin{align*}
    \mathbf{B}^0 &\leftarrow \left( (NK)^{-1} (\mathbf{X}^{0})^\top \mathbf{X}^{0} + \sum\limits_{\ell=0}^{L} \rho_{0\ell} \mathbf{I}_p\right)^{-1} \left( (NK)^{-1} (\mathbf{X}^{0})^\top \mathbf{Y}^0 + \sum\limits_{\ell=0}^{L} \rho_{0\ell}(\boldsymbol{\Gamma}^{0\ell} - \mathbf{U}^\ell) \right) \\
    \mathbf{B}^\ell &\leftarrow \left( (NK)^{-1} (\mathbf{X}^{\ell})^\top \mathbf{X}^{\ell} + \rho_\ell \mathbf{I}_p\right)^{-1} \left( (NK)^{-1} (\mathbf{X}^{\ell})^\top \mathbf{Y}^\ell + \rho_\ell (\boldsymbol{\Gamma}^\ell - \mathbf{V}^\ell) \right), \quad \ell \in [L].
\end{align*}
The auxiliary-variable updates exploit the block separability of the
$\ell_{2,1}$ norm, with one block for each feature row. They use the group
soft-thresholding operator $S_a$ of \eqref{eq:group-soft-threshold}; see
Section~6.5.1 of \citet{parikh2014proximal}. For each $j\in[p]$, the
target-copy update is
$\boldsymbol{\Gamma}_j^{00}\leftarrow S_{\lambda_0/\rho_{00}}(\mathbf{B}_j^0+\mathbf{U}_j^0)$.
For each source $\ell\in[L]$, set
\[
\mathbf q_j^{0\ell}:=\mathbf B_j^0+\mathbf U_j^\ell,
\qquad
\mathbf q_j^\ell:=\mathbf B_j^\ell+\mathbf V_j^\ell,
\qquad
\mu_\ell:=\frac{\rho_{0\ell}\rho_\ell}{\rho_{0\ell}+\rho_\ell},
\]
so that the row-wise joint subproblem is
\begin{equation}
\min_{\boldsymbol \Gamma_j^\ell,\boldsymbol \Gamma_j^{0\ell}}
\lambda_\ell\lVert\boldsymbol \Gamma_j^\ell-\boldsymbol \Gamma_j^{0\ell}\rVert_2
+
\frac{\rho_{0\ell}}{2}\lVert\boldsymbol \Gamma_j^{0\ell}-\mathbf q_j^{0\ell}\rVert_2^2
+
\frac{\rho_\ell}{2}\lVert\boldsymbol \Gamma_j^\ell- \mathbf q_j^\ell\rVert_2^2.
\label{eq:admm-row-subproblem}
\end{equation}
Reparameterizing \eqref{eq:admm-row-subproblem} by
$\mathbf r_j^\ell=\boldsymbol \Gamma_j^\ell-\boldsymbol\Gamma_j^{0\ell}$ and
minimizing the quadratic part over $\boldsymbol \Gamma_j^{0\ell}$ at fixed
$\mathbf r_j^\ell$ leaves a single proximal problem in the difference, with
solution
\begin{equation}
\mathbf r_j^\ell
=
S_{\lambda_\ell/\mu_\ell}\left(\mathbf q_j^\ell- \mathbf q_j^{0\ell}\right).
\label{eq:admm-row-prox}
\end{equation}
Back-substituting \eqref{eq:admm-row-prox} gives the auxiliary-variable updates
\newpage
\begin{align*}
    \boldsymbol{\Gamma}_j^{0\ell}
    &\leftarrow
    \frac{\rho_{0\ell} \mathbf q_j^{0\ell}+\rho_\ell \mathbf q_j^\ell-\rho_\ell \mathbf r_j^\ell}
    {\rho_{0\ell}+\rho_\ell}, \quad \ell\in[L],\\
    \boldsymbol{\Gamma}_j^{\ell}
    &\leftarrow
    \frac{\rho_{0\ell} \mathbf q_j^{0\ell}+\rho_\ell \mathbf q_j^\ell+\rho_{0\ell} \mathbf r_j^\ell}
    {\rho_{0\ell}+\rho_\ell}, \quad \ell\in[L].
\end{align*}

\subsection{Residual diagnostics and penalty adaptation}
\label{app:computation-residuals}
Let $t$ index the iterations. After the primal and auxiliary updates at
iteration $t+1$, define the constraint-specific primal residuals and dual
residual contributions by
\begin{equation}
\begin{aligned}
\mathbf R_{\mathrm p,0\ell}^{(t+1)}&:=\mathbf B^{0,(t+1)}-\boldsymbol\Gamma^{0\ell,(t+1)}, & \mathbf R_{\mathrm d,0\ell}^{(t+1)}&:=\rho_{0\ell}^{(t)}\bigl(\boldsymbol\Gamma^{0\ell,(t)}-\boldsymbol\Gamma^{0\ell,(t+1)}\bigr), \quad 0\leq\ell\leq L,\\
\mathbf R_{\mathrm p,\ell}^{(t+1)}&:=\mathbf B^{\ell,(t+1)}-\boldsymbol\Gamma^{\ell,(t+1)}, & \mathbf R_{\mathrm d,\ell}^{(t+1)}&:=\rho_{\ell}^{(t)}\bigl(\boldsymbol\Gamma^{\ell,(t)}-\boldsymbol\Gamma^{\ell,(t+1)}\bigr), \quad \ell\in[L].
\end{aligned}
\label{eq:admm-residuals}
\end{equation}
The primal residual concatenates the $\mathbf R_{\mathrm p}$ terms in
\eqref{eq:admm-residuals}. The corresponding two-block dual residual has target
component $\sum_{\ell=0}^L\mathbf R_{\mathrm d,0\ell}^{(t+1)}$ and source components
$\mathbf R_{\mathrm d,1}^{(t+1)},\ldots,\mathbf R_{\mathrm d,L}^{(t+1)}$. We stop when the maximum
Frobenius norm among the constraint-specific primal residuals and dual residual
contributions in \eqref{eq:admm-residuals} falls below the prescribed tolerance,
or when the iteration cap is reached. To adapt each penalty
separately, we use the corresponding constraint-specific pair in
\eqref{eq:admm-residuals} as a residual diagnostic. For a pair and its penalty,
write $\mathbf R_{\mathrm p}^{(t+1)}$, $\mathbf R_{\mathrm d}^{(t+1)}$, and $\rho^{(t)}$
for the associated quantities, and set
\begin{equation}
\rho^{(t+1)}:=\begin{cases}
\min\{\tau_{\mathrm{incr}}\rho^{(t)},\rho_{\max}\},
& t<T_{\mathrm{adapt}}\ \text{and}\
\lVert\mathbf R_{\mathrm p}^{(t+1)}\rVert_F
>\mu\lVert\mathbf R_{\mathrm d}^{(t+1)}\rVert_F,\\
\rho^{(t)}/\tau_{\mathrm{decr}},
& t<T_{\mathrm{adapt}}\ \text{and}\
\lVert\mathbf R_{\mathrm d}^{(t+1)}\rVert_F
>\mu\lVert\mathbf R_{\mathrm p}^{(t+1)}\rVert_F,\\
\rho^{(t)}, & \text{otherwise},
\end{cases}
\label{eq:admm-penalty-update}
\end{equation}
where $T_{\mathrm{adapt}}\in\mathbb N$ is a fixed adaptation horizon,
$\mu>1$, $\tau_{\mathrm{incr}},\tau_{\mathrm{decr}}>1$, and
$1\leq\rho_{\max}<\infty$. After the additive scaled-dual update, we
multiply the corresponding $\mathbf U^\ell$ or $\mathbf V^\ell$ by
$\rho^{(t)}/\rho^{(t+1)}$ whenever the penalty changes. This rescaling preserves
the associated unscaled dual variable; see Sections~3.3.1 and~3.4.1 of
\citet{boyd2011distributed}.

\subsection{Convergence and computational cost}
\label{app:computation-convergence}
The convergence claim in Section~\ref{sec:opt-tmtl} follows from standard
theory. After iteration $T_{\mathrm{adapt}}$ the penalties in
\eqref{eq:admm-penalty-update} are fixed and remain positive, with lower bound
$\tau_{\mathrm{decr}}^{-T_{\mathrm{adapt}}}$. Rescaling each equality
constraint by the square root of its fixed penalty then gives the standard
two-block formulation, initialized at the end of the adaptive phase, to which
\citet[Section~3.2]{boyd2011distributed} applies.

For the dense implementation, forming the Gram and design--response
cross-product matrices costs $O(Np^2+NpK)$ operations, followed by
$O((L+1)p^3)$ operations to eigendecompose the Gram matrices. The cached
quantities and ADMM variables require $O((L+1)p^2+LpK)$ memory, excluding the
data. Each subsequent iteration costs $O((L+1)p^2K)$ for the coefficient
updates and $O(LpK)$ for the auxiliary and dual updates.

%% file: supp_generalization.tex
\section{Task-specific designs and unequal domain sizes}
\label{sec:generalization}
The main text presents TMTL under the common-design assumption
$\mathbf X^{\ell 1}=\cdots=\mathbf X^{\ell K}$ within each domain and develops
the theory for sources of equal size. Neither restriction is needed to define
the estimator. Let $n_0=n_T$, allow $n_1,\ldots,n_L$ to differ, and set
\[
N:=\sum_{\ell=0}^L n_\ell,
\qquad
\omega_\ell:=\frac{n_\ell}{N},\quad 0\leq\ell\leq L.
\]
Thus, $\sum_{\ell=0}^L\omega_\ell=1$.
All tasks and domains use the same $p$ predictors, indexed consistently,
although their design matrices may differ. Throughout this section
$\mathbf X^{\ell k}\in\mathbb R^{n_\ell\times p}$ denotes the design matrix
for task $k\in[K]$ in domain $\ell$, and
$\mathbf y^{\ell k},\boldsymbol\epsilon^{\ell k}\in\mathbb R^{n_\ell}$ the
corresponding response and error vectors. The common-design model of the main
text is the special case $\mathbf X^{\ell1}=\cdots=\mathbf X^{\ell K}$.
For each domain $0\leq \ell \leq L$, define the stacked response, error, coefficient vector, and block-diagonal design by
\[
\mathbf y^\ell
=
\begin{pmatrix}
\mathbf y^{\ell 1}\\
\vdots\\
\mathbf y^{\ell K}
\end{pmatrix},
\qquad
\boldsymbol\epsilon^\ell
=
\begin{pmatrix}
\boldsymbol\epsilon^{\ell 1}\\
\vdots\\
\boldsymbol\epsilon^{\ell K}
\end{pmatrix},
\qquad
\boldsymbol\beta^\ell
=
\begin{pmatrix}
\boldsymbol\beta^{\ell 1}\\
\vdots\\
\boldsymbol\beta^{\ell K}
\end{pmatrix},
\qquad
\mathcal X^\ell
=
\operatorname{diag}(\mathbf X^{\ell 1},\ldots,\mathbf X^{\ell K}).
\]
Then $\mathbf y^\ell,\boldsymbol\epsilon^\ell\in\mathbb R^{n_\ell K}$,
$\boldsymbol\beta^\ell\in\mathbb R^{pK}$, and
$\mathcal X^\ell\in\mathbb R^{n_\ell K\times pK}$. The $\ell$-th domain satisfies the vectorized linear model
\begin{align*}
\mathbf{y}^{\ell} &= \mathcal X^{\ell} \boldsymbol{\beta}^{\ell} + \boldsymbol{\epsilon}^{\ell}, \quad 0 \leq \ell \leq L.
\end{align*}
Under the common-design assumption, $\mathbf X^{\ell k}=\mathbf X^\ell$
for all $k\in[K]$, so $\mathcal X^\ell=\mathbf I_K\otimes\mathbf X^\ell$.
Thus, $\mathbf X^\ell$ denotes the ordinary design matrix, while
$\mathcal X^\ell$ denotes the block-diagonal design in the vectorized model.
The row-wise grouping is inherited from the matrix formulation. For
$\boldsymbol\beta^\ell\in\mathbb R^{pK}$, let $\beta_{jk}^\ell$ denote the
$(p(k-1)+j)$-th coordinate, corresponding to feature $j$ and task $k$, and set
\[
\boldsymbol\beta_j^\ell
=
(\beta_{j1}^\ell,\ldots,\beta_{jK}^\ell)^\top\in\mathbb R^K.
\]
The group penalty is
\[
\lVert\boldsymbol\beta^\ell\rVert_{2,1}
:=\sum_{j=1}^p\lVert\boldsymbol\beta_j^\ell\rVert_2,
\]
which agrees with the row-sparsity norm in the matrix formulation.
The fused estimator solves
\begin{align*}
\widehat{\boldsymbol\beta}^0,\ldots,\widehat{\boldsymbol\beta}^L
&\in
\argminA_{\boldsymbol\beta^0,\ldots,\boldsymbol\beta^L}
\left[
\sum_{\ell=0}^L
\frac{\omega_\ell}{2n_\ell K}
\lVert
\mathbf y^\ell-\mathcal X^\ell\boldsymbol\beta^\ell
\rVert_2^2
\right.\\
&\qquad\qquad\left.
+
\lambda_0
\left\{
\lVert\boldsymbol\beta^0\rVert_{2,1}
+
\sum_{\ell=1}^L
a_\ell
\lVert\boldsymbol\beta^\ell-\boldsymbol\beta^0\rVert_{2,1}
\right\}
\right],
\\
\widehat w^F
&:=
\sum_{\ell=0}^L\omega_\ell\widehat{\boldsymbol\beta}^\ell.
\end{align*}
Because $\omega_\ell/n_\ell=1/N$, the loss gives equal weight to every
observation while allowing the domain sizes to differ.
For unequal source sizes, the balanced-source weight choices
$a_\ell=8\sqrt{n_S/N}$ and $a_\ell=8n_S/N$ suggest the empirical candidates
$a_\ell=8\sqrt{\omega_\ell}$ and $a_\ell=8\omega_\ell$, respectively.
The debiased estimator then solves the target-correction problem
\begin{align*}
\widehat d
&\in
\argminA_{d\in\mathbb R^{pK}}
\frac{1}{2n_0K}
\lVert
\mathbf y^0-\mathcal X^0(\widehat w^F+d)
\rVert_2^2
+
\widetilde\lambda\lVert d\rVert_{2,1},
\\
\widehat w^D
&:=
\widehat w^F+\widehat d.
\end{align*}
Setting $n_1=\cdots=n_L=n_S$ and
$\mathbf X^{\ell1}=\cdots=\mathbf X^{\ell K}$ recovers the estimator in the
main text. The error bounds proved in Section~\ref{app:theory-proofs} retain
equal source sample sizes; no error bounds for unequal source sizes are
claimed here.

%% file: supp_proofs_upper.tex
\section{Proofs of Theorems~\ref{thm:fused} and~\ref{thm:debiased}}
\label{app:theory-proofs}
We use the vectorized notation of Section~\ref{sec:generalization}. Stacking
the columns of a coefficient matrix identifies the two, so that
$\beta^\ell=\operatorname{vec}(\mathbf B^\ell)$ and
$w^F=\operatorname{vec}(\mathbf W^F)\in\mathbb R^{pK}$, with the fitted
quantities $\widehat\beta^\ell=\operatorname{vec}(\widehat{\mathbf B}^\ell)$,
$\widehat w^F=\operatorname{vec}(\widehat{\mathbf W}^F)$ and
$\widehat d=\operatorname{vec}(\widehat{\mathbf D})$ identified in the same
way; $\mathcal X^\ell$ and $y^\ell$ denote the correspondingly stacked design
and response. Under this identification $\lVert\cdot\rVert_{2,2}$ is the
Frobenius norm and the mixed row norms agree with those of the main text, so
boldface is dropped without ambiguity. The fused
proof first bounds estimation error around the population coefficient $w^F$
and then adds its discrepancy from the target coefficient $\beta^0$. The
correction proof shows how target data can remove that population discrepancy,
while accounting for the error inherited from estimating $w^F$. For both
arguments, score bounds control the empirical correlation between the noise and the design and curvature bounds relate
prediction error to coefficient error. Once these bounds hold, the remaining
steps are deterministic.

\subsection{Assumptions and notation}
\label{app:proof-preliminaries}
Throughout, $n_0=n_T$ and $n_1=\cdots=n_L=n_S$. For task-specific designs,
we use the following extension of Assumption~\ref{assumption_design}.
\begin{assu}[Task-specific sub-Gaussian design]
\label{assumption_task_specific_design}
For each $0\leq \ell\leq L$ and $k\in[K]$, let $\mathbf x_i^{\ell k}$
denote the transpose of the $i$th row of $\mathbf X^{\ell k}$.
The vectors $\mathbf x_1^{\ell k},\ldots,\mathbf x_{n_\ell}^{\ell k}$ are
independent, mean-zero, and sub-Gaussian with covariance
$\boldsymbol\Sigma^{\ell k}$. There exists $\kappa_X<\infty$ such that
\[
\lVert u^\top\mathbf x_i^{\ell k}\rVert_{\psi_2}
\leq
\kappa_X(u^\top\boldsymbol\Sigma^{\ell k}u)^{1/2}
\quad\text{for all }u\in\mathbb R^p.
\]
The eigenvalues of $\boldsymbol\Sigma^{\ell k}$ are uniformly bounded above
and away from zero.
\end{assu}
We impose Assumption~\ref{assumption_error} with the errors jointly
independent of the full collection
$\{\mathbf X^{\ell k}:0\leq\ell\leq L,\quad k\in[K]\}$.
No independence across task-specific design matrices is required. When
$\mathbf X^{\ell1}=\cdots=\mathbf X^{\ell K}$ within each domain, these
conditions reduce to the design and error assumptions in the main paper.

For $v\in\mathbb R^{pK}$, let $v_j\in\mathbb R^K$ collect the coefficients
of feature $j$ across tasks, and define
\newpage
\[
\lVert v\rVert_{2,1}=\sum_{j=1}^p\lVert v_j\rVert_2,
\qquad
\lVert v\rVert_{2,2}=\left(\sum_{j=1}^p\lVert v_j\rVert_2^2\right)^{1/2}.
\]
We also use $\lVert v\rVert_{2,\infty}:=\max_{j\in[p]}\lVert v_j\rVert_2$.
The $\ell_{2,2}$ norm is the Frobenius norm of the corresponding $p\times K$
coefficient matrix. Write
$\mathcal S(\beta^0):=\{j\in[p]:\beta_j^0\neq0\}$ for the target support.
The heterogeneity quantities and population fused coefficient are
\[
h_\ell^\star:=\lVert\beta^\ell-\beta^0\rVert_{2,1},
\qquad
\bar h:=\frac{n_S}{N}\sum_{\ell=1}^Lh_\ell^\star,
\qquad
w^F:=\frac{n_S}{N}\sum_{\ell=1}^L\beta^\ell+\frac{n_T}{N}\beta^0,
\]
so
\[
\epsilon_D:=\lVert w^F-\beta^0\rVert_{2,1}
=\left\lVert\frac{n_S}{N}\sum_{\ell=1}^L(\beta^\ell-\beta^0)\right\rVert_{2,1}
\leq\bar h.
\]
These definitions agree with Section~\ref{sec:problem-setup-tmtl} under
vectorization.

The fusion penalty acts on differences between source and target coefficients.
It is therefore convenient to use the target coefficient and these differences
as coordinates. Set
\[
\theta^\ell=\beta^\ell-\beta^0\quad(\ell\in[L]),
\qquad \theta^0=\beta^0,
\qquad
\theta=((\theta^1)^\top,\ldots,(\theta^L)^\top,(\theta^0)^\top)^\top.
\]
Here $\theta^0$ is the target coefficient and each $\theta^\ell$ is a source
contrast, so the coefficient in source domain $\ell$ is
$\theta^0+\theta^\ell$. In these coordinates, the penalty separates over the
target and contrast blocks. The loss and penalized estimator are
\[
\mathcal L(\theta)=\frac{1}{2NK}\left\{
\sum_{\ell=1}^L\lVert y^\ell-\mathcal X^\ell(\theta^\ell+\theta^0)\rVert_2^2
+\lVert y^0-\mathcal X^0\theta^0\rVert_2^2
\right\},
\]
\[
\widehat\theta\in\argminA_{\theta}
\left\{\mathcal L(\theta)+\lambda_0\lVert\theta^0\rVert_{2,1}
+\sum_{\ell=1}^L\lambda_\ell\lVert\theta^\ell\rVert_{2,1}\right\}.
\]
Throughout, script letters denote losses and the plain letter $L$ denotes the
number of source domains.
Define the estimation errors by
\[
\widehat\Delta^\ell:=\widehat\theta^\ell-\theta^\ell\quad(0\leq\ell\leq L),
\qquad
\widehat\Delta:=((\widehat\Delta^1)^\top,\ldots,
(\widehat\Delta^L)^\top,(\widehat\Delta^0)^\top)^\top.
\]
For a source domain, $\widehat\Delta^\ell$ is the error in its contrast;
the error in its fitted coefficient is
$\widehat\Delta^\ell+\widehat\Delta^0$. Averaging the fitted domain
coefficients gives the fused estimation error
\[
\widehat\Delta^F:=\frac{n_S}{N}\sum_{\ell=1}^L\widehat\Delta^\ell
+\widehat\Delta^0=\widehat w^F-w^F.
\]
Thus $\widehat\Delta^F$ measures estimation around $w^F$. The error of
$\widehat w^F$ relative to the target also includes the population discrepancy
$w^F-\beta^0$, whose mixed norm is $\epsilon_D$.

For $0\leq\ell\leq L$, let
$\widehat\Sigma^\ell:=(\mathcal X^\ell)^\top\mathcal X^\ell/(n_\ell K)
\in\mathbb R^{pK\times pK}$, the empirical Gram matrix of the stacked design.
It is not the $p\times p$ population covariance
$\boldsymbol\Sigma^{\ell k}$ of Assumption~\ref{assumption_task_specific_design};
under a common design within domain $\ell$ the two are related by
\newpage
\[
\widehat\Sigma^\ell
=\frac{1}{K}\,\mathbf I_K\otimes
\frac{(\mathbf X^\ell)^\top\mathbf X^\ell}{n_\ell}.
\]
For $\Delta^0,\ldots,\Delta^L\in\mathbb R^{pK}$, write
\[
\Delta=((\Delta^1)^\top,\ldots,(\Delta^L)^\top,(\Delta^0)^\top)^\top.
\]
The Hessian $\widehat\Sigma$ of $\mathcal L$ satisfies
\begin{align*}
\Delta^\top\widehat\Sigma\Delta
&=\frac{1}{NK}\left\{
\sum_{\ell=1}^L\lVert\mathcal X^\ell(\Delta^\ell+\Delta^0)\rVert_2^2
+\lVert\mathcal X^0\Delta^0\rVert_2^2\right\}\\
&=\sum_{\ell=1}^L\frac{n_S}{N}
(\Delta^\ell+\Delta^0)^\top\widehat\Sigma^\ell(\Delta^\ell+\Delta^0)
+\frac{n_T}{N}(\Delta^0)^\top\widehat\Sigma^0\Delta^0.
\end{align*}
The penalty contribution from source heterogeneity is summarized by
\[
H_\lambda:=\sum_{\ell=1}^L\lambda_\ell h_\ell^\star.
\]
It vanishes when all source coefficients equal the target coefficient.
Unlike $\epsilon_D$, it sums the sizes of the source contrasts before any
directional cancellation can occur.
Constants in concentration and curvature bounds depend only on the fixed
design and noise constants; $c_0$ may also depend on the tail constant $c$.
Constants in rate bounds may additionally depend on fixed tuning constants
and the envelopes $M_F$ and $M_T$ in Section~\ref{sec:theory}.
Sample-size thresholds may depend on how quickly the stated $o(1)$ terms vanish.

\subsection{High-probability ingredients}
\label{app:high-prob-events}
Lemma~\ref{lemma_score_bounds} bounds the noise terms in the gradients by
the corresponding penalty levels. This lets the penalties control the random
linear terms when the loss is expanded around the true coefficients.
Lemma~\ref{lemma_rscrsm} then relates empirical quadratic loss to coefficient
error, with a tolerance involving the mixed norm. Sparsity and the cone
inequalities below will control this tolerance. All subsequent arguments are
deterministic on the intersection of these events.
\begin{lemma}[Score bounds]
\label{lemma_score_bounds}
Under Assumptions~\ref{assumption_task_specific_design} and~\ref{assumption_error}, suppose that $n_S\ge n_T$. Fix $c>0$. There are constants $c_0,C,c_1>0$, with $c_0$ also depending on $c$, such that, if
\[
\lambda_0\ge c_0N^{-1}\sqrt{\frac{(L+1)n_S}{K}\psi_K},
\qquad
\lambda_\ell\ge c_0N^{-1}\sqrt{\frac{n_S}{K}\psi_K},\quad \ell\in[L],
\]
then the event
\[
\mathcal E_{\rm score}:=
\left\{
\left\lVert\frac{1}{NK}(\mathcal X^\ell)^\top\epsilon^\ell\right\rVert_{2,\infty}
\leq\frac{\lambda_\ell}{2}
\text{ for all }\ell\in[L],
\quad
\left\lVert\frac{1}{NK}\sum_{\ell=0}^L(\mathcal X^\ell)^\top\epsilon^\ell\right\rVert_{2,\infty}
\leq\frac{\lambda_0}{2}
\right\}
\]
satisfies
\newpage
\[
\mathbb P(\mathcal E_{\rm score})
\geq
1-C(L+1)p^{1-c}
-C(L+1)Kp\exp(-c_1n_T).
\]
If $\widetilde\lambda\ge c_0\sqrt{\psi_K/(K n_T)}$, then the target-score event
\[
\mathcal E_{\rm tscore}:=
\left\{
\left\lVert\frac{1}{n_TK}(\mathcal X^0)^\top\epsilon^0\right\rVert_{2,\infty}
\leq\frac{\widetilde\lambda}{2}
\right\}
\]
also has probability at least
\[
1-Cp^{1-c}-CKp\exp(-c_1n_T).
\]
\end{lemma}
\begin{proof}
For $0\le\ell\le L$, define
\begin{equation}
\mathcal E_{\rm col}^{\ell}:=\left\{\max_{j\in[p],k\in[K]}\frac{1}{n_\ell}\sum_{i=1}^{n_\ell}(X_{ij}^{\ell k})^2\le C_{\rm col}\right\},
\label{eq:score-column-event}
\end{equation}
where $C_{\rm col}$ is sufficiently large.
Assumption~\ref{assumption_task_specific_design} and sub-exponential
concentration give
\begin{equation}
\mathbb P\bigl((\mathcal E_{\rm col}^{\ell})^c\bigr)\le CKp\exp(-c_1n_\ell).
\label{eq:score-column-probability}
\end{equation}
Conditional on all design matrices, each score block is a centered Gaussian
vector with independent coordinates. If such a vector $Z\in\mathbb R^K$
has coordinate variances at most $\nu^2$, its squared norm is stochastically
bounded by $\nu^2$ times a chi-square variable with $K$ degrees of freedom.
The tail inequality of \citep[Lemma~1]{laurent2000adaptive} therefore gives
\begin{equation}
\mathbb P\left\{\lVert Z\rVert_2>\nu\sqrt{C_cK\psi_K}\right\}\le p^{-c},\qquad C_c=\max\{2,3c\},
\label{eq:score-gaussian-tail}
\end{equation}
since $K+2\sqrt{Kc\log p}+2c\log p\le C_c(K+\log p)$.

On $\bigcap_{\ell=0}^L\mathcal E_{\rm col}^{\ell}$, the conditional
coordinate variances of the source and aggregate score blocks are at most
$Cn_S/(N^2K^2)$ and $CN/(N^2K^2)$, respectively. Since $N\le(L+1)n_S$,
\eqref{eq:score-gaussian-tail} bounds their row norms at the scales
$N^{-1}\sqrt{n_S\psi_K/K}$ and
$N^{-1}\sqrt{(L+1)n_S\psi_K/K}$. Choosing $c_0$ sufficiently large and taking
a union bound over the $(L+1)p$ blocks, then applying
\eqref{eq:score-column-probability}, yields
\begin{equation}
\mathbb P(\mathcal E_{\rm score}^c)\le C(L+1)p^{1-c}+C(L+1)Kp\exp(-c_1n_T).
\label{eq:score-failure-probability}
\end{equation}
For the target score, only $\mathcal E_{\rm col}^0$ is needed. On this event
its conditional coordinate variances are at most $C/(n_TK^2)$, so
\eqref{eq:score-gaussian-tail} gives the row scale
$\sqrt{\psi_K/(Kn_T)}$. A union bound over the $p$ target blocks, together
with \eqref{eq:score-column-probability} for $\ell=0$, gives
$\mathbb P(\mathcal E_{\rm tscore}^c)\le Cp^{1-c}+CKp\exp(-c_1n_T)$ after
enlarging $c_0$ if necessary.
\end{proof}
\begin{lemma}[Restricted curvature]
\label{lemma_rscrsm}
Under Assumption~\ref{assumption_task_specific_design}, for each $0\leq\ell\leq L$, with probability at least $1-C_1K\exp(-C_2n_\ell)$, the following inequalities hold simultaneously for all $\Delta^\ell\in\mathbb R^{pK}$:
\begin{align*}
(\Delta^\ell)^\top\widehat\Sigma^\ell\Delta^\ell
&\ge
\frac{\alpha_\ell}{K}\lVert\Delta^\ell\rVert_{2,2}^2
-\frac{\varphi_\ell\log p}{K n_\ell}\lVert\Delta^\ell\rVert_{2,1}^2,\\
(\Delta^\ell)^\top\widehat\Sigma^\ell\Delta^\ell
&\le
\frac{\gamma_\ell}{K}\lVert\Delta^\ell\rVert_{2,2}^2
+\frac{\tau_\ell\log p}{K n_\ell}\lVert\Delta^\ell\rVert_{2,1}^2.
\end{align*}
Here $\alpha_\ell$ and $\gamma_\ell$ are curvature constants and $\varphi_\ell$
and $\tau_\ell$ the corresponding mixed-norm tolerances; none is a coefficient
vector. The constants satisfy $0<\alpha_{\min}:=\min_\ell\alpha_\ell\leq\max_\ell\gamma_\ell<\infty$ and $\max_\ell(\varphi_\ell\vee\tau_\ell)<\infty$.
\end{lemma}
\begin{proof}
Fix $0\leq\ell\leq L$ and $k\in[K]$. Lemma~1 of
\citet{he2024transfusion}, based on Supplementary Lemmas~13 and~15 of
\citet{loh2012highdimensional}, gives, simultaneously for all
$u\in\mathbb R^p$,
\begin{align}
\frac{1}{n_\ell}\lVert X^{\ell k}u\rVert_2^2
&\ge c\lVert u\rVert_2^2-C\frac{\log p}{n_\ell}\lVert u\rVert_1^2,
\notag\\
\frac{1}{n_\ell}\lVert X^{\ell k}u\rVert_2^2
&\le C\lVert u\rVert_2^2+C\frac{\log p}{n_\ell}\lVert u\rVert_1^2.
\label{eq:taskwise-rsc-rsm}
\end{align}
The probability is at least $1-C_1\exp(-C_2n_\ell)$. The constants are
uniform because the sub-Gaussian parameters and covariance eigenvalues in
Assumption~\ref{assumption_task_specific_design} are uniformly bounded.
Write $\Delta^\ell=(u^{(1)\top},\ldots,u^{(K)\top})^\top$. A union bound over
$k\in[K]$ and summation of \eqref{eq:taskwise-rsc-rsm} give the Euclidean terms
in the lemma. For the tolerance terms, Minkowski's inequality yields
\begin{align}
\sum_{k=1}^K\lVert u^{(k)}\rVert_1^2
&=
\left\lVert\sum_{j=1}^p
\bigl(|u_j^{(1)}|,\ldots,|u_j^{(K)}|\bigr)^\top\right\rVert_2^2
\notag\\
&\leq
\left\{\sum_{j=1}^p
\left(\sum_{k=1}^K|u_j^{(k)}|^2\right)^{1/2}\right\}^2
\notag\\
&=\lVert\Delta^\ell\rVert_{2,1}^2.
\label{eq:taskwise-minkowski}
\end{align}
Dividing the summed inequalities by $K$ proves the claim.
\end{proof}
Let $\mathcal E_{\rm rsc}$ denote the event on which Lemma~\ref{lemma_rscrsm} holds for every $0\leq \ell\leq L$. Define
\[
\mathcal E_F:=\mathcal E_{\rm score}\cap\mathcal E_{\rm rsc},
\qquad
\mathcal E_D:=\mathcal E_F\cap\mathcal E_{\rm tscore}.
\]
By Lemmas~\ref{lemma_score_bounds} and~\ref{lemma_rscrsm},
\[
\mathbb P(\mathcal E_F)
\ge
1-C_1(L+1)Kp\exp(-C_2n_T)-C_3(L+1)p^{1-c}.
\]
When the debiasing tuning parameter is chosen as in Theorem~\ref{thm:debiased}, the same probability bound holds for $\mathcal E_D$, with possibly different constants.
Both first-stage tunings satisfy Lemma~\ref{lemma_score_bounds}. Indeed, relative to its respective lower bounds, the ratios for $\lambda_0$ are $1$ and $\sqrt{N/n_S}\ge1$, and the ratio for each $\lambda_\ell$ is $8\sqrt{(L+1)n_S/N}\ge8$ in both regimes. Also $a_\ell\ge8n_S/N$ in both regimes because $n_S/N\le1$.
The remainder of the fused proof is deterministic on $\mathcal E_F$, and that of the debiased proof is deterministic on $\mathcal E_D$. The event probabilities in Lemmas~\ref{lemma_score_bounds} and~\ref{lemma_rscrsm} yield the probability bounds in Theorems~\ref{thm:fused} and~\ref{thm:debiased}.
\subsection{Deterministic inequalities for the fused estimator}
\label{app:fused-deterministic}
On $\mathcal E_F$, optimality first gives a basic inequality comparing the
quadratic loss with the penalty contributions. The resulting cone inequality bounds
the total weighted error by the fused error on the target support, together
with the heterogeneity term $H_\lambda$. This restriction makes the curvature
bound useful even when the empirical design matrices are singular. Combining
the two inequalities then bounds the coefficient error of the fused average.
Subscripts $S$ and $S^c$ restrict a vector to the feature groups inside and
outside the target support, respectively.
\begin{lemma}[Fused basic inequality and cone]
\label{lemma_fused_basic_cone}
Assume $\mathcal E_{\rm score}$ holds and $a_\ell=\lambda_\ell/\lambda_0\ge 8n_S/N$ for every $\ell\in[L]$. Let $S=\mathcal S(\beta^0)$. Then the fused error satisfies
\begin{equation}
0
\ge
\frac12\widehat\Delta^\top\widehat\Sigma\widehat\Delta
-\frac32\lambda_0\lVert\widehat\Delta^F_S\rVert_{2,1}
+\frac12\lambda_0\lVert\widehat\Delta^F_{S^c}\rVert_{2,1}
-2H_\lambda,
\label{eq:fused-basic-lower}
\end{equation}
and
\begin{equation}
\sum_{\ell=1}^L\lambda_\ell\lVert\widehat\Delta^\ell\rVert_{2,1}
+
\lambda_0\lVert\widehat\Delta^0\rVert_{2,1}
\leq
8\left(\lambda_0\lVert\widehat\Delta^F_S\rVert_{2,1}+H_\lambda\right).
\label{eq:fused-cone}
\end{equation}
\end{lemma}
\begin{proof}
Set $\eta_S:=n_S/N$ and
$P_\Delta:=\sum_{\ell=1}^L\lambda_\ell
\lVert\widehat\Delta^\ell\rVert_{2,1}
+\lambda_0\lVert\widehat\Delta^0\rVert_{2,1}$.
By optimality of $\widehat\theta$, expansion of the quadratic loss at $\theta$,
and the definition of $\mathcal E_{\rm score}$,
\begin{align}
0
&\ge
\frac12\widehat\Delta^\top\widehat\Sigma\widehat\Delta
-\frac12\sum_{\ell=1}^L\lambda_\ell\lVert\widehat\Delta^\ell\rVert_{2,1}
-\frac12\lambda_0\lVert\widehat\Delta^0\rVert_{2,1}
\notag\\
&\quad+
\sum_{\ell=1}^L\lambda_\ell
\{\lVert\theta^\ell+\widehat\Delta^\ell\rVert_{2,1}-\lVert\theta^\ell\rVert_{2,1}\}
+
\lambda_0\{\lVert\theta^0+\widehat\Delta^0\rVert_{2,1}-\lVert\theta^0\rVert_{2,1}\}.
\label{eq:fused-basic-score}
\end{align}
For each source contrast, the triangle inequality and
$\lVert\theta^\ell\rVert_{2,1}=h_\ell^\star$ give the first inequality below.
Since $\theta^0_{S^c}=0$, decomposability gives the second:
\begin{align}
\lVert\theta^\ell+\widehat\Delta^\ell\rVert_{2,1}
-\lVert\theta^\ell\rVert_{2,1}
&\geq \lVert\widehat\Delta^\ell\rVert_{2,1}-2h_\ell^\star,
\qquad \ell\in[L],
\notag\\
\lVert\theta^0+\widehat\Delta^0\rVert_{2,1}-\lVert\theta^0\rVert_{2,1}
&\geq
-\lVert\widehat\Delta^0_S\rVert_{2,1}
+\lVert\widehat\Delta^0_{S^c}\rVert_{2,1}.
\label{eq:fused-penalty-lower-bounds}
\end{align}
Substitution into \eqref{eq:fused-basic-score} yields
\begin{align}
0
&\geq
\frac12\widehat\Delta^\top\widehat\Sigma\widehat\Delta
+
\frac12\sum_{\ell=1}^L\lambda_\ell\lVert\widehat\Delta^\ell\rVert_{2,1}
-\frac32\lambda_0\lVert\widehat\Delta^0_S\rVert_{2,1}
+\frac12\lambda_0\lVert\widehat\Delta^0_{S^c}\rVert_{2,1}
-2H_\lambda.
\label{eq:fused-basic-target}
\end{align}
Because $\widehat\Sigma$ is positive semidefinite, dropping its quadratic form
from \eqref{eq:fused-basic-target}, multiplying by two, and then adding
$\lambda_0\lVert\widehat\Delta^0_S\rVert_{2,1}$ to both sides give
\begin{equation}
P_\Delta
\leq
4\lambda_0\lVert\widehat\Delta^0_S\rVert_{2,1}+4H_\lambda.
\label{eq:fused-total-penalty}
\end{equation}
The identity
$\widehat\Delta^0=\widehat\Delta^F-\eta_S
\sum_{\ell=1}^L\widehat\Delta^\ell$ and
$\lambda_\ell\geq8\eta_S\lambda_0$ imply
\begin{align}
\lambda_0\lVert\widehat\Delta^0_S\rVert_{2,1}
&\leq
\lambda_0\lVert\widehat\Delta^F_S\rVert_{2,1}
+\eta_S\lambda_0\sum_{\ell=1}^L
\lVert\widehat\Delta^\ell_S\rVert_{2,1}
\notag\\
&\leq
\lambda_0\lVert\widehat\Delta^F_S\rVert_{2,1}
+\frac18P_\Delta.
\label{eq:fused-support-transfer}
\end{align}
Combining \eqref{eq:fused-total-penalty} and
\eqref{eq:fused-support-transfer} gives
\newpage
\begin{equation}
P_\Delta
\leq
4\lambda_0\lVert\widehat\Delta^F_S\rVert_{2,1}
+P_\Delta/2+4H_\lambda.
\label{eq:fused-cone-absorption}
\end{equation}
Rearranging \eqref{eq:fused-cone-absorption} proves
\eqref{eq:fused-cone}.
It remains to express \eqref{eq:fused-basic-target} in terms of
$\widehat\Delta^F$. The same identity and the triangle inequality give
\begin{align}
&\frac12\sum_{\ell=1}^L
\lambda_\ell\lVert\widehat\Delta^\ell\rVert_{2,1}
-\frac32\lambda_0\lVert\widehat\Delta^0_S\rVert_{2,1}
+\frac12\lambda_0\lVert\widehat\Delta^0_{S^c}\rVert_{2,1}
\notag\\
&\quad\geq
-\frac32\lambda_0\lVert\widehat\Delta^F_S\rVert_{2,1}
+\frac12\lambda_0\lVert\widehat\Delta^F_{S^c}\rVert_{2,1}
\notag\\
&\qquad+
\sum_{\ell=1}^L
\left\{
\left(\frac{\lambda_\ell}{2}-\frac{3\eta_S\lambda_0}{2}\right)
\lVert\widehat\Delta^\ell_S\rVert_{2,1}
+
\left(\frac{\lambda_\ell}{2}-\frac{\eta_S\lambda_0}{2}\right)
\lVert\widehat\Delta^\ell_{S^c}\rVert_{2,1}
\right\}
\notag\\
&\quad\geq
-\frac32\lambda_0\lVert\widehat\Delta^F_S\rVert_{2,1}
+\frac12\lambda_0\lVert\widehat\Delta^F_{S^c}\rVert_{2,1},
\label{eq:fused-penalty-transfer}
\end{align}
where the last inequality follows from
$\lambda_\ell\geq8\eta_S\lambda_0$. Substitution of
\eqref{eq:fused-penalty-transfer} into \eqref{eq:fused-basic-target} proves
\eqref{eq:fused-basic-lower}.
\end{proof}
The role of \eqref{eq:fused-basic-lower} is especially clear when
$H_\lambda=0$. Dropping its nonnegative quadratic term gives the familiar
cone restriction
$\lVert\widehat\Delta^F_{S^c}\rVert_{2,1}\leq
3\lVert\widehat\Delta^F_S\rVert_{2,1}$: the error outside the target support
cannot dominate the error on that support. For nonzero $H_\lambda$, the same
argument allows an additional $4H_\lambda/\lambda_0$ on the right-hand side.

The next lemma converts curvature for the individual domain losses into
curvature for their weighted coefficient average. To control the accumulated
mixed-norm tolerances, define
\begin{equation}
\underline\lambda:=\min_{\ell\in[L]}\lambda_\ell,\qquad D_\lambda:=\underline\lambda^2\wedge\frac{\lambda_0^2}{L+1}.
\label{eq:fused-penalty-curvature-scale}
\end{equation}
The factor $L+1$ accounts for the target block appearing in every domain loss.
The choice of $D_\lambda$ lets the squared tolerances be bounded by the square
of the weighted penalty error in \eqref{eq:fused-cone}.
\begin{lemma}[Fused average curvature]
\label{lemma_fused_average_curvature}
On $\mathcal E_{\rm rsc}$, if \eqref{eq:fused-cone} holds, then
\[
\widehat\Delta^\top\widehat\Sigma\widehat\Delta
\ge
\frac{\alpha_{\min}(1-u_m)}{K}
\lVert\widehat\Delta^F\rVert_{2,2}^2
-
\frac{v_m}{K}H_\lambda,
\]
where
\[
u_m=C_u
\frac{\lambda_0^2}{D_\lambda}
\frac{s\log p}{N},
\qquad
v_m=C_v
\frac{1}{D_\lambda}
\frac{\log p}{N}H_\lambda.
\]
Here $C_u,C_v>0$ are fixed constants depending only on the fixed design and
curvature constants.
For the fused-stage tunings in Theorems~\ref{thm:fused} and~\ref{thm:debiased}, $u_m=o(1)$. Under the corresponding theorem assumptions, there is a fixed $M_v<\infty$ such that $v_m/K\le M_v$ for all sufficiently large $m$.
\end{lemma}
Here $u_m$ measures the fraction of the quadratic curvature lost to the
sparsity tolerance. The condition $u_m=o(1)$ therefore leaves a positive
curvature coefficient for all sufficiently large $m$. The term
$v_mH_\lambda/K$ is an additive remainder from source heterogeneity. Bounding
$v_m/K$ keeps this remainder of the same order as the heterogeneity term in
\eqref{eq:fused-basic-lower}.
\begin{proof}
Since $Ln_S+n_T=N$, weighted Jensen's inequality gives
\begin{align}
\lVert\widehat\Delta^F\rVert_{2,2}^2
&=\left\lVert\frac{n_S}{N}\sum_{\ell=1}^L(\widehat\Delta^\ell+\widehat\Delta^0)+\frac{n_T}{N}\widehat\Delta^0\right\rVert_{2,2}^2\notag\\
&\le\frac{n_S}{N}\sum_{\ell=1}^L\lVert\widehat\Delta^\ell+\widehat\Delta^0\rVert_{2,2}^2+\frac{n_T}{N}\lVert\widehat\Delta^0\rVert_{2,2}^2.
\label{eq:fused-weighted-jensen}
\end{align}
Lemma~\ref{lemma_rscrsm} and \eqref{eq:fused-weighted-jensen} imply
\begin{equation}
\widehat\Delta^\top\widehat\Sigma\widehat\Delta\ge\frac{\alpha_{\min}}{K}\lVert\widehat\Delta^F\rVert_{2,2}^2-G(\widehat\Delta),
\label{eq:fused-curvature-tolerance}
\end{equation}
where the curvature tolerance satisfies
\begin{align}
G(\widehat\Delta)
&\le C\frac{\log p}{KN}\left\{\sum_{\ell=1}^L\lVert\widehat\Delta^\ell\rVert_{2,1}^2+(L+1)\lVert\widehat\Delta^0\rVert_{2,1}^2\right\}\notag\\
&\le C\frac{\log p}{KND_\lambda}\left(\sum_{\ell=1}^L\lambda_\ell\lVert\widehat\Delta^\ell\rVert_{2,1}+\lambda_0\lVert\widehat\Delta^0\rVert_{2,1}\right)^2.
\label{eq:fused-tolerance-penalty}
\end{align}
The first inequality follows from the triangle inequality, and the second
from the definition of $D_\lambda$. Applying \eqref{eq:fused-cone} and
$\lVert\widehat\Delta^F_S\rVert_{2,1}\le\sqrt{s}\lVert\widehat\Delta^F\rVert_{2,2}$ gives
\begin{equation}
G(\widehat\Delta)\le C\frac{\log p}{KND_\lambda}\left(\lambda_0^2s\lVert\widehat\Delta^F\rVert_{2,2}^2+H_\lambda^2\right).
\label{eq:fused-tolerance-cone}
\end{equation}
Substituting \eqref{eq:fused-tolerance-cone} into
\eqref{eq:fused-curvature-tolerance} proves the curvature inequality with
the stated $u_m$ and $v_m$.

For the first fused-stage tuning, $\lambda_\ell=8\sqrt{n_S/N}\lambda_0$
and $N\le(L+1)n_S$, so $D_\lambda\asymp\lambda_0^2/(L+1)$. Hence
\begin{equation}
u_m\lesssim\frac{(L+1)s\log p}{N}\lesssim\frac{s\log p}{n_S}=o(1).
\label{eq:fused-curvature-first-u}
\end{equation}
Also, $H_\lambda\asymp\bar h\sqrt{(L+1)\psi_K/(KN)}$, giving
\begin{equation}
\frac{v_m}{K}\lesssim\frac{L\bar h\log p}{\sqrt{K n_S\psi_K}}\lesssim M_F.
\label{eq:fused-curvature-first-v}
\end{equation}
For the second fused-stage tuning in Theorem~\ref{thm:debiased},
$\lambda_\ell=8(n_S/N)\lambda_0$, and
\begin{equation}
D_\lambda\asymp\lambda_0^2\left\{\left(\frac{n_S}{N}\right)^2\wedge\frac{1}{L+1}\right\}.
\label{eq:fused-curvature-second-D}
\end{equation}
Together with $H_\lambda\asymp\lambda_0\bar h$, this yields
\begin{equation}
u_m\lesssim\frac{Ls\log p}{n_S}=o(1),\qquad\frac{v_m}{K}\lesssim\frac{L\bar h\log p}{\sqrt{K n_S\psi_K}}\lesssim M_F.
\label{eq:fused-curvature-second-uv}
\end{equation}
Since $\psi_K\ge1$, the theorem assumptions therefore give
$v_m/K\le M_v:=C_M M_F$ eventually in both regimes, where $C_M$ depends
only on the fixed design, curvature, and tuning constants.
\end{proof}
Combining Lemmas~\ref{lemma_fused_basic_cone}
and~\ref{lemma_fused_average_curvature} gives the two norms needed below.
The $\ell_{2,2}$ bound controls the coefficient error reported in
Theorem~\ref{thm:fused}. The $\ell_{2,1}$ bound will also be needed to control
the target curvature tolerance in the correction proof.
\begin{lemma}[Fused auxiliary norms]
\label{lemma_fused_auxiliary_norms}
Suppose that $\mathcal E_{\rm score}\cap\mathcal E_{\rm rsc}$ holds,
$a_\ell\ge 8n_S/N$ for every $\ell\in[L]$, and $u_m<1$. Then
\begin{align*}
\frac1{\sqrt K}\lVert\widehat\Delta^F\rVert_{2,2}
&\le
\frac{3}{\alpha_{\min}(1-u_m)}\sqrt{Ks}\lambda_0
+
\left\{
\frac{4+v_m/K}{\alpha_{\min}(1-u_m)}
\right\}^{1/2}\sqrt{H_\lambda},\\
\frac1{\sqrt K}\lVert\widehat\Delta^F\rVert_{2,1}
&\le
\frac{12}{\alpha_{\min}(1-u_m)}\sqrt K\,s\lambda_0
+4\left\{
\frac{s(4+v_m/K)}{\alpha_{\min}(1-u_m)}
\right\}^{1/2}\sqrt{H_\lambda}
+\frac{4H_\lambda}{\sqrt K\lambda_0}.
\end{align*}
In particular, if $u_m\le1/2$ and $v_m/K\le M_v$ for some fixed
$M_v<\infty$, then
\begin{align*}
\frac1{\sqrt K}\lVert\widehat\Delta^F\rVert_{2,2}
&\le
C\left(\sqrt{Ks}\lambda_0+\sqrt{H_\lambda}\right),\\
\frac1{\sqrt K}\lVert\widehat\Delta^F\rVert_{2,1}
&\le
C\left(
\sqrt K\,s\lambda_0
+\sqrt{s}\sqrt{H_\lambda}
+\frac{H_\lambda}{\sqrt K\lambda_0}
\right),
\end{align*}
where $C$ may depend on $M_v$ and the fixed curvature constants.
\end{lemma}
\begin{proof}
Lemma~\ref{lemma_fused_basic_cone} gives \eqref{eq:fused-cone} and
\eqref{eq:fused-basic-lower}. Combining \eqref{eq:fused-basic-lower} with
Lemma~\ref{lemma_fused_average_curvature}, dropping the nonnegative
off-support term, and using
$\lVert\widehat\Delta^F_S\rVert_{2,1}\le
\sqrt{s}\lVert\widehat\Delta^F\rVert_{2,2}$ gives
\[
0\ge
\frac{\alpha_{\min}(1-u_m)}{2K}
\lVert\widehat\Delta^F\rVert_{2,2}^2
-\frac32\lambda_0\sqrt{s}\lVert\widehat\Delta^F\rVert_{2,2}
-\left(2+\frac{v_m}{2K}\right)H_\lambda.
\]
Because $u_m<1$, solving this quadratic inequality gives the first bound.
For the mixed norm, \eqref{eq:fused-basic-lower} and the positive
semidefiniteness of $\widehat\Sigma$ imply
\[
\lVert\widehat\Delta^F\rVert_{2,1}
\le
4\sqrt{s}\lVert\widehat\Delta^F\rVert_{2,2}
+\frac{4H_\lambda}{\lambda_0}.
\]
Substituting the first bound proves the second. The simplified bounds follow
by applying $u_m\le1/2$ and $v_m/K\le M_v$.
\end{proof}
\subsection{Proof of Theorem~\ref{thm:fused}}
\label{app:proof-fused}
\begin{proof}
Work on $\mathcal E_F$. We combine the deterministic fused bound with the
tuning in Theorem~\ref{thm:fused} and then account for the deterministic bias
$\epsilon_D$.
Under the assumptions of Theorem~\ref{thm:fused},
Lemma~\ref{lemma_fused_average_curvature} gives $u_m\le1/2$ and
$v_m/K\le M_v$ eventually along the sequence, for some fixed $M_v<\infty$.
Thus, for all sufficiently large $m$, the simplified conclusion of
Lemma~\ref{lemma_fused_auxiliary_norms} gives
\begin{equation}
\frac1{\sqrt K}\lVert\widehat\Delta^F\rVert_{2,2}
\le
C\left(\sqrt{Ks}\lambda_0+\sqrt{H_\lambda}\right).
\label{eq:fused-proof-auxiliary}
\end{equation}
The prescribed tuning gives
\begin{align}
\sqrt{Ks}\lambda_0
&\lesssim
\sqrt{\frac{s\psi_K}{N}},
\notag\\
H_\lambda
&=\sum_{\ell=1}^L\lambda_\ell h_\ell^\star
\lesssim
\bar h\sqrt{\frac{\psi_K}{K n_S}}.
\label{eq:fused-proof-tuning}
\end{align}
Therefore, \eqref{eq:fused-proof-auxiliary} and
\eqref{eq:fused-proof-tuning} imply
\begin{equation}
\frac1{\sqrt K}\lVert\widehat w^F-w^F\rVert_{2,2}
\le
C\left(
\sqrt{\frac{s\psi_K}{N}}
+
\sqrt{\frac{\bar h}{\sqrt K}\sqrt{\frac{\psi_K}{n_S}}}
\right).
\label{eq:fused-proof-estimation}
\end{equation}
Finally,
\begin{equation}
\lVert\widehat w^F-\beta^0\rVert_{2,2}
\le
\lVert\widehat w^F-w^F\rVert_{2,2}
+
\lVert w^F-\beta^0\rVert_{2,2}
\le
\lVert\widehat\Delta^F\rVert_{2,2}+\epsilon_D,
\label{eq:fused-proof-triangle}
\end{equation}
where $\lVert w^F-\beta^0\rVert_{2,2}\leq
\lVert w^F-\beta^0\rVert_{2,1}=\epsilon_D$. Equations
\eqref{eq:fused-proof-estimation} and \eqref{eq:fused-proof-triangle} prove
\eqref{rate_fused} because the vectorized $\ell_{2,2}$ norm equals the matrix
Frobenius norm. The stated probability bound follows from
Section~\ref{app:high-prob-events}.
\end{proof}
\subsection{Debiasing inequalities and quadratic control}
\label{app:debiasing-auxiliary}
The population correction $d^\star=\beta^0-w^F$ exactly removes the discrepancy
between $w^F$ and the target. The fitted procedure instead starts from
$\widehat w^F$, so even adding $d^\star$ leaves the estimation error
$\widehat\Delta^F$. Define the correction error by
\[
d^\star=\beta^0-w^F,
\qquad
\widehat\Delta^D=\widehat d-d^\star.
\]
The population correction has mixed norm
$\lVert d^\star\rVert_{2,1}=\epsilon_D$. The error of the final debiased
estimator is the sum of the fused and correction errors:
\begin{equation}
\widehat w^D-\beta^0=\widehat\Delta^F+\widehat\Delta^D.
\label{eq:debiased-error-decomposition}
\end{equation}
The correction proof must therefore account for the prediction discrepancy
inherited from estimating $w^F$. On the target sample, this discrepancy is
\begin{equation}
Q_F:=(\widehat\Delta^F)^\top\widehat\Sigma^0\widehat\Delta^F=\frac{1}{n_TK}\lVert\mathcal X^0(\widehat w^F-w^F)\rVert_2^2.
\label{eq:debiased-inherited-prediction}
\end{equation}
All arguments hold on $\mathcal E_D$ and do not require independence between
the fused fit and the correction, which use the same target sample.

The four lemmas below separate the roles of the two stages.
Lemma~\ref{lemma_debias_basic} bounds the correction's quadratic error and mixed norm
in terms of $Q_F$ and the size $\epsilon_D$ of the population correction.
Lemma~\ref{lemma_debias_correction} converts this bound into coefficient
error using target curvature. Lemma~\ref{lemma_qf_control} controls $Q_F$
through the fused error norms, and Lemma~\ref{lemma_debiased_rate_algebra}
checks the resulting bounds under the prescribed tunings.

Let
\[
\widetilde{\mathcal L}(d)=\frac{1}{2n_TK}\lVert y^0-\mathcal X^0(\widehat w^F+d)\rVert_2^2
\]
denote the target-domain loss at the corrected coefficient $\widehat w^F+d$.
To compare the fitted correction with $d^\star$, write the penalized objective
difference as
\[
\mathcal J(\Delta)
=
\widetilde{\mathcal L}(d^\star+\Delta)-\widetilde{\mathcal L}(d^\star)
+\widetilde\lambda\{\lVert d^\star+\Delta\rVert_{2,1}-\lVert d^\star\rVert_{2,1}\}.
\]
\begin{lemma}[Debiasing basic inequality]
\label{lemma_debias_basic}
On $\mathcal E_{\rm tscore}$,
\[
\frac{\widetilde\lambda}{2}\lVert\widehat\Delta^D\rVert_{2,1}
+
\frac14(\widehat\Delta^D)^\top\widehat\Sigma^0\widehat\Delta^D
\le
Q_F+2\widetilde\lambda\epsilon_D.
\]
\end{lemma}
\begin{proof}
Since $\widehat d=d^\star+\widehat\Delta^D$ minimizes the debiasing objective, $\mathcal J(\widehat\Delta^D)\le0$. At $d^\star$, the residual is $\epsilon^0-\mathcal X^0\widehat\Delta^F$, so
\newpage
\[
\nabla\widetilde{\mathcal L}(d^\star)
=-\frac{1}{n_TK}(\mathcal X^0)^\top\epsilon^0+
\widehat\Sigma^0\widehat\Delta^F.
\]
Using $\mathcal E_{\rm tscore}$, H\"older's inequality, and Young's inequality for the positive semidefinite matrix $\widehat\Sigma^0$ gives
\[
\langle\nabla\widetilde{\mathcal L}(d^\star),\widehat\Delta^D\rangle
\ge
-\frac{\widetilde\lambda}{2}\lVert\widehat\Delta^D\rVert_{2,1}
-\frac14(\widehat\Delta^D)^\top\widehat\Sigma^0\widehat\Delta^D
-Q_F.
\]
The penalty difference is bounded below by
\[
\lVert d^\star+\widehat\Delta^D\rVert_{2,1}-\lVert d^\star\rVert_{2,1}
\ge
\lVert\widehat\Delta^D\rVert_{2,1}-2\epsilon_D.
\]
Substitution into $\mathcal J(\widehat\Delta^D)\le0$ proves the claim.
\end{proof}
Lemma~\ref{lemma_debias_basic} controls an empirical quadratic error. To bound
the coefficient error, the next lemma uses the lower curvature bound of
Lemma~\ref{lemma_rscrsm}. Its smallness condition ensures that the mixed-norm
tolerance contributes at most a constant multiple of
$Q_F+2\widetilde\lambda\epsilon_D$; otherwise, small prediction error alone
would not give the stated coefficient bound.
\begin{lemma}[Debiased correction]
\label{lemma_debias_correction}
Set $\widetilde\lambda=c_0\sqrt{\psi_K/(K n_T)}$ with a fixed sufficiently large $c_0$ as in Lemma~\ref{lemma_score_bounds}. On $\mathcal E_D$, suppose that
\[
\frac{\log p}{\psi_K}\{Q_F+2\widetilde\lambda\epsilon_D\}
\le c_\star
\]
for a sufficiently small constant $c_\star>0$ depending only on the target
curvature constants and the fixed tuning constant $c_0$. Then
\[
\frac1{\sqrt K}\lVert\widehat\Delta^D\rVert_{2,2}
\le
C\left(\sqrt{Q_F}+\sqrt{\widetilde\lambda\epsilon_D}\right).
\]
\end{lemma}
\begin{proof}
Let $q=Q_F+2\widetilde\lambda\epsilon_D$. Lemma~\ref{lemma_debias_basic} gives $\widetilde\lambda\lVert\widehat\Delta^D\rVert_{2,1}/2\le q$, and the target lower curvature bound in Lemma~\ref{lemma_rscrsm} gives
\[
\frac{\alpha_0}{4K}\lVert\widehat\Delta^D\rVert_{2,2}^2
\le
q+\frac{\varphi_0\log p}{4K n_T}\lVert\widehat\Delta^D\rVert_{2,1}^2.
\]
Using $\widetilde\lambda^2=c_0^2\psi_K/(K n_T)$ and $\widetilde\lambda\lVert\widehat\Delta^D\rVert_{2,1}/2\le q$, the tolerance term is at most
\[
\frac{\varphi_0}{c_0^2}q\left\{\frac{q\log p}{\psi_K}\right\}.
\]
Choose $c_\star>0$ so that $\varphi_0c_\star/c_0^2\le1$. Under the stated condition, the tolerance term is at most $q$, and hence
$K^{-1}\lVert\widehat\Delta^D\rVert_{2,2}^2\le8q/\alpha_0$, which proves the result.
\end{proof}
It remains to control $Q_F$ using the fused fit. Here the upper curvature
bound is needed: it turns the coefficient norms from
Lemma~\ref{lemma_fused_auxiliary_norms} into an upper bound on the prediction
discrepancy in \eqref{eq:debiased-inherited-prediction}.
\begin{lemma}[Target-domain quadratic error]
\label{lemma_qf_control}
On $\mathcal E_{\rm rsc}$,
\[
Q_F
\le
C\left\{
\frac1K\lVert\widehat\Delta^F\rVert_{2,2}^2
+
\frac{\log p}{K n_T}\lVert\widehat\Delta^F\rVert_{2,1}^2
\right\}.
\]
Consequently, on $\mathcal E_{\rm score}\cap\mathcal E_{\rm rsc}$, suppose
that the assumptions of Lemma~\ref{lemma_fused_auxiliary_norms} hold and that,
for some fixed $M_v<\infty$, $u_m\le1/2$ and $v_m/K\le M_v$. With
$H_\lambda=\sum_{\ell=1}^L\lambda_\ell h_\ell^\star$,
\begin{align*}
Q_F
\le C\Bigg[&
Ks\lambda_0^2+H_\lambda
+\frac{\log p}{n_T}
\left\{
K s^2\lambda_0^2+sH_\lambda+\frac{H_\lambda^2}{K\lambda_0^2}
\right\}
\Bigg].
\end{align*}
The constant $C$ in the displayed upper bound may depend on $M_v$ and the fixed
curvature constants.
\end{lemma}
\begin{proof}
The first claim is the target-domain upper curvature bound in Lemma~\ref{lemma_rscrsm}. Substituting the two bounds from Lemma~\ref{lemma_fused_auxiliary_norms} into that claim and applying $(a+b+c)^2\le3(a^2+b^2+c^2)$ proves the displayed upper bound.
\end{proof}
\begin{lemma}[Rate calculations for the debiased estimator]
\label{lemma_debiased_rate_algebra}
Under the assumptions of Theorem~\ref{thm:debiased} and with the two-stage
tuning specified there, the following bounds hold on $\mathcal E_F$ for all
sufficiently large $m$:
\[
\frac1{\sqrt K}\lVert\widehat\Delta^F\rVert_{2,2}+\sqrt{Q_F}
\le
C\left[
\sqrt{\frac{s\psi_K}{N}}
+
\sqrt{\frac{\bar h}{\sqrt K}\sqrt{\frac{\psi_K}{n_T}}}
\right],
\]
and
\[
\sqrt{\widetilde\lambda\epsilon_D}
\le
C\sqrt{\frac{\bar h}{\sqrt K}\sqrt{\frac{\psi_K}{n_T}}}.
\]
Moreover,
\[
\frac{\log p}{\psi_K}\{Q_F+2\widetilde\lambda\epsilon_D\}=o(1).
\]
The constant $C$ may depend on the fixed design and noise constants and on
$c$, $c_0$, $M_F$, and $M_T$, but not on $m$.
\end{lemma}
\begin{proof}
Set
\begin{equation}
R_s:=\frac{s\psi_K}{N},\qquad R_h:=\frac{\bar h}{\sqrt K}\sqrt{\frac{\psi_K}{n_T}}.
\label{eq:debiased-rate-scales}
\end{equation}
These are the squared pooled-estimation and heterogeneity contributions to
the desired rate. Bounds of order $R_s+R_h$ for $Q_F$ and the squared
normalized fused error give the first conclusion. The same bounds will also
verify the smallness condition in Lemma~\ref{lemma_debias_correction}.
Since $\epsilon_D\le\bar h$ and
$\widetilde\lambda=c_0\sqrt{\psi_K/(Kn_T)}$,
\begin{equation}
\widetilde\lambda\epsilon_D\le c_0R_h.
\label{eq:debiased-rate-bias}
\end{equation}
Lemma~\ref{lemma_fused_average_curvature} gives $u_m\le1/2$ and
$v_m/K\le M_v$ for all sufficiently large $m$, so
Lemmas~\ref{lemma_fused_auxiliary_norms} and~\ref{lemma_qf_control} apply.

In regime~A of Theorem~\ref{thm:debiased}, its tuning and
$N=n_T+Ln_S\asymp Ln_S$ give
\begin{equation}
Ks\lambda_0^2\lesssim R_s,\qquad H_\lambda\lesssim\bar h\sqrt{\frac{\psi_K}{Kn_S}}\le R_h.
\label{eq:debiased-rate-A}
\end{equation}
Outside regime~A, $s\psi_K/n_S<R_h$, and the alternative tuning gives
\begin{equation}
Ks\lambda_0^2\lesssim\frac{s\psi_K}{n_S}<R_h,\qquad H_\lambda\lesssim\bar h\sqrt{\frac{\psi_K}{Kn_S}}\le R_h.
\label{eq:debiased-rate-outside-A}
\end{equation}
For either tuning, the first two tolerance terms in
Lemma~\ref{lemma_qf_control} satisfy
\begin{equation}
\frac{\log p}{n_T}\left(Ks^2\lambda_0^2+sH_\lambda\right)\le M_T\left(Ks\lambda_0^2+H_\lambda\right).
\label{eq:debiased-rate-first-tolerances}
\end{equation}
The third tolerance term depends on the relative source penalties. Since
$\sum_{\ell=1}^Lh_\ell^\star=(N/n_S)\bar h$, their two choices give
\begin{equation}
\frac{H_\lambda}{\lambda_0}=\begin{cases}8\sqrt{N/n_S}\bar h,&\text{in regime A},\\8\bar h,&\text{otherwise}.\end{cases}
\label{eq:debiased-heterogeneity-penalty-ratio}
\end{equation}
Squaring \eqref{eq:debiased-heterogeneity-penalty-ratio} introduces the factor
$N/n_S\asymp L$ in regime~A. Consequently,
\begin{equation}
\frac{\log p}{n_T}\frac{H_\lambda^2}{K\lambda_0^2}\lesssim\begin{cases}L\bar h^2\log p/(Kn_T),&\text{in regime A},\\\bar h^2\log p/(Kn_T),&\text{otherwise}.\end{cases}
\label{eq:debiased-rate-third-tolerance}
\end{equation}
If $R_h=0$, then $\bar h=H_\lambda=0$, and this term vanishes. If $R_h>0$,
its ratio to $R_h$ is bounded, up to a fixed constant, by
\begin{equation}
\begin{cases}\displaystyle\frac{L\bar h\log p}{\sqrt{Kn_T\psi_K}}\le\frac{Ls\log p}{n_S}=o(1),&\text{in regime A},\\[2mm]\displaystyle\frac{\bar h\log p}{\sqrt{Kn_T\psi_K}}\le\frac{\bar h\log p}{\sqrt{Kn_T}}=o(1),&\text{otherwise}.\end{cases}
\label{eq:debiased-rate-third-ratios}
\end{equation}
The first inequality uses $R_h\le s\psi_K/n_S$: the defining restriction of
regime~A is what controls its extra factor of $L$. The second uses
$\psi_K\ge1$ and the heterogeneity condition in Theorem~\ref{thm:debiased}.
Thus Lemma~\ref{lemma_qf_control} and
\eqref{eq:debiased-rate-A}--\eqref{eq:debiased-rate-third-ratios} yield
\begin{equation}
Ks\lambda_0^2+H_\lambda+Q_F\le C(R_s+R_h).
\label{eq:debiased-rate-combined}
\end{equation}
Applying Lemma~\ref{lemma_fused_auxiliary_norms} to
\eqref{eq:debiased-rate-combined} proves the first bound in the lemma;
\eqref{eq:debiased-rate-bias} proves the second. Finally,
\begin{equation}
\frac{\log p}{\psi_K}(R_s+R_h)=\frac{s\log p}{N}+\frac{\bar h\log p}{\sqrt{Kn_T\psi_K}}=o(1)
\label{eq:debiased-rate-smallness}
\end{equation}
by the assumptions of Theorem~\ref{thm:debiased}. Combining
\eqref{eq:debiased-rate-bias}, \eqref{eq:debiased-rate-combined}, and
\eqref{eq:debiased-rate-smallness} proves the final assertion.
\end{proof}
\subsection{Proof of Theorem~\ref{thm:debiased}}
\label{app:proof-debiased}
\begin{proof}
Work on $\mathcal E_D\subseteq\mathcal E_F$.
Lemma~\ref{lemma_debiased_rate_algebra} verifies the condition of
Lemma~\ref{lemma_debias_correction} for all sufficiently large $m$.
Applying these two lemmas to \eqref{eq:debiased-error-decomposition} gives
\begin{align}
\frac1{\sqrt K}\lVert\widehat w^D-\beta^0\rVert_{2,2}
&\le\frac1{\sqrt K}\lVert\widehat\Delta^F\rVert_{2,2}+C\left(\sqrt{Q_F}+\sqrt{\widetilde\lambda\epsilon_D}\right)\notag\\
&\le C\left[\sqrt{\frac{s\psi_K}{N}}+\sqrt{\frac{\bar h}{\sqrt K}\sqrt{\frac{\psi_K}{n_T}}}\right].
\label{eq:debiased-proof-combined}
\end{align}
Translating the vectorized $\ell_{2,2}$ norm to the Frobenius norm proves
\eqref{rate_debiased}. The stated probability bound follows from
Section~\ref{app:high-prob-events}.
\end{proof}

\subsection{Comparison with the single-output transfer rate}
\label{app:single-output-comparison}
Section~\ref{sec:theory} compares the leading estimation term of
\eqref{rate_fused} with its single-output counterpart. To make that
comparison precise, write the scalar-response analogues of the two
heterogeneity quantities for coefficient vectors
$\boldsymbol\beta^\ell\in\mathbb R^p$ as
\[
\bar h'
:=
\frac{n_S}{N}\sum_{\ell=1}^L
\lVert\boldsymbol{\beta}^\ell-\boldsymbol{\beta}^0\rVert_1,
\qquad
\epsilon_D'
:=
\left\lVert
\sum_{\ell=1}^L\frac{n_S}{N}
(\boldsymbol{\beta}^\ell-\boldsymbol{\beta}^0)
\right\rVert_1,
\]
and, translating the source-count notation of \citet{he2024transfusion} to $L$,
set $\varrho_{\mathrm{TF}}:=L\bar h'\sqrt{\log p/n_S}$. Theorem~1 of
\citet{he2024transfusion} gives the squared $\ell_2$-error rate
$s\log p/N+(1+\varrho_{\mathrm{TF}})\bar h'\sqrt{\log p/n_S}+(\epsilon_D')^2$, so
the corresponding $\ell_2$-norm rate is
\begin{equation}
\sqrt{\frac{s\log p}{N}}+\sqrt{1+\varrho_{\mathrm{TF}}}\sqrt{\bar h'\sqrt{\frac{\log p}{n_S}}}+\epsilon_D'.
\label{rate_fused_he}
\end{equation}
The factor $\sqrt{1+\varrho_{\mathrm{TF}}}$ can be absorbed into the hidden
constant only when $\varrho_{\mathrm{TF}}=O(1)$, equivalently when
$n_S\gtrsim L^2(\bar h')^2\log p$; under the stronger condition
$n_S\gg L^2(\bar h')^2\log p$ it tends to one. In the scalar-response case
$K=1$, the conditions of Theorem~\ref{thm:fused} imply
$\varrho_{\mathrm{TF}}=O((\log p)^{-1/2})$, so the simplified rate without this
factor is valid in the regime of that theorem. Comparing
\eqref{rate_fused_he} with \eqref{rate_fused} therefore isolates the effect of
$\psi_K$: the mixed row norms replace the $\ell_1$ norms, and the
support-selection logarithm is divided among the responses.

%% file: appendix_d_lower_bound.tex
\section{Proof of the minimax lower bound}
\label{app:multitask-lower-bound}

{\emergencystretch=1.5em
This section proves Theorem~\ref{thm:multitask-transfer-lower-bound} in the
homogeneous Gaussian submodel \eqref{eq:lower-gaussian-model}. Throughout,
$\Theta_{\mathrm{MT}}(s,h)$, $\omega_{q,K}$, $r_{\mathrm{pool}}$, and
$r_{\mathrm{shift}}$ are as in Section~\ref{sec:lower-bound}. The integer $q$
denotes the number of active rows in a coefficient matrix. We use the
abbreviation
\begin{equation}
d_{q,K}:=K\omega_{q,K}=K+\log\left(\frac{\mathrm e p}{q}\right),\qquad 1\leq q\leq p.
\label{eq:lower-block-complexity}
\end{equation}
As $qd_{q,K} = qK + q\log(ep/q)$, the first term reflects the $qK$ coefficients within the $q$ active rows, and the second reflects the choice of those rows among the $p$ predictors.
\par}

The proof compares estimation to distinguishing among finitely many coefficient matrices. We construct matrices separated in Frobenius norm but with pairwise Kullback--Leibler divergences small relative to the logarithm of the family size. 
Fano's inequality then shows that no estimator can reliably identify the correct matrix. We apply this argument to the two separate constructions, both using the packing lemma below.
In the first, all $L+1$ domains share the same candidate matrix, so all $N$ observations help distinguish the candidates; this yields $r_{\mathrm{pool}}$. In the second construction, every source matrix is fixed at zero and only the target varies, so only the $n_T$ target observations help distinguish the candidates. Choosing the number of active target rows yields the three branches of $r_{\mathrm{shift}}$. 
The two constructions give separate lower bounds; taking the larger gives $(r_{\mathrm{pool}}+r_{\mathrm{shift}})/2$. We now state the packing lemma.

\begin{samepage}
\begin{lemma}[Row-sparse packing]
\label{lemma:lower-row-packing}
Suppose $p$ and $K$ are positive integers with $p\geq8$, and let
$q\in[p]$ satisfy $q\leq p/8$. For every $a>0$, there is a finite set
$\mathcal M(q,a)\subset\mathbb R^{p\times K}$ with
$|\mathcal M(q,a)|\ge16$ and universal constants
$c_{\mathrm p},C_{\mathrm p}>0$ such that
\begin{equation}
\log|\mathcal M(q,a)|\geq c_{\mathrm p}qd_{q,K}.
\label{eq:lower-packing-cardinality}
\end{equation}
Every $\mathbf A\in\mathcal M(q,a)$ has exactly $q$ nonzero rows, each with
Euclidean norm $a$, and all distinct $\mathbf A,\mathbf A'\in\mathcal M(q,a)$
satisfy
\begin{equation}
c_{\mathrm p}qa^2\leq\lVert\mathbf A-\mathbf A'\rVert_F^2\leq C_{\mathrm p}qa^2.
\label{eq:lower-packing-properties}
\end{equation}
In particular, $\lVert\mathbf A\rVert_{2,1}=qa$ for every
$\mathbf A\in\mathcal M(q,a)$.
\end{lemma}
\end{samepage}

\begin{proof}
We construct the row supports and the signs of the nonzero coefficients
separately. The support construction supplies the $q\log(\mathrm ep/q)$
part of the logarithm of the packing size, and the sign construction supplies
the $qK$ part. Both use the greedy Hamming-packing argument of
\citet{gilbert1952comparison}. 
Here a word is defined as an ordered list of values. 
In the support construction below, its entries will be integers specifying which row to select within each block of predictors.
In the sign construction, its entries will be $-1$ or $1$, specifying coefficient signs. 
An individual value allowed in a word is called a symbol, and the set of allowed values is called its alphabet.
Thus, the sign alphabet is $\{-1,1\}$; the support
alphabet will be specified when the blocks are defined later. 
The Hamming distance $d_H$ between two words of the same length is the number of positions at which their values differ.
We give the counting details, including those needed when $q$ or $K$ is small.

First, we record the elementary entropy estimate used in both constructions.
Let $H(x)=-x\log x-(1-x)\log(1-x)$ for $0<x<1$, with natural logarithms.
For any positive integer $m$, $0<x\leq1/2$, and $j\leq xm$, we have
$x^j(1-x)^{m-j}\geq\exp\{-mH(x)\}$, because $x/(1-x)\leq1$.
The binomial expansion therefore gives
\begin{equation}
1=\sum_{j=0}^{m}\binom mjx^j(1-x)^{m-j}\geq\exp\{-mH(x)\}\sum_{j=0}^{\lfloor xm\rfloor}\binom mj.
\label{eq:lower-entropy-count}
\end{equation}
Thus the last sum is at most $\exp\{mH(x)\}$. This bounds the number of
words near any fixed word once the alphabet size is accounted for.

Next, we construct a collection of well-separated row supports. Set $b=\lfloor p/q\rfloor\ge8$ and partition $qb$ of the
$p$ row indices into $q$ disjoint blocks of size $b$.
A word in $[b]^q$
chooses one row from each block and hence determines a support of size $q$.
Start with the $b$ constant words and add words greedily, always requiring
Hamming distance at least $\lceil q/4\rceil$ from every word already chosen.
The initial constant words satisfy this requirement, and the procedure
stops because $[b]^q$ is finite. At that point every word is within distance
$\lceil q/4\rceil-1$ of a chosen word; otherwise it could be added.
Consequently, Hamming balls of this radius around the chosen words cover
$[b]^q$.

A ball of radius $\lceil q/4\rceil-1$ contains
$\sum_{j=0}^{\lceil q/4\rceil-1}\binom qj(b-1)^j$ words: choose the $j$
positions that change and then one of $b-1$ new symbols at each position.
Since the radius is smaller than $q/4$, \eqref{eq:lower-entropy-count}
with $m=q$ and $x=1/4$ bounds this count by
$\exp\{qH(1/4)\}b^{q/4}$. Covering all $b^q$ words therefore requires at
least $\exp\{q[(3/4)\log b-H(1/4)]\}$ chosen words. Their supports form a
collection $\mathcal V_q$. The constant-word initialization also ensures
$|\mathcal V_q|\geq b\geq8$.

Because $b\geq8$ and $p/q<b+1$, the quantity
$(3/4)\log b-H(1/4)$ is bounded below by a universal positive multiple of
$\log(\mathrm ep/q)$. Moreover, each changed position replaces one selected
row by a different row in the same block. The symmetric difference of two
supports thus has twice the Hamming distance of their words. We obtain
\begin{equation}
\log|\mathcal V_q|\geq c q\log\left(\frac{\mathrm e p}{q}\right),\qquad |S\mathbin\triangle S'|\geq\frac q2\quad(S\ne S', S,S'\in\mathcal V_q).
\label{eq:lower-support-code}
\end{equation}

We then construct a separate collection of well-separated sign patterns for the $qK$ coefficients. Use words in $\{-1,1\}^{qK}$.
Start with the
all-positive and all-negative words, and extend greedily to a maximal
collection $\mathcal U_q$ with Hamming distance at least
$\lceil qK/8\rceil$. The same covering argument, now with alphabet size two,
uses \eqref{eq:lower-entropy-count} at $m=qK$ and $x=1/8$. It gives
$\log|\mathcal U_q|\geq qK\{\log2-H(1/8)\}$. Since
$\log2-H(1/8)>0$, and the initial two words remain in the collection,
\begin{equation}
\log|\mathcal U_q|\geq cqK,\qquad d_H(u,u')\geq\frac{qK}{8}\quad(u\ne u', u,u'\in\mathcal U_q),
\label{eq:lower-sign-code}
\end{equation}
with $|\mathcal U_q|\geq2$.

We now combine the supports and signs. For
$S=\{j_1<\cdots<j_q\}\in\mathcal V_q$ and $u\in\mathcal U_q$, define
$\mathbf A(S,u)$ by
\begin{equation}
A_{j_rk}(S,u):=\frac{a}{\sqrt K}u_{(r-1)K+k},\qquad r\in[q],\quad k\in[K],
\label{eq:lower-packing-matrix}
\end{equation}
and set all remaining rows to zero. Each active row has $K$ entries of
magnitude $a/\sqrt K$, so its Euclidean norm is $a$. In particular,
$\lVert\mathbf A(S,u)\rVert_{2,1}=qa$ and
$\lVert\mathbf A(S,u)\rVert_F^2=qa^2$.

If two matrices have different supports, each row in the symmetric
difference contributes $a^2$ to their squared distance. By
\eqref{eq:lower-support-code}, the total is at least $qa^2/2$.
If the supports agree, each changed sign contributes $4a^2/K$, and
\eqref{eq:lower-sign-code} gives
\begin{equation}
\lVert\mathbf A(S,u)-\mathbf A(S,u')\rVert_F^2=\frac{4a^2}{K}d_H(u,u')\geq\frac{qa^2}{2}\qquad(u\ne u').
\label{eq:lower-fixed-support-distance}
\end{equation}
The triangle inequality bounds every squared distance by $4qa^2$.
Thus \eqref{eq:lower-packing-properties} holds for
$\mathcal M(q,a):=\{\mathbf A(S,u):S\in\mathcal V_q,u\in\mathcal U_q\}$.
Each pair $(S,u)$ determines a different matrix, so the cardinalities
multiply. In particular, $|\mathcal M(q,a)|\geq8\cdot2=16$, and adding the
logarithmic bounds in \eqref{eq:lower-support-code} and
\eqref{eq:lower-sign-code} proves \eqref{eq:lower-packing-cardinality}.
The same universal constants work for all admissible $p,q,K$, including
$q=K=1$.
\end{proof}

\begin{proof}[Proof of Theorem~\ref{thm:multitask-transfer-lower-bound}]
We first relate the distance between coefficient tuples to the divergence
between their data distributions. Under \eqref{eq:lower-gaussian-model},
the design vectors have the same distribution for every coefficient tuple.
Conditional on a design vector $\mathbf x_i^\ell$, the response has mean
$(\mathbf B^\ell)^\top\mathbf x_i^\ell$ and covariance
$\sigma^2\mathbf I_K$. For two Gaussian distributions with this common
covariance, direct integration of the log density ratio gives divergence
equal to the squared distance between their means divided by $2\sigma^2$.
The chain rule for relative entropy
\citep[Theorem~2.5.3]{cover2006elements} and independence of the observations
therefore yield
\begin{equation}
D_{\mathrm{KL}}(\mathbb P_{\mathbf B}\Vert\mathbb P_{\mathbf B'})=\frac{1}{2\sigma^2}\sum_{\ell=0}^L\sum_{i=1}^{n_\ell}\mathbb E\left[\lVert(\mathbf B^\ell-\mathbf B'^\ell)^\top\mathbf x_i^\ell\rVert_2^2\right].
\label{eq:lower-kl-conditioning}
\end{equation}
Here the expectation is over the common design distribution. Since
$\mathbb E[\mathbf x_i^\ell(\mathbf x_i^\ell)^\top]=\mathbf I_p$, each
expectation in \eqref{eq:lower-kl-conditioning} is
$\lVert\mathbf B^\ell-\mathbf B'^\ell\rVert_F^2$. Hence
\begin{equation}
D_{\mathrm{KL}}(\mathbb P_{\mathbf B}\Vert\mathbb P_{\mathbf B'})=\frac{1}{2\sigma^2}\sum_{\ell=0}^Ln_\ell\lVert\mathbf B^\ell-\mathbf B'^\ell\rVert_F^2.
\label{eq:lower-kl-identity}
\end{equation}
This identity accounts for both the observed covariates and the responses. Because the covariates follow the same distribution under $\mathbb P_{\mathbf B}$ and $\mathbb P_{\mathbf B'}$, their distribution contributes no separate term to the KL divergence.

Next, we explain how a finite family gives a lower bound for estimation.
Consider coefficient tuples $\{\mathbf B_m:m\in[M]\}$ whose target matrices
satisfy the separation in \eqref{eq:lower-packing-properties}, and suppose
$M\geq16$, $\log M\geq c_{\mathrm p}qd_{q,K}$, and all pairwise divergences
are at most a sufficiently small universal multiple of $qd_{q,K}$. Write
$\mathbb P_m:=\mathbb P_{\mathbf B_m}$. Given any estimator
$\widetilde{\mathbf B}$ based on the complete data, let $\widehat J$ be the
index of the nearest target matrix $\mathbf B_m^0$ in Frobenius norm,
breaking ties by the smallest index. If its Frobenius error is less than
half the minimum distance between target matrices, the triangle inequality
ensures that this decoder selects the correct index. Equivalently, a
decoding error can occur only if the normalized squared estimation error is
at least $c_{\mathrm p}qa^2/(4K)$.

To bound the decoding error, let $J$ be uniform on $[M]$, and let $Z$ denote
the complete data drawn from $\mathbb P_J$. The marginal data distribution
is the mixture $\overline{\mathbb P}=M^{-1}\sum_m\mathbb P_m$.
The mutual information $I(J;Z)$ measures how much the data reveal about the
chosen index. Its expression in terms of relative entropy and convexity in
the second argument \citep[Theorem~2.7.2]{cover2006elements} give
\begin{equation}
I(J;Z)=\frac{1}{M}\sum_mD_{\mathrm{KL}}(\mathbb P_m\Vert\overline{\mathbb P})\leq\frac{1}{M^2}\sum_{m,m'}D_{\mathrm{KL}}(\mathbb P_m\Vert\mathbb P_{m'}).
\label{eq:lower-fano-information}
\end{equation}
{\emergencystretch=1.5em
Choosing the constant in the pairwise divergence bound small enough makes
\eqref{eq:lower-fano-information} at most $(\log M)/8$. Fano's inequality
\citep[Theorem~2.10.1]{cover2006elements}, written with natural logarithms,
then gives
\begin{equation}
\mathbb P(\widehat J\ne J)\geq1-\frac{I(J;Z)+\log2}{\log M}\geq1-\frac18-\frac14=\frac58.
\label{eq:lower-fano-testing}
\end{equation}
The last inequality uses $M\geq16$. Thus the probability of a decoding
error, averaged over the $M$ candidates, is bounded away from zero, and
at least one candidate has error probability at least as large. By the
estimation-to-decoding implication above, there is a universal
$c_{\mathrm F}>0$ such that
\begin{equation}
\inf_{\widetilde{\mathbf B}}\max_m\mathbb P_m\left\{\frac{1}{K}\lVert\widetilde{\mathbf B}-\mathbf B_m^0\rVert_F^2\geq c_{\mathrm F}\frac{qa^2}{K}\right\}\geq c_{\mathrm F}.
\label{eq:lower-fano-consequence}
\end{equation}
For example, any fixed positive $c_{\mathrm F}\leq\min\{c_{\mathrm p}/4,5/8\}$
works. The packing size guaranteed by Lemma~\ref{lemma:lower-row-packing}
therefore suffices without imposing further lower bounds on $p$, $K$, or $s$.
\par}

We now establish the pooled component of the lower bound. Assign the same coefficient matrix to every domain.
Take $q=s$ and choose
\begin{equation}
a_N:=c_{\mathrm{amp}}\sigma\sqrt{\frac{d_{s,K}}{N}},\qquad \mathbf B_m^0=\mathbf B_m^1=\cdots=\mathbf B_m^L:=\mathbf A_m\in\mathcal M(s,a_N),
\label{eq:lower-pooled-construction}
\end{equation}
where $c_{\mathrm{amp}}>0$ is a sufficiently small universal constant.
Every tuple in \eqref{eq:lower-pooled-construction} belongs to
$\Theta_{\mathrm{MT}}(s,h)$ for every $h\geq0$: its target has $s$ active
rows, and all source--target contrasts are zero. Since all domains share
the same matrix, \eqref{eq:lower-kl-identity} sums their sample sizes to $N$.
The upper distance bound in \eqref{eq:lower-packing-properties} gives
\begin{equation}
D_{\mathrm{KL}}(\mathbb P_{\mathbf B_m}\Vert\mathbb P_{\mathbf B_{m'}})\leq\frac{C_{\mathrm p}Ns a_N^2}{2\sigma^2}=\frac{C_{\mathrm p}c_{\mathrm{amp}}^2}{2}s d_{s,K}.
\label{eq:lower-pooled-divergence}
\end{equation}
By \eqref{eq:lower-packing-cardinality}, the logarithm of the family size is
at least $c_{\mathrm p}s d_{s,K}$. Choosing $c_{\mathrm{amp}}$ small enough
therefore verifies the divergence condition for
\eqref{eq:lower-fano-consequence}. The resulting normalized squared
separation has order $s a_N^2/K$, which equals
$c_{\mathrm{amp}}^2\sigma^2s\omega_{s,K}/N$. As this finite family is contained
in the parameter class, \eqref{eq:lower-fano-consequence} implies
\begin{equation}
\inf_{\widetilde{\mathbf B}}\sup_{\mathbf B\in\Theta_{\mathrm{MT}}(s,h)}\mathbb P_{\mathbf B}\left\{\frac{1}{K}\lVert\widetilde{\mathbf B}-\mathbf B^0\rVert_F^2\geq c r_{\mathrm{pool}}\right\}\geq c_{\mathrm F}.
\label{eq:lower-pooled-bound}
\end{equation}
This construction shows why the pooled estimation cost remains even when
the source and target coefficients agree exactly.

Now, we establish the source--target shift component. If $h=0$, then $r_{\mathrm{shift}} = 0$, so, \eqref{eq:lower-pooled-bound} already suffices.
Suppose $h>0$. Fix every source matrix at zero and let the target vary over
a packing with $q\leq s$ active rows. The row norm must satisfy two
requirements. The contrast constraint requires $qa\leq h$, while the
divergence condition requires $n_Tqa^2/\sigma^2$ to be a small multiple of
$qd_{q,K}$. Accordingly, for each $q\in[s]$, take
\begin{equation}
a_{T,q}:=c_{\mathrm{amp}}\left\{\frac hq\wedge\sigma\sqrt{\frac{d_{q,K}}{n_T}}\right\},\qquad \mathbf B_m^0:=\mathbf A_m\in\mathcal M(q,a_{T,q}),\qquad \mathbf B_m^\ell:=\mathbf0\quad(\ell\in[L]).
\label{eq:lower-shift-construction}
\end{equation}
Here we use the same universal $c_{\mathrm{amp}}$ as above, decreased if
necessary so that $c_{\mathrm{amp}}\leq1$. The target has at most $s$ active
rows, and its contrast with each source is
$\lVert\mathbf A_m\rVert_{2,1}=qa_{T,q}\leq c_{\mathrm{amp}}h\leq h$.
Thus every tuple belongs to $\Theta_{\mathrm{MT}}(s,h)$.

The source distributions are now identical across all candidates, so their
terms in \eqref{eq:lower-kl-identity} vanish. Only the target observations
contribute, and \eqref{eq:lower-packing-properties} gives
\newpage
\begin{equation}
D_{\mathrm{KL}}(\mathbb P_{\mathbf B_m}\Vert\mathbb P_{\mathbf B_{m'}})\leq\frac{C_{\mathrm p}n_Tq a_{T,q}^2}{2\sigma^2}\leq\frac{C_{\mathrm p}c_{\mathrm{amp}}^2}{2}q d_{q,K}.
\label{eq:lower-shift-divergence}
\end{equation}
Hence \eqref{eq:lower-fano-consequence} applies, even though the estimator
may use all target and source observations. Substituting $a_{T,q}$ into the
separation $qa_{T,q}^2/K$ yields, for every $q\in[s]$,
\begin{equation}
\inf_{\widetilde{\mathbf B}}\sup_{\mathbf B\in\Theta_{\mathrm{MT}}(s,h)}\mathbb P_{\mathbf B}\left\{\frac{1}{K}\lVert\widetilde{\mathbf B}-\mathbf B^0\rVert_F^2\geq c\left(\frac{h^2}{Kq}\wedge\frac{\sigma^2q d_{q,K}}{Kn_T}\right)\right\}\geq c_{\mathrm F}.
\label{eq:lower-shift-bound-q}
\end{equation}

It remains to choose $q$ so that \eqref{eq:lower-shift-bound-q} gives the
three branches of $r_{\mathrm{shift}}$. Set
$A_T:=\sigma^2d_{s,K}/n_T$. Since $d_{q,K}\geq d_{s,K}$ for $q\leq s$,
\begin{equation}
\frac{h^2}{Kq}\wedge\frac{\sigma^2q d_{q,K}}{Kn_T}\geq\frac1K\left(\frac{h^2}{q}\wedge A_Tq\right).
\label{eq:lower-support-balance}
\end{equation}
On the right, $h^2/q$ decreases with $q$ and $A_Tq$ increases with $q$.
They are equal at $q=h/\sqrt{A_T}$. This suggests balancing the terms when
that value lies between $1$ and $s$, and using an endpoint otherwise.
It is enough to make these choices in \eqref{eq:lower-support-balance};
we do not need the exact maximizer of the expression involving $d_{q,K}$.

If $h\leq\sqrt{A_T}$, choose $q=1$. Then the smaller term in
\eqref{eq:lower-support-balance} is $h^2/K$. In this range it is also the
smallest of $h^2/K$, $h\sqrt{A_T}/K$, and $sA_T/K$.

If $\sqrt{A_T}<h<s\sqrt{A_T}$, choose
$q=\lfloor h/\sqrt{A_T}\rfloor$. This is an integer in $[s]$ and satisfies
$h/(2\sqrt{A_T})\leq q\leq h/\sqrt{A_T}$. Consequently,
\begin{equation}
\frac{h^2}{q}\geq h\sqrt{A_T},\qquad A_Tq\geq\frac12h\sqrt{A_T},
\label{eq:lower-middle-regime}
\end{equation}
so \eqref{eq:lower-support-balance} is at least
$h\sqrt{A_T}/(2K)$. Here $h\sqrt{A_T}/K$ is the smallest of the three
expressions above. This intermediate range is empty when $s=1$.

Finally, if $h\geq s\sqrt{A_T}$, choose $q=s$. Then
$h^2/s\geq sA_T$, so \eqref{eq:lower-support-balance} equals $sA_T/K$ on
the right. This is the smallest of the three expressions in this range.
The adjacent rate expressions agree at the boundary values of $h$.

Combining these choices and using $d_{s,K}=K\omega_{s,K}$ gives
\begin{equation}
\max_{q\in[s]}\left\{\frac{h^2}{Kq}\wedge\frac{\sigma^2q d_{q,K}}{Kn_T}\right\}\geq\frac{1}{2K}\left\{h^2\wedge h\sqrt{A_T}\wedge sA_T\right\}=\frac12r_{\mathrm{shift}}.
\label{eq:lower-shift-optimization}
\end{equation}
The constants in \eqref{eq:lower-shift-bound-q} do not depend on $q$.
We may therefore use the choice corresponding to each fixed $h$ to conclude
\begin{equation}
\inf_{\widetilde{\mathbf B}}\sup_{\mathbf B\in\Theta_{\mathrm{MT}}(s,h)}\mathbb P_{\mathbf B}\left\{\frac{1}{K}\lVert\widetilde{\mathbf B}-\mathbf B^0\rVert_F^2\geq c r_{\mathrm{shift}}\right\}\geq c_{\mathrm F}.
\label{eq:lower-shift-bound}
\end{equation}
\newpage
To combine the two components, fix $(p,K,L,s,n_T,n_S,h)$ and decrease the
constant $c$ if needed so that it is the same in
\eqref{eq:lower-pooled-bound} and \eqref{eq:lower-shift-bound}. Use the pooled
construction when $r_{\mathrm{pool}}\geq r_{\mathrm{shift}}$ and the shift
construction otherwise. Since
$r_{\mathrm{pool}}\vee r_{\mathrm{shift}}\geq
(r_{\mathrm{pool}}+r_{\mathrm{shift}})/2$, this proves
\eqref{eq:lower-probability} with $c_{\mathrm{lb}}=c/2$ and
$c_1=c_{\mathrm F}$. The two components need not be realized by the same
coefficient tuple: the supremum over the parameter class includes both
constructions, and the larger component is at least half their sum.

Finally, the expectation bound follows because the loss is nonnegative.
For every estimator, every coefficient tuple, and every $t>0$,
\begin{equation}
\mathbb E_{\mathbf B}\left[\frac1K\lVert\widetilde{\mathbf B}-\mathbf B^0\rVert_F^2\right]\geq t\mathbb P_{\mathbf B}\left\{\frac1K\lVert\widetilde{\mathbf B}-\mathbf B^0\rVert_F^2\geq t\right\}.
\label{eq:lower-risk-from-probability}
\end{equation}
Take $t=c_{\mathrm{lb}}(r_{\mathrm{pool}}+r_{\mathrm{shift}})$, take the
supremum over coefficient tuples and then the infimum over estimators, and
apply \eqref{eq:lower-probability}. This gives
\eqref{eq:lower-expectation}.
\end{proof}

%% file: appendix_c_safe_tmtl.tex
\section{Target-certified transfer multitask learning}
\label{app:safe-tmtl}

Section~\ref{sec:projection} defines the target confidence set
$\mathcal C_T(\lambda_T,\kappa_T)$ in \eqref{eq:safe-confidence-set} and states
Theorem~\ref{thm:safe-projection} and Proposition~\ref{prop:safe-tmtl-rates}.
This section supplies the supporting lemmas and the proofs. Two properties of
the target set carry the argument. With high probability the set contains the
true coefficient matrix $\mathbf B^0$, so projection onto it cannot increase
Frobenius-norm distance from $\mathbf B^0$ and therefore preserves the
candidate's transfer guarantee. On the same event every matrix in the set has
normalized Frobenius error bounded by the target-only rate, which limits the
error when that guarantee is large. Both properties rest on the target-sample
and calibration conditions of Section~\ref{app:safe-preliminaries}; they
follow from the score bound of Section~\ref{app:safe-concentration} together
with the deterministic inequalities of Section~\ref{app:safe-deterministic}.

The target constraint follows \citet{he2026parameter}, whose TransMission
estimator constrains the target block inside a joint scalar-response
optimization and reports that block. We adapt their confidence set and the
truth-feasibility argument in the proof of their Lemma~15 in Appendix~B.4,
replacing the scalar $\ell_1$ and $\ell_\infty$ norms by the mixed
$\ell_{2,1}$ and $\ell_{2,\infty}$ norms. TMTL reports the weighted estimator
$\widehat{\mathbf W}^F$ in \eqref{fused} or its target-corrected version
$\widehat{\mathbf W}^D$ in \eqref{debiasedestimator}. The former combines the
fitted target block with source blocks that need not lie in the target set,
and the latter adds the correction $\widehat{\mathbf D}$. Thus, constraining
$\widehat{\mathbf B}^0$ alone would not control the reported estimator. We
instead project each reported coefficient matrix onto the target set.

\subsection{Assumptions and notation}
\label{app:safe-preliminaries}

Let $S=\mathcal S(\mathbf B^0)$, and let $\alpha_0$ and $\varphi_0$ be the
target curvature constants in Lemma~\ref{lemma_rscrsm}. Subscripts $S$ and
$S^c$ below restrict a matrix to the indicated rows. The target-domain least
squares loss, written in the matrix notation of the main text, is
\begin{equation}
\mathcal L_T(\mathbf B):=\frac{1}{2n_TK}\lVert\mathbf Y^0-\mathbf X^{0}\mathbf B\rVert_F^2,\qquad \nabla\mathcal L_T(\mathbf B)=\frac{1}{n_TK}(\mathbf X^{0})^\top(\mathbf X^{0}\mathbf B-\mathbf Y^0),
\label{eq:safe-target-loss}
\end{equation}
so that the first restriction in \eqref{eq:safe-confidence-set} reads
$\lVert\nabla\mathcal L_T(\mathbf B)\rVert_{2,\infty}\leq\lambda_T$. The target-domain
term of the loss $\mathcal L$ in Section~\ref{app:proof-preliminaries} equals
$(n_T/N)\mathcal L_T$, the two differing by that factor because $\mathcal L$
is normalized by $NK$ and $\mathcal L_T$ by $n_TK$; here the coefficient
matrix replaces the target and contrast blocks. Under vectorization, $\mathcal L_T$ is
the loss $\widetilde{\mathcal L}$ of Section~\ref{app:debiasing-auxiliary},
which carries the same $n_TK$ normalization, reparametrized by the coefficient
itself rather than by a correction to $\widehat w^F$. Define the target
score event

\begin{equation}
\mathcal E_{T,\mathrm{score}}(\lambda_T):=\left\{\lVert\nabla\mathcal L_T(\mathbf B^0)\rVert_{2,\infty}\leq\frac{\lambda_T}{2}\right\},
\label{eq:safe-score-event}
\end{equation}
and let $\mathcal E_{T,\mathrm{rsc}}$ denote the target-domain component of
$\mathcal E_{\mathrm{rsc}}$. On this second event, the lower curvature bound
of Lemma~\ref{lemma_rscrsm}, written in matrix notation, is
\begin{equation}
\frac{1}{n_TK}\lVert\mathbf X^0\Delta\rVert_F^2\geq\frac{\alpha_0}{K}\lVert\Delta\rVert_F^2-\frac{\varphi_0\log p}{Kn_T}\lVert\Delta\rVert_{2,1}^2\qquad\text{for every }\Delta\in\mathbb R^{p\times K}.
\label{eq:safe-target-curvature}
\end{equation}
The last term in \eqref{eq:safe-target-curvature} is a tolerance for high
dimensionality, which the sparsity conditions imposed below make small
relative to the first term for the errors under consideration. The score
event \eqref{eq:safe-score-event}, meanwhile, ensures that the random gradient
at the truth is smaller than the penalty scale.

\subsection{High-probability ingredients}
\label{app:safe-concentration}

The following bound uses only the target observations and imposes no condition
on source sample sizes.

\begin{lemma}[Target score concentration]
\label{lemma:safe-target-score}
Under Assumptions~\ref{assumption_design}--\ref{assumption_error}, fix $c>2$.
There are constants $c_T,C_1,C_2,C_3>0$ such that, if
\begin{equation}
\lambda_T\geq c_T\sqrt{\frac{\psi_K}{Kn_T}},
\label{eq:safe-score-tuning}
\end{equation}
then
\begin{equation}
\mathbb P\{\mathcal E_{T,\mathrm{score}}(\lambda_T)\}\geq1-C_1p\exp(-C_2n_T)-C_3p^{1-c}.
\label{eq:safe-score-probability}
\end{equation}
The constants depend only on the constants in the target-design and noise
assumptions and on $c$.
\end{lemma}

\begin{proof}
At the truth, $\nabla\mathcal L_T(\mathbf B^0)=-(n_TK)^{-1}(\mathbf X^{0})^\top\mathbf E^0$. For $j\in[p]$,
let $\mathbf X^0_{:,j}$ denote the $j$th column of $\mathbf X^0$ and set
\begin{equation}
Z_j:=\frac{1}{n_TK}(\mathbf X^0_{:,j})^\top\mathbf E^0\in\mathbb R^K.
\label{eq:safe-score-row}
\end{equation}
It suffices to control $\max_{j\in[p]}\lVert Z_j\rVert_2$. Define
\begin{equation}
\mathcal E_{T,\mathrm{col}}:=\left\{\max_{j\in[p]}\frac{1}{n_T}\lVert\mathbf X^0_{:,j}\rVert_2^2\leq C_{\mathrm{col}}\right\}.
\label{eq:safe-column-event}
\end{equation}
The design assumptions bound the second moment of each predictor uniformly.
Its squared observations have sub-exponential tails, so concentration of
their sample mean gives an exponentially small failure probability for each
column when $C_{\mathrm{col}}$ is sufficiently large. This is the target-only
version of the argument for \eqref{eq:score-column-probability}. Here the
design is shared across tasks, so there are only $p$ column events to combine.
A union bound gives
\begin{equation}
\mathbb P(\mathcal E_{T,\mathrm{col}}^c)\leq C_1p\exp(-C_2n_T).
\label{eq:safe-column-probability}
\end{equation}
\newpage
Conditional on the target design, each coordinate of $Z_j$ in
\eqref{eq:safe-score-row} is a linear combination of the Gaussian errors for
one task. The errors are independent across tasks and have common target
variance $\sigma_0^2$, so
\begin{equation}
\operatorname{Cov}(Z_j\mid\mathbf X^0)=\frac{\sigma_0^2\lVert\mathbf X^0_{:,j}\rVert_2^2}{n_T^2K^2}\mathbf I_K.
\label{eq:safe-score-covariance}
\end{equation}
By \eqref{eq:safe-score-covariance} its coordinates are independent and mean
zero, and on $\mathcal E_{T,\mathrm{col}}$ they satisfy
\begin{equation}
\max_{k\in[K]}\operatorname{Var}(Z_{jk}\mid\mathbf X^{0})\leq\frac{C}{n_TK^2}.
\label{eq:safe-score-variance}
\end{equation}
The chi-square tail bound of \citet[Lemma~1]{laurent2000adaptive}, applied
using \eqref{eq:safe-score-variance}, now gives, for $t>0$,
\begin{equation}
\mathbb P\left\{\lVert Z_j\rVert_2^2>\frac{C}{n_TK^2}\left(K+2\sqrt{Kt}+2t\right)\mathrel{\Big|}\mathbf X^0\right\}\leq e^{-t}\quad\text{on }\mathcal E_{T,\mathrm{col}}.
\label{eq:safe-score-chi-square}
\end{equation}
Take $t=c\log p$. Since $K\psi_K=K+\log p$ and
$2\sqrt{Kc\log p}\leq K+c\log p$, the threshold in
\eqref{eq:safe-score-chi-square} is at most a constant multiple of
$\psi_K/(Kn_T)$. Taking square roots gives
\begin{equation}
\mathbb P\left\{\lVert Z_j\rVert_2>C\sqrt{\frac{\psi_K}{Kn_T}}\mathrel{\Big|}\mathbf X^{0}\right\}\leq p^{-c}\quad\text{on }\mathcal E_{T,\mathrm{col}}.
\label{eq:safe-score-row-probability}
\end{equation}
Choose $c_T$ in \eqref{eq:safe-score-tuning} so that the threshold in
\eqref{eq:safe-score-row-probability} is at most $\lambda_T/2$.
A union bound over $j\in[p]$ gives conditional failure probability at most
$p^{1-c}$. This step does not require the score rows to be independent.
Adding the probability of $\mathcal E_{T,\mathrm{col}}^c$ from
\eqref{eq:safe-column-probability} proves \eqref{eq:safe-score-probability}.
When $p=1$, the displayed probability bound is valid trivially after taking
$C_3\geq1$, so the use of $t>0$ imposes no additional restriction.
\end{proof}

\subsection{Deterministic inequalities for the target set}
\label{app:safe-deterministic}

The two lemmas of this section hold on
$\mathcal E_{T,\mathrm{score}}(\lambda_T)\cap\mathcal E_{T,\mathrm{rsc}}$ and
involve no further randomness. The first bounds the pilot error in both the
Frobenius norm and the sum of row norms; the latter bound supplies an
allowance $\kappa_T$ large enough to cover any downward error in the pilot's
coefficient norm, which together with \eqref{eq:safe-score-event} places
$\mathbf B^0$ in $\mathcal C_T(\lambda_T,\kappa_T)$.

\begin{lemma}[Target pilot and truth feasibility]
\label{lemma:safe-target-pilot}
Suppose that $|S|\leq s$ and
$16\varphi_0s\log p/n_T\leq\alpha_0/2$. On
$\mathcal E_{T,\mathrm{score}}(\lambda_T)\cap\mathcal E_{T,\mathrm{rsc}}$,
the pilot in \eqref{eq:safe-target-pilot} satisfies
\begin{equation}
\lVert\widehat{\mathbf B}^{\mathrm{tar}}-\mathbf B^0\rVert_F\leq\frac{6K\sqrt{s}\lambda_T}{\alpha_0},\qquad \lVert\widehat{\mathbf B}^{\mathrm{tar}}-\mathbf B^0\rVert_{2,1}\leq\frac{24Ks\lambda_T}{\alpha_0}.
\label{eq:safe-pilot-bounds}
\end{equation}
Consequently, if
\begin{equation}
\kappa_T:=\frac{24Ks\lambda_T}{\alpha_0},
\label{eq:safe-kappa}
\end{equation}
then $\mathbf B^0\in\mathcal C_T(\lambda_T,\kappa_T)$.
\end{lemma}
\newpage
\begin{proof}
Let $\Delta^{\mathrm{tar}}=\widehat{\mathbf B}^{\mathrm{tar}}-\mathbf B^0$.
By optimality of the pilot, its penalized loss is no larger than the
penalized loss at $\mathbf B^0$. Expanding the quadratic loss around the
truth separates a nonnegative quadratic term from the inner product with
$\nabla\mathcal L_T(\mathbf B^0)$, whose absolute value is at most
$\lVert\nabla\mathcal L_T(\mathbf B^0)\rVert_{2,\infty}
\lVert\Delta^{\mathrm{tar}}\rVert_{2,1}$ by duality of the row norms. On
\eqref{eq:safe-score-event}, it follows that
\begin{equation}
\frac{1}{2n_TK}\lVert\mathbf X^{0}\Delta^{\mathrm{tar}}\rVert_F^2+\lambda_T\left(\lVert\mathbf B^0+\Delta^{\mathrm{tar}}\rVert_{2,1}-\lVert\mathbf B^0\rVert_{2,1}\right)\leq\frac{\lambda_T}{2}\lVert\Delta^{\mathrm{tar}}\rVert_{2,1}.
\label{eq:safe-pilot-basic}
\end{equation}
To separate errors on active and inactive rows, use
$\mathbf B^0_{S^c}=0$ and the triangle inequality on the rows in $S$:
\begin{equation}
\lVert\mathbf B^0+\Delta^{\mathrm{tar}}\rVert_{2,1}-\lVert\mathbf B^0\rVert_{2,1}\geq\lVert\Delta^{\mathrm{tar}}_{S^c}\rVert_{2,1}-\lVert\Delta^{\mathrm{tar}}_S\rVert_{2,1}.
\label{eq:safe-pilot-decomposition}
\end{equation}
Substitute \eqref{eq:safe-pilot-decomposition} into
\eqref{eq:safe-pilot-basic} and collect the two row norms to obtain
\begin{equation}
\frac{1}{2n_TK}\lVert\mathbf X^0\Delta^{\mathrm{tar}}\rVert_F^2+\frac{\lambda_T}{2}\lVert\Delta^{\mathrm{tar}}_{S^c}\rVert_{2,1}\leq\frac{3\lambda_T}{2}\lVert\Delta^{\mathrm{tar}}_S\rVert_{2,1}.
\label{eq:safe-pilot-split}
\end{equation}
Dropping the nonnegative quadratic term in \eqref{eq:safe-pilot-split} and
applying Cauchy--Schwarz to $\lVert\Delta^{\mathrm{tar}}_S\rVert_{2,1}$ yield
\begin{equation}
\lVert\Delta^{\mathrm{tar}}_{S^c}\rVert_{2,1}\leq3\lVert\Delta^{\mathrm{tar}}_S\rVert_{2,1},\qquad \lVert\Delta^{\mathrm{tar}}\rVert_{2,1}\leq4\sqrt{s}\lVert\Delta^{\mathrm{tar}}\rVert_F.
\label{eq:safe-pilot-cone}
\end{equation}
Substituting the second inequality in \eqref{eq:safe-pilot-cone} into
\eqref{eq:safe-target-curvature} bounds the tolerance term there by a multiple
of $\lVert\Delta^{\mathrm{tar}}\rVert_F^2$ and gives
\begin{equation}
\frac{1}{n_TK}\lVert\mathbf X^{0}\Delta^{\mathrm{tar}}\rVert_F^2\geq\frac{1}{K}\left(\alpha_0-\frac{16\varphi_0s\log p}{n_T}\right)\lVert\Delta^{\mathrm{tar}}\rVert_F^2\geq\frac{\alpha_0}{2K}\lVert\Delta^{\mathrm{tar}}\rVert_F^2.
\label{eq:safe-pilot-curvature}
\end{equation}
The last inequality in \eqref{eq:safe-pilot-curvature} uses the stated target
sparsity condition. In the other direction, dropping the inactive-row term in
\eqref{eq:safe-pilot-split} and multiplying by two gives
\begin{equation}
\frac{1}{n_TK}\lVert\mathbf X^{0}\Delta^{\mathrm{tar}}\rVert_F^2\leq3\lambda_T\lVert\Delta^{\mathrm{tar}}_S\rVert_{2,1}\leq3\sqrt{s}\lambda_T\lVert\Delta^{\mathrm{tar}}\rVert_F.
\label{eq:safe-pilot-upper}
\end{equation}
If $\Delta^{\mathrm{tar}}=0$, the bounds in \eqref{eq:safe-pilot-bounds} are
immediate. Otherwise, combining \eqref{eq:safe-pilot-curvature} with
\eqref{eq:safe-pilot-upper} and dividing by
$\lVert\Delta^{\mathrm{tar}}\rVert_F$ proves the Frobenius bound in
\eqref{eq:safe-pilot-bounds}, and multiplication by $4\sqrt{s}$ through
\eqref{eq:safe-pilot-cone} gives the mixed-norm bound. This reasoning also
covers $s=0$, because \eqref{eq:safe-pilot-cone} then forces
$\Delta^{\mathrm{tar}}=0$.

The score constraint for $\mathbf B^0$ follows from
\eqref{eq:safe-score-event}. By \eqref{eq:safe-pilot-bounds} and
\eqref{eq:safe-kappa},
\begin{equation}
\lVert\mathbf B^0\rVert_{2,1}\leq\lVert\widehat{\mathbf B}^{\mathrm{tar}}\rVert_{2,1}+\lVert\widehat{\mathbf B}^{\mathrm{tar}}-\mathbf B^0\rVert_{2,1}\leq\lVert\widehat{\mathbf B}^{\mathrm{tar}}\rVert_{2,1}+\kappa_T.
\label{eq:safe-truth-radius}
\end{equation}
Thus both restrictions in \eqref{eq:safe-confidence-set} hold for
$\mathbf B^0$.
\end{proof}

Truth feasibility alone does not bound the error of every point of the set.
The second lemma adds the target curvature \eqref{eq:safe-target-curvature} to
the two restrictions in \eqref{eq:safe-confidence-set} and bounds the radius of
the whole set. Its conclusion holds simultaneously for every feasible matrix,
including matrices chosen using the same target observations.
\newpage
\begin{lemma}[Target confidence radius]
\label{lemma:safe-confidence-radius}
Assume the conditions of Lemma~\ref{lemma:safe-target-pilot}, set $\kappa_T$ as
in \eqref{eq:safe-kappa}, and suppose that
$s\log p/n_T\leq c_{\mathrm{sp}}$, where
$c_{\mathrm{sp}}>0$ is sufficiently small, depends only on
$\alpha_0$ and $\varphi_0$, and satisfies
$16\varphi_0c_{\mathrm{sp}}\leq\alpha_0/2$. On
$\mathcal E_{T,\mathrm{score}}(\lambda_T)\cap\mathcal E_{T,\mathrm{rsc}}$,
every $\mathbf B\in\mathcal C_T(\lambda_T,\kappa_T)$ satisfies
\begin{equation}
\lVert\mathbf B-\mathbf B^0\rVert_F\leq C_{\mathrm{safe}}K\sqrt{s}\lambda_T,
\label{eq:safe-confidence-radius}
\end{equation}
where $C_{\mathrm{safe}}$ depends only on $\alpha_0$ and $\varphi_0$.
\end{lemma}

\begin{proof}
Fix $\mathbf B\in\mathcal C_T(\lambda_T,\kappa_T)$ and set
$\Delta=\mathbf B-\mathbf B^0$. If $s=0$, Lemma~\ref{lemma:safe-target-pilot}
gives $\widehat{\mathbf B}^{\mathrm{tar}}=\mathbf B^0=\mathbf0$, while
\eqref{eq:safe-kappa} gives $\kappa_T=0$. The mixed-norm restriction in
\eqref{eq:safe-confidence-set} then forces $\mathbf B=\mathbf0$, so
\eqref{eq:safe-confidence-radius} is immediate. Hence, suppose $s\geq1$.

The mixed-norm restriction in \eqref{eq:safe-confidence-set}, combined with
\eqref{eq:safe-pilot-bounds} and \eqref{eq:safe-kappa}, gives
\begin{equation}
\lVert\mathbf B\rVert_{2,1}\leq\lVert\widehat{\mathbf B}^{\mathrm{tar}}\rVert_{2,1}+\kappa_T\leq\lVert\mathbf B^0\rVert_{2,1}+2\kappa_T.
\label{eq:safe-feasible-norm}
\end{equation}
One copy of $\kappa_T$ in \eqref{eq:safe-feasible-norm} comes from the
definition of the set, and the other from the pilot's possible overestimation
of the true norm. Conversely, rowwise triangle inequalities give
$\lVert\mathbf B^0+\Delta\rVert_{2,1}\geq
\lVert\mathbf B^0\rVert_{2,1}-\lVert\Delta_S\rVert_{2,1}
+\lVert\Delta_{S^c}\rVert_{2,1}$. Combining this inequality with
\eqref{eq:safe-feasible-norm}, then using
$\lVert\Delta_S\rVert_{2,1}\leq\sqrt{s}\lVert\Delta\rVert_F$, yields
\begin{equation}
\lVert\Delta_{S^c}\rVert_{2,1}\leq\lVert\Delta_S\rVert_{2,1}+2\kappa_T,\qquad \lVert\Delta\rVert_{2,1}\leq2\sqrt{s}\lVert\Delta\rVert_F+2\kappa_T.
\label{eq:safe-feasible-cone}
\end{equation}
Inequality \eqref{eq:safe-feasible-cone} is an approximate version of the pilot
cone inequality \eqref{eq:safe-pilot-cone}: the feasible matrix may have
additional error on inactive rows, but that error is limited by $2\kappa_T$.
Turning to the score restriction, the loss is quadratic, so its gradient
difference equals $(n_TK)^{-1}(\mathbf X^0)^\top\mathbf X^0\Delta$, while the
score norms at $\mathbf B$ and $\mathbf B^0$ are at most $\lambda_T$ and
$\lambda_T/2$, respectively. Taking the inner product with $\Delta$ and
applying duality therefore gives
\begin{equation}
\frac{1}{n_TK}\lVert\mathbf X^{0}\Delta\rVert_F^2=\left\langle\nabla\mathcal L_T(\mathbf B)-\nabla\mathcal L_T(\mathbf B^0),\Delta\right\rangle\leq\frac{3\lambda_T}{2}\lVert\Delta\rVert_{2,1}.
\label{eq:safe-gradient-upper}
\end{equation}
The curvature bound \eqref{eq:safe-target-curvature} controls the same
quantity from below. To handle its tolerance term, square the second
inequality in \eqref{eq:safe-feasible-cone} and use
$(u+v)^2\leq2u^2+2v^2$ to get
$\lVert\Delta\rVert_{2,1}^2\leq
8s\lVert\Delta\rVert_F^2+8\kappa_T^2$. Substitution gives
\begin{equation}
\frac{1}{n_TK}\lVert\mathbf X^{0}\Delta\rVert_F^2\geq\left(\frac{\alpha_0}{K}-\frac{8\varphi_0s\log p}{Kn_T}\right)\lVert\Delta\rVert_F^2-\frac{8\varphi_0\log p}{Kn_T}\kappa_T^2.
\label{eq:safe-gradient-lower}
\end{equation}
The condition $s\log p/n_T\leq c_{\mathrm{sp}}$ ensures that the
coefficient of $\lVert\Delta\rVert_F^2$ in
\eqref{eq:safe-gradient-lower} is at least $\alpha_0/(2K)$.
Combining this lower bound with \eqref{eq:safe-gradient-upper}, and using
\eqref{eq:safe-feasible-cone} on its right-hand side, yields
\begin{equation}
\frac{\alpha_0}{2K}\lVert\Delta\rVert_F^2\leq3\sqrt{s}\lambda_T\lVert\Delta\rVert_F+3\lambda_T\kappa_T+\frac{8\varphi_0\log p}{Kn_T}\kappa_T^2.
\label{eq:safe-radius-quadratic}
\end{equation}
The scale suggested by \eqref{eq:safe-radius-quadratic} is
$K\sqrt{s}\lambda_T$, the same scale as the pilot error in
\eqref{eq:safe-pilot-bounds}. To verify that its multiplier is bounded, set
$r=\lVert\Delta\rVert_F/(K\sqrt{s}\lambda_T)$. Substitution of
\eqref{eq:safe-kappa} into \eqref{eq:safe-radius-quadratic} and division by
$Ks\lambda_T^2$ give
\begin{equation}
\frac{\alpha_0}{2}r^2\leq3r+\frac{72}{\alpha_0}+\frac{4608\varphi_0}{\alpha_0^2}\frac{s\log p}{n_T}.
\label{eq:safe-radius-normalized}
\end{equation}
Define
\newpage
\begin{equation}
A_{\mathrm{safe}}:=\frac{72}{\alpha_0}+\frac{4608\varphi_0c_{\mathrm{sp}}}{\alpha_0^2},\qquad C_{\mathrm{safe}}:=\frac{3+\sqrt{9+2\alpha_0A_{\mathrm{safe}}}}{\alpha_0}.
\label{eq:safe-radius-constants}
\end{equation}
Since $s\log p/n_T\leq c_{\mathrm{sp}}$, solving
\eqref{eq:safe-radius-normalized} with $C_{\mathrm{safe}}$ as in
\eqref{eq:safe-radius-constants} gives $r\leq C_{\mathrm{safe}}$, which is
\eqref{eq:safe-confidence-radius}. All inequalities above hold for every
feasible $\mathbf B$ on the same score and curvature event. The result
therefore controls the whole random set, rather than one matrix fixed before
seeing the data.
\end{proof}

\subsection{Proof of Theorem~\ref{thm:safe-projection}}
\label{app:safe-projection-proof}

\begin{proof}
Assumption~\ref{assumption_design} is the common-design special case of
Assumption~\ref{assumption_task_specific_design}, so
Lemma~\ref{lemma:safe-target-score} and the target component of
Lemma~\ref{lemma_rscrsm} imply
\begin{equation}
\mathbb P\left\{\mathcal E_{T,\mathrm{score}}(\lambda_T)\cap\mathcal E_{T,\mathrm{rsc}}\right\}\geq1-C_1Kp\exp(-C_2n_T)-C_3p^{1-c}.
\label{eq:safe-event-probability}
\end{equation}
On this event, Lemma~\ref{lemma:safe-target-pilot} gives
$\mathbf B^0\in\mathcal C_T(\lambda_T,\kappa_T)$, and
Lemma~\ref{lemma:safe-confidence-radius} gives
\begin{equation}
\frac{1}{\sqrt K}\lVert\widehat{\mathbf W}^{\mathrm{safe}}-\mathbf B^0\rVert_F\leq C_{\mathrm{safe}}\sqrt{Ks}\lambda_T=C_{\mathrm{safe}}c_T\sqrt{\frac{s\psi_K}{n_T}}.
\label{eq:safe-target-branch}
\end{equation}
It remains to compare the projection with the original candidate. The
first-order characterization of a projection onto a closed convex set is
standard; see \citep[Sections~1.2 and~2.3]{parikh2014proximal}. Here it follows
directly by considering the line segment from
$\widehat{\mathbf W}^{\mathrm{safe}}$ toward $\mathbf B^0$, which lies in the
set because both endpoints are feasible. The squared distance to
$\widehat{\mathbf W}$ is minimized at the projected endpoint, so its
one-sided derivative along this segment is nonnegative. Equivalently,
\begin{equation}
\left\langle\widehat{\mathbf W}-\widehat{\mathbf W}^{\mathrm{safe}},\mathbf B^0-\widehat{\mathbf W}^{\mathrm{safe}}\right\rangle\leq0.
\label{eq:safe-projection-optimality}
\end{equation}
Expand the squared norm of
$\widehat{\mathbf W}-\mathbf B^0=
(\widehat{\mathbf W}-\widehat{\mathbf W}^{\mathrm{safe}})
+(\widehat{\mathbf W}^{\mathrm{safe}}-\mathbf B^0)$.
By \eqref{eq:safe-projection-optimality}, the cross term in this expansion
is nonnegative, and hence
\begin{equation}
\lVert\widehat{\mathbf W}^{\mathrm{safe}}-\mathbf B^0\rVert_F^2+\lVert\widehat{\mathbf W}^{\mathrm{safe}}-\widehat{\mathbf W}\rVert_F^2\leq\lVert\widehat{\mathbf W}-\mathbf B^0\rVert_F^2.
\label{eq:safe-pythagorean}
\end{equation}
Dropping the second term on the left of \eqref{eq:safe-pythagorean} shows
that projection cannot increase coefficient error. Taking its square root,
normalizing by $\sqrt K$, and combining with \eqref{eq:safe-target-branch}
proves \eqref{eq:safe-oracle-bound}. The event in
\eqref{eq:safe-event-probability} concerns only the target set and the truth.
Once it holds, the radius and projection arguments apply simultaneously to
every candidate $\widehat{\mathbf W}$. In particular, the candidate may have
been fitted or selected using these same target observations; no independence
or additional sample splitting is needed for this projection conclusion.
\end{proof}

\subsection{Proof of Proposition~\ref{prop:safe-tmtl-rates}}
\label{app:safe-rates-proof}

\begin{proof}
Apply \eqref{eq:safe-oracle-bound} first with
$\widehat{\mathbf W}=\widehat{\mathbf W}^F$ and then with
$\widehat{\mathbf W}=\widehat{\mathbf W}^D$. On the intersection with the
event of Theorem~\ref{thm:fused}, substitution of \eqref{rate_fused} gives
\eqref{eq:safe-fused-rate}. On the intersection with the event of
Theorem~\ref{thm:debiased}, substitution of \eqref{rate_debiased} gives
\eqref{eq:safe-debiased-rate}. In each case, a union bound between the
corresponding theorem event and \eqref{eq:safe-event-probability} gives the
probability in \eqref{eq:safe-fused-probability}, up to a change in its
constants: the additional terms $Kp\exp(-C_2n_T)$ and $p^{1-c}$ are absorbed by the corresponding terms multiplied by $L+1$. This union bound does not require independence of the events.
\end{proof}

\subsection{Calibration and a possible plug-in}
\label{app:safe-plugin}

The theoretical value of $\kappa_T$ in \eqref{eq:safe-tuning} depends on the
sparsity level $s$ and on the target curvature constant $\alpha_0$, neither of
which is known. A group-norm plug-in could set $\widehat s(\lambda_T)$ equal to
the active-row count of the target pilot and use
$\kappa_T=C_\kappa K\widehat s(\lambda_T)\lambda_T$, with $\lambda_T$ and
$C_\kappa$ selected by target validation. This follows the plug-in principle,
but not the exact implementation, of \citet{he2026parameter}: they set
$\kappa_T=\widehat s(\lambda_T)\lambda_T$ under their scalar normalization,
screen unconstrained joint candidates for feasibility, select among feasible
candidates by target validation, and fall back to the target-only estimator if
none is feasible. Such a data-driven radius is not covered by
Lemma~\ref{lemma:safe-target-pilot} or Theorem~\ref{thm:safe-projection}.
Neither result shows that an estimated radius keeps $\mathbf B^0$ feasible, and
feasibility of the truth is what the target-only branch in
\eqref{eq:safe-oracle-bound} rests on.

%% file: supp_simulation.tex
\section{Additional simulation details and results} \label{app:simulation_p100}
We first record the test-sample and method-specific tuning details of the
simulation study in Section~\ref{sec:simulation_p200}, whose results are
reported there in
Tables~\ref{tab:simulationresult_frobnormerror}--\ref{tab:simulationresult_negativetransfer}.
We then examine whether the same patterns persist in a smaller problem where
all penalty parameters can be tuned over a full Cartesian grid.

\subsection{Evaluation and method-specific tuning}
\label{app:simulation-solver}
{\emergencystretch=1.5em
Each replication uses an independent target test sample of size
$n_{\mathrm{test}}=100$. The estimators, structured source weights, and
candidate grid are described in Section~\ref{sec:simulation-methods}.
The taskwise \textbf{TSTL(Debiased)} correction uses its own cross-validation
procedure, and \textbf{MTL(Full)} selects its penalty using pooled target and
source observations. The correction stages are tuned with the selected fused
fits held fixed.
\par}

\FloatBarrier
\subsection{Design of the second study}
\label{sec:simulation-design-p100}
The second study fixes $p=100$ and $K=50$ and varies the number of source
domains over $L\in\{1,2,3\}$. These smaller problems allow us to tune each
fused penalty separately, rather than restricting the penalties to a
structured weight vector. We also record the ADMM iterations needed by
\textbf{TMTL(Fused)} to meet its stopping criterion. Because the dimensions
and source configurations differ from the main study, this comparison checks
whether its qualitative findings persist under full-grid tuning; it does not
isolate the effect of the structured weights.

The target and error distributions, the sparsity level $s=10$, the coefficient
magnitudes, and the mild covariate shift in the source designs are as in
Section~\ref{sec:simulation-design}, with $n_T=100$, $n_S=500$, $p=100$, and
$K=50$. The contrast matrix $\Delta_\alpha$ is generated as in
Section~\ref{sec:simulation-design}. The source coefficient matrices are
\newpage
\begin{itemize}
    \item $L=1$: $\mathbf B^1=\mathbf B^0+\Delta_\alpha$, with $\alpha\in\{0,1/5,1/2,1\}$.
    \item $L=2$, balanced shifts: $\mathbf B^1=\mathbf B^0+\Delta_{1/2}$ and $\mathbf B^2=\mathbf B^0-\Delta_{1/2}$.
    \item $L=2$, aligned shifts: $\mathbf B^1=\mathbf B^2=\mathbf B^0+\Delta_{1/2}$.
    \item $L=3$, balanced shifts:
    $\mathbf B^1=\mathbf B^2=\mathbf B^0+(3/4)\Delta_{1/2}$ and
    $\mathbf B^3=\mathbf B^0-(3/2)\Delta_{1/2}$.
    \item $L=3$, aligned shifts: $\mathbf B^1=\mathbf B^2=\mathbf B^3=\mathbf B^0+\Delta_{1/2}$.
\end{itemize}
The single-source case adds a setting with $\alpha=0$, in which the source and
target coefficients coincide. The balanced shifts have $\epsilon_D=0$, whereas the aligned shifts have
positive directional bias. 
For each setting, the coefficient matrices are
held fixed and the covariates and responses are regenerated over 100 Monte
Carlo replications, as in Section~\ref{sec:simulation-design}.

\subsection{Tuning in the second study}
\label{sec:simulation-methods-p100}
We compare the same six estimators as in
Section~\ref{sec:simulation-methods}. The fused penalties
$\lambda_0,\ldots,\lambda_L$ and the correction penalty
$\widetilde\lambda$ are selected jointly by three-fold cross-validation
over a Cartesian grid. For $L=1$, each grid contains seven logarithmically
spaced values from $10^{-4}$ to $1$; for $L=2,3$, it contains five values
from $10^{-4}$ to $10^{-1}$. Coarser grids keep the number of fits manageable
as $L$ grows, illustrating the computational motivation for the structured
weights in the main study.

We report log coefficient error, test RMSE, and negative-transfer frequency,
as defined in Section~\ref{sec:simulation-metric}. Dispersions are standard
deviations across replications.

\subsection{Results of the second study}
\label{sec:simulation-results-p100}
Full-grid tuning preserves the main finding: fusion is effective when source
shifts cancel, while correction is valuable when they align
(Tables~\ref{tab:simulationresult_frobnormerror_p100}--\ref{tab:simulationresult_negativetransfer_p100}).
With balanced shifts, the two TMTL estimators have nearly identical errors
and outperform the other methods in both coefficient estimation and
prediction. The $L=2$ and $L=3$ settings use different contrast configurations,
so their comparison should not be read as the effect of adding one source.

For a single source that coincides with, or is close to, the target
($\alpha=0$ or $1/5$), correction adds little and complete pooling is
competitive. As the discrepancy grows, \textbf{TMTL(Debiased)} becomes the
most accurate method. At $\alpha=1/2$ and $\alpha=1$, and in both aligned
multiple-source settings, it is the only source-based estimator with no
observed negative transfer; every other source-based method performs worse
than \textbf{MTL(Target)} in all replications. Thus, the advantage of target
correction under aligned shifts persists when the penalties are tuned freely.

The mean ADMM stopping count ranges from about 200 to 350 iterations
(Table~\ref{tab:iteration_fused_p100}). These counts describe empirical
computational effort and do not provide a finite-iteration accuracy guarantee.

\begin{table}[!htbp]
\caption{Log coefficient-estimation error \eqref{eq:coeferr} at $p=100$ and
$K=50$. Entries give the mean (standard deviation) over 100 replications,
rounded to two decimal places. Smaller values indicate better estimation;
boldface marks the smallest mean before rounding, including ties.}
\label{tab:simulationresult_frobnormerror_p100}
\centering
\small
\renewcommand{\arraystretch}{1.15}
\setlength{\tabcolsep}{3pt}
\begin{tabular}{llcccccc}
\toprule
\multirow{2}{*}{\textbf{L}} &
\multirow{2}{*}{\textbf{Shifts}} &
\multicolumn{2}{c}{\textbf{TMTL}} &
\multicolumn{2}{c}{\textbf{MTL}} &
\multicolumn{2}{c}{\textbf{TSTL}} \\
\cmidrule(lr){3-4}\cmidrule(lr){5-6}\cmidrule(lr){7-8}
& &
\textbf{Fused} & \textbf{Debiased} &
\textbf{Target} & \textbf{Full} &
\textbf{Fused} & \textbf{Debiased} \\
\midrule
\multirow{4}{*}{1}
& $\alpha=0$   & -5.65 (0.05) & -5.65 (0.05) & -4.11 (0.10) & \textbf{-5.67 (0.05)} & -4.83 (0.07) & -4.80 (0.08) \\
& $\alpha=1/5$ & -5.15 (0.04) & -5.15 (0.04) & -4.11 (0.10) & \textbf{-5.15 (0.04)} & -4.43 (0.06) & -4.44 (0.05) \\
& $\alpha=1/2$ & -3.41 (0.04) & \textbf{-4.38 (0.06)} & -4.11 (0.10) & -3.40 (0.04) & -3.39 (0.03) & -3.84 (0.05) \\
& $\alpha=1$   & -2.41 (0.04) & \textbf{-4.31 (0.07)} & -4.11 (0.10) & -1.91 (0.03) & -2.20 (0.02) & -3.48 (0.05) \\
\midrule
\multirow{2}{*}{2}
& Balanced  & \textbf{-6.14 (0.07)} & \textbf{-6.14 (0.07)} & -4.18 (0.14) & -5.09 (0.13) & -5.22 (0.06) & -5.14 (0.07) \\
& Aligned   & -3.35 (0.02) & \textbf{-4.63 (0.06)} & -4.18 (0.14) & -2.99 (0.01) & -3.33 (0.02) & -3.68 (0.04) \\
\midrule
\multirow{2}{*}{3}
& Balanced  & \textbf{-6.63 (0.08)} & \textbf{-6.63 (0.08)} & -4.18 (0.14) & -5.45 (0.30) & -5.44 (0.07) & -5.33 (0.08) \\
& Aligned   & -3.30 (0.02) & \textbf{-4.62 (0.07)} & -4.18 (0.14) & -2.95 (0.01) & -3.35 (0.02) & -3.67 (0.03) \\
\bottomrule
\end{tabular}

\end{table}
\begin{table}[!htbp]
\caption{Target test RMSE \eqref{eq:predrmse} at $p=100$ and
$K=50$. Entries give the mean (standard deviation) over 100 replications,
rounded to two decimal places. Smaller values indicate better prediction;
boldface marks the smallest mean before rounding, including ties.}
\label{tab:simulationresult_predictionerror_p100}
\centering
\small
\renewcommand{\arraystretch}{1.15}
\setlength{\tabcolsep}{3pt}
\begin{tabular}{llcccccc}
\toprule
\multirow{2}{*}{\textbf{L}} &
\multirow{2}{*}{\textbf{Shifts}} &
\multicolumn{2}{c}{\textbf{TMTL}} &
\multicolumn{2}{c}{\textbf{MTL}} &
\multicolumn{2}{c}{\textbf{TSTL}} \\
\cmidrule(lr){3-4}\cmidrule(lr){5-6}\cmidrule(lr){7-8}
& &
\textbf{Fused} & \textbf{Debiased} &
\textbf{Target} & \textbf{Full} &
\textbf{Fused} & \textbf{Debiased} \\
\midrule
\multirow{4}{*}{1}
& $\alpha=0$   & 0.06 (0.00) & 0.06 (0.00) & 0.13 (0.01) & \textbf{0.06 (0.00)} & 0.09 (0.00) & 0.09 (0.00) \\
& $\alpha=1/5$ & 0.08 (0.00) & 0.08 (0.00) & 0.13 (0.01) & \textbf{0.08 (0.00)} & 0.11 (0.00) & 0.11 (0.00) \\
& $\alpha=1/2$ & 0.18 (0.01) & \textbf{0.11 (0.00)} & 0.13 (0.01) & 0.18 (0.01) & 0.18 (0.01) & 0.15 (0.00) \\
& $\alpha=1$   & 0.30 (0.01) & \textbf{0.12 (0.01)} & 0.13 (0.01) & 0.38 (0.02) & 0.33 (0.01) & 0.18 (0.01) \\
\midrule
\multirow{2}{*}{2}
& Balanced  & \textbf{0.05 (0.00)} & \textbf{0.05 (0.00)} & 0.12 (0.01) & 0.08 (0.01) & 0.07 (0.00) & 0.08 (0.00) \\
& Aligned   & 0.19 (0.01) & \textbf{0.10 (0.00)} & 0.12 (0.01) & 0.22 (0.01) & 0.19 (0.01) & 0.16 (0.00) \\
\midrule
\multirow{2}{*}{3}
& Balanced  & \textbf{0.04 (0.00)} & \textbf{0.04 (0.00)} & 0.12 (0.01) & 0.07 (0.01) & 0.07 (0.00) & 0.07 (0.00) \\
& Aligned   & 0.19 (0.01) & \textbf{0.10 (0.00)} & 0.12 (0.01) & 0.23 (0.01) & 0.19 (0.01) & 0.16 (0.00) \\
\bottomrule
\end{tabular}

\end{table}
\begin{table}[!htbp]
\caption{Negative-transfer frequency \eqref{eq:negative_transfer} at $p=100$
and $K=50$: the proportion of 100 replications in which coefficient error
exceeds that of \textbf{MTL(Target)}. Boldface marks the smallest frequency,
except in rows where all methods tie. The target-only reference is omitted.}
\label{tab:simulationresult_negativetransfer_p100}
\centering
\small
\renewcommand{\arraystretch}{1.15}
\setlength{\tabcolsep}{5pt}
\begin{tabular}{llccccc}
\toprule
\multirow{2}{*}{\textbf{L}} &
\multirow{2}{*}{\textbf{Shifts}} &
\multicolumn{2}{c}{\textbf{TMTL}} &
\multicolumn{1}{c}{\textbf{MTL}} &
\multicolumn{2}{c}{\textbf{TSTL}} \\
\cmidrule(lr){3-4}\cmidrule(lr){5-5}\cmidrule(lr){6-7}
& &
\textbf{Fused} & \textbf{Debiased} &
\textbf{Full} &
\textbf{Fused} & \textbf{Debiased} \\
\midrule
\multirow{4}{*}{1}
& $\alpha=0$   & 0.00 & 0.00 & 0.00 & 0.00 & 0.00 \\
& $\alpha=1/5$ & 0.00 & 0.00 & 0.00 & 0.00 & 0.00 \\
& $\alpha=1/2$ & 1.00 & \textbf{0.00} & 1.00 & 1.00 & 1.00 \\
& $\alpha=1$   & 1.00 & \textbf{0.00} & 1.00 & 1.00 & 1.00 \\
\midrule
\multirow{2}{*}{2}
& Balanced  & 0.00 & 0.00 & 0.00 & 0.00 & 0.00 \\
& Aligned   & 1.00 & \textbf{0.00} & 1.00 & 1.00 & 1.00 \\
\midrule
\multirow{2}{*}{3}
& Balanced  & 0.00 & 0.00 & 0.00 & 0.00 & 0.00 \\
& Aligned   & 1.00 & \textbf{0.00} & 1.00 & 1.00 & 1.00 \\
\bottomrule
\end{tabular}

\end{table}
\begin{table}[!htbp]
\caption{ADMM iterations required by \textbf{TMTL(Fused)} to meet the stopping
criterion at $p=100$ and $K=50$, using the hyperparameters selected within
each replication. Entries report the mean (standard deviation).}
\label{tab:iteration_fused_p100}
\centering
\small
\renewcommand{\arraystretch}{1.15}
\setlength{\tabcolsep}{12pt}
\begin{tabular}{llc}
\toprule
$L$ & Shifts & Iterations \\
\midrule
\multirow{4}{*}{1}
& $\alpha=0$   & 203.30 (45.69) \\
& $\alpha=1/5$ & 213.45 (69.14) \\
& $\alpha=1/2$ & 221.90 (63.07) \\
& $\alpha=1$   & 275.55 (55.18) \\
\midrule
\multirow{2}{*}{2}
& Balanced & 309.20 (28.54) \\
& Aligned  & 306.35 (35.60) \\
\midrule
\multirow{2}{*}{3}
& Balanced & 348.35 (76.90) \\
& Aligned  & 294.00 (62.34) \\
\bottomrule
\end{tabular}

\end{table}

%% file: supp_realdata.tex
\section{Details of the single-cell application}
\label{app:realdata}

\subsection{Preprocessing}
\label{app:realdata-preprocessing}
We preprocess the data described in Section~\ref{sec:realdata-data} of the
main paper separately within each sequencing site. We first discard cell
types with fewer than 100 cells and retain genes with nonzero counts in at
least 30\% of the remaining cells. We then remove cells whose total RNA or
protein counts lie below the 0.5th percentile or above the 99.5th percentile
of the respective empirical distribution. For the retained cells, we
normalize gene counts by library size and apply a logarithmic transformation.
Writing $c_{ig}$ for the raw count of gene $g$ in cell $i$, we set
\begin{equation}
\label{eq:lognorm}
x_{ig}=\log\left(1+10^4\,\frac{c_{ig}}{\sum_{g'}c_{ig'}}\right),
\end{equation}
as implemented by the \texttt{LogNormalize} method of \texttt{NormalizeData}
in \texttt{Seurat} \citep{hao2024dictionary}, and standardize each resulting
feature across cells.

We then retain up to 2500 genes with the highest variability as predictors.
At each site, the detection filter leaves fewer than
2500 candidates, so all remain eligible. To reduce collinearity, whenever the
correlation between two candidate genes exceeds $0.8$, we remove the gene
with lower marginal variability.

For the protein responses, we apply the centered log-ratio (CLR)
transformation \citep{luecken2021sandbox, heumos2023best}. For protein $k$ in
cell $i$ from domain $\ell$, this gives
\[
y_i^{\ell k}
\leftarrow
\log
\left\{
1+
\frac{y_i^{\ell k}}
{\left(\prod_{q=1}^{K_{\mathrm{CLR}}}(1+y_i^{\ell q})\right)^{1/K_{\mathrm{CLR}}}}
\right\},
\]
where $K_{\mathrm{CLR}}$ is the number of proteins in the normalization
panel. We then retain proteins whose maximal absolute correlation with a
retained predictor is at least 0.2, focusing the evaluation on proteins with
detectable marginal associations with RNA predictors.

Because CLR normalization does not generally center each protein across
cells, we handle intercepts separately within each domain and random split.
We compute the column means of the scaled predictors and CLR responses using
only that domain's training cells, then subtract these means from its
training, validation, and test observations. All six methods are fitted to
the centered variables. For target prediction, we add the target
training-response means back to the fitted responses. This is equivalent to
fitting unpenalized domain- and protein-specific intercepts and prevents
marginal differences in protein abundance between cell types from being
absorbed into the slope estimates.

\subsection{Implementation settings}
\label{app:realdata-implementation}
The methods and tuning grids follow Section~\ref{sec:realdata-methods} of the
main paper. For \textbf{TMTL(Debiased)}, we tune the correction with the
selected fused fit held fixed. At all four sites, the taskwise
\textbf{TSTL(Debiased)} correction is tuned and refitted using the combined
target training and validation observations, with the fused fit held fixed.
Target test cells remain held out throughout.

{\emergencystretch=1.5em
For the fused TMTL fit, ADMM stops when the maximum Frobenius norm among the
constraint-specific primal residuals and dual residual contributions in
\eqref{eq:admm-residuals} falls below $\sqrt{K}\times10^{-4}$, or after 1000
iterations. For the multitask correction and both MTL baselines, ADMM stops
when the Frobenius norms of both the primal and dual residuals fall below the
same threshold, or after 1000 iterations. The single-task baselines are
solved separately for each of the $K$ responses;
\textbf{TSTL(Fused)} uses a tolerance of $5\times10^{-3}$ and a cap of
200 iterations.
\par}

\subsection{Complete prediction results}
\label{app:realdata-results}
Table~\ref{tab:realdatarmse} gives the mean relative MSE and its standard
deviation over 20 random splits for all six methods and all 48 target
domains, grouped by sequencing site. Values below one indicate improvement
over the intercept-only model, and smaller values indicate better
prediction.

The benefit of target correction varies across sites. At Sites~1 and~4,
\textbf{TMTL(Debiased)} improves on \textbf{TMTL(Fused)} for every target
domain; at Sites~2 and~3, correction makes little difference in several
domains. Correction does not make transfer uniformly preferable to
other approaches: target-only fitting is better in some domains at
Sites~1--3, and complete pooling outperforms \textbf{TMTL(Debiased)} in four
domains at Site~4. These differences complement the comparisons with
target-only fitting and fusion in Figure~\ref{fig:realdata-benefits} of the
main paper.

\begin{table}[!htbp]
\caption{Relative MSE $Q_{\text{method},r}^{(\ell)}$ \eqref{eq:realdata-relative-mse} for target domains at Site~1. The four parts of this table report all 48 domains. Entries are means (standard deviations) over 20 random splits, rounded to two decimal places. Targets are ordered by increasing $n_\ell$, the total number of cells in the domain. Boldface marks the lowest mean in each row, including ties, before rounding to two decimal places.}
\label{tab:realdatarmse}
\centering

\small
\renewcommand{\arraystretch}{1.1}
\setlength{\tabcolsep}{1.7pt}
\begin{tabular}{>{\raggedright\arraybackslash}p{0.21\textwidth}rcccccc}
\toprule
\multirow{2}{*}{\textbf{Target}} &
\multirow{2}{*}{$n_\ell$} &
\multicolumn{2}{c}{\textbf{TMTL}} &
\multicolumn{2}{c}{\textbf{MTL}} &
\multicolumn{2}{c}{\textbf{TSTL}} \\
\cmidrule(lr){3-4}\cmidrule(lr){5-6}\cmidrule(lr){7-8}
& & \textbf{Fused} & \textbf{Debiased} & \textbf{Target} &
\textbf{Full} & \textbf{Fused} & \textbf{Debiased} \\
\midrule
CD8$^+$ T CD57$^+$ CD45RO$^+$ & 111 & 0.99 (0.01) & \textbf{0.93 (0.01)} & 0.94 (0.01) & 1.00 (0.01) & 1.01 (0.02) & 0.98 (0.03) \\
CD8$^+$ T TIGIT$^+$ CD45RO$^+$ & 115 & 0.93 (0.02) & \textbf{0.89 (0.03)} & 0.93 (0.02) & 0.93 (0.02) & 0.95 (0.01) & 0.95 (0.03) \\
T reg & 131 & 0.87 (0.02) & \textbf{0.83 (0.03)} & 0.88 (0.02) & 0.88 (0.02) & 0.90 (0.02) & 0.88 (0.02) \\
CD8$^+$ T CD69$^+$ CD45RA$^+$ & 134 & 0.95 (0.01) & \textbf{0.92 (0.02)} & 0.94 (0.02) & 0.93 (0.01) & 0.97 (0.01) & 0.96 (0.02) \\
CD8$^+$ T CD49f$^+$ & 144 & 0.82 (0.02) & \textbf{0.81 (0.02)} & 0.89 (0.02) & 0.82 (0.02) & 0.87 (0.02) & 0.86 (0.02) \\
gdT CD158b$^+$ & 214 & 0.88 (0.01) & 0.73 (0.02) & \textbf{0.72 (0.02)} & 0.87 (0.02) & 0.91 (0.01) & 0.75 (0.02) \\
MAIT & 232 & 0.94 (0.01) & \textbf{0.91 (0.02)} & 0.95 (0.02) & 0.94 (0.01) & 0.96 (0.01) & 0.92 (0.01) \\
CD4$^+$ T activated integrinB7$^+$ & 375 & 0.74 (0.01) & \textbf{0.68 (0.01)} & 0.72 (0.02) & 0.71 (0.01) & 0.78 (0.01) & 0.73 (0.02) \\
CD8$^+$ T CD57$^+$ CD45RA$^+$ & 388 & 0.90 (0.01) & \textbf{0.78 (0.02)} & 0.79 (0.02) & 0.91 (0.01) & 0.92 (0.01) & 0.83 (0.01) \\
CD8$^+$ T naive & 1201 & 0.90 (0.01) & 0.90 (0.01) & \textbf{0.89 (0.01)} & 0.91 (0.01) & 0.92 (0.01) & 0.90 (0.01) \\
CD4$^+$ T naive & 1697 & 0.87 (0.01) & \textbf{0.84 (0.01)} & 0.84 (0.01) & 0.89 (0.01) & 0.88 (0.01) & 0.86 (0.01) \\
CD4$^+$ T activated & 1705 & 0.67 (0.01) & 0.62 (0.01) & \textbf{0.61 (0.01)} & 0.65 (0.01) & 0.70 (0.01) & 0.64 (0.01) \\
\bottomrule
\end{tabular}

\end{table}

\clearpage
\begin{table}[!t]
\centering
\ContinuedFloat
\caption{Relative MSE across target domains (continued): Site~2.}
\small
\renewcommand{\arraystretch}{1.1}
\setlength{\tabcolsep}{1.7pt}
\begin{tabular}{>{\raggedright\arraybackslash}p{0.21\textwidth}rcccccc}
\toprule
\multirow{2}{*}{\textbf{Target}} &
\multirow{2}{*}{$n_\ell$} &
\multicolumn{2}{c}{\textbf{TMTL}} &
\multicolumn{2}{c}{\textbf{MTL}} &
\multicolumn{2}{c}{\textbf{TSTL}} \\
\cmidrule(lr){3-4}\cmidrule(lr){5-6}\cmidrule(lr){7-8}
& & \textbf{Fused} & \textbf{Debiased} & \textbf{Target} &
\textbf{Full} & \textbf{Fused} & \textbf{Debiased} \\
\midrule
CD8$^+$ T CD69$^+$ CD45RA$^+$ & 105 & 0.97 (0.01) & \textbf{0.97 (0.01)} & 1.00 (0.03) & 0.98 (0.01) & 0.99 (0.01) & 1.01 (0.02) \\
CD8$^+$ T CD57$^+$ CD45RO$^+$ & 115 & 0.90 (0.01) & \textbf{0.79 (0.02)} & 0.81 (0.02) & 0.90 (0.01) & 0.92 (0.01) & 0.85 (0.02) \\
CD4$^+$ T activated integrinB7$^+$ & 207 & \textbf{0.88 (0.01)} & \textbf{0.88 (0.01)} & 0.96 (0.00) & \textbf{0.88 (0.01)} & 0.90 (0.01) & 0.90 (0.01) \\
T reg & 218 & 0.89 (0.02) & \textbf{0.88 (0.02)} & 0.94 (0.01) & 0.91 (0.01) & 0.92 (0.02) & 0.90 (0.02) \\
CD8$^+$ T CD49f$^+$ & 219 & 0.88 (0.01) & \textbf{0.87 (0.01)} & 0.95 (0.01) & 0.88 (0.01) & 0.91 (0.01) & 0.90 (0.01) \\
CD8$^+$ T TIGIT$^+$ CD45RA$^+$ & 421 & \textbf{0.99 (0.01)} & \textbf{0.99 (0.01)} & 1.00 (0.00) & 1.00 (0.01) & 1.00 (0.01) & 1.00 (0.01) \\
CD8$^+$ T naive & 457 & 0.94 (0.01) & 0.93 (0.01) & 0.94 (0.01) & \textbf{0.93 (0.01)} & 0.94 (0.01) & 0.94 (0.01) \\
CD8$^+$ T TIGIT$^+$ CD45RO$^+$ & 615 & 0.91 (0.01) & \textbf{0.86 (0.01)} & 0.88 (0.01) & 0.91 (0.01) & 0.93 (0.01) & 0.89 (0.01) \\
CD4$^+$ T naive & 1708 & 0.91 (0.00) & \textbf{0.90 (0.00)} & 0.91 (0.01) & 0.90 (0.00) & 0.91 (0.00) & 0.90 (0.00) \\
CD4$^+$ T activated & 1971 & 0.86 (0.00) & 0.85 (0.00) & \textbf{0.84 (0.01)} & 0.87 (0.00) & 0.88 (0.00) & 0.85 (0.00) \\
\bottomrule
\end{tabular}

\end{table}

\clearpage
\begin{table}[!t]
\centering
\ContinuedFloat
\caption{Relative MSE across target domains (continued): Site~3.}
\small
\renewcommand{\arraystretch}{1.1}
\setlength{\tabcolsep}{1.7pt}
\begin{tabular}{>{\raggedright\arraybackslash}p{0.21\textwidth}rcccccc}
\toprule
\multirow{2}{*}{\textbf{Target}} &
\multirow{2}{*}{$n_\ell$} &
\multicolumn{2}{c}{\textbf{TMTL}} &
\multicolumn{2}{c}{\textbf{MTL}} &
\multicolumn{2}{c}{\textbf{TSTL}} \\
\cmidrule(lr){3-4}\cmidrule(lr){5-6}\cmidrule(lr){7-8}
& & \textbf{Fused} & \textbf{Debiased} & \textbf{Target} &
\textbf{Full} & \textbf{Fused} & \textbf{Debiased} \\
\midrule
MAIT & 165 & 0.96 (0.01) & \textbf{0.89 (0.02)} & 0.90 (0.01) & 0.95 (0.01) & 0.97 (0.01) & 0.93 (0.01) \\
CD8$^+$ T CD57$^+$ CD45RO$^+$ & 171 & \textbf{0.97 (0.01)} & \textbf{0.97 (0.01)} & 0.99 (0.01) & 0.97 (0.01) & 0.98 (0.01) & 0.98 (0.02) \\
CD8$^+$ T TIGIT$^+$ CD45RO$^+$ & 209 & \textbf{0.97 (0.02)} & \textbf{0.97 (0.02)} & 1.00 (0.01) & \textbf{0.97 (0.02)} & 0.99 (0.01) & 0.99 (0.02) \\
CD4$^+$ T activated integrinB7$^+$ & 219 & 0.91 (0.01) & \textbf{0.90 (0.01)} & 0.91 (0.01) & 0.90 (0.01) & 0.92 (0.01) & 0.92 (0.01) \\
CD8$^+$ T CD69$^+$ CD45RA$^+$ & 245 & 0.99 (0.01) & 0.99 (0.01) & 1.00 (0.00) & \textbf{0.99 (0.01)} & 0.99 (0.01) & 0.99 (0.01) \\
CD8$^+$ T CD49f$^+$ & 274 & 0.91 (0.01) & \textbf{0.86 (0.01)} & 0.88 (0.01) & 0.89 (0.01) & 0.91 (0.01) & 0.87 (0.01) \\
CD8$^+$ T TIGIT$^+$ CD45RA$^+$ & 345 & 0.93 (0.01) & \textbf{0.91 (0.02)} & 0.92 (0.01) & 0.92 (0.02) & 0.94 (0.01) & 0.92 (0.02) \\
CD8$^+$ T CD69$^+$ CD45RO$^+$ & 357 & 0.97 (0.01) & \textbf{0.93 (0.01)} & 0.94 (0.01) & 0.97 (0.01) & 0.98 (0.01) & \textbf{0.93 (0.01)} \\
CD8$^+$ T CD57$^+$ CD45RA$^+$ & 444 & 1.00 (0.00) & 0.97 (0.01) & \textbf{0.97 (0.00)} & 0.99 (0.00) & 0.99 (0.00) & 0.98 (0.01) \\
CD8$^+$ T naive & 592 & \textbf{0.90 (0.01)} & \textbf{0.90 (0.01)} & 0.93 (0.01) & 0.91 (0.01) & 0.91 (0.00) & 0.90 (0.01) \\
CD4$^+$ T naive & 1221 & 0.89 (0.00) & \textbf{0.87 (0.01)} & 0.88 (0.00) & 0.89 (0.01) & 0.90 (0.00) & 0.89 (0.00) \\
CD4$^+$ T activated & 1660 & 0.89 (0.00) & \textbf{0.85 (0.00)} & 0.86 (0.00) & 0.87 (0.00) & 0.89 (0.00) & 0.86 (0.00) \\
\bottomrule
\end{tabular}

\end{table}

\clearpage
\begin{table}[!t]
\centering
\ContinuedFloat
\caption{Relative MSE across target domains (continued): Site~4.}
\small
\renewcommand{\arraystretch}{1.1}
\setlength{\tabcolsep}{1.7pt}
\begin{tabular}{>{\raggedright\arraybackslash}p{0.21\textwidth}rcccccc}
\toprule
\multirow{2}{*}{\textbf{Target}} &
\multirow{2}{*}{$n_\ell$} &
\multicolumn{2}{c}{\textbf{TMTL}} &
\multicolumn{2}{c}{\textbf{MTL}} &
\multicolumn{2}{c}{\textbf{TSTL}} \\
\cmidrule(lr){3-4}\cmidrule(lr){5-6}\cmidrule(lr){7-8}
& & \textbf{Fused} & \textbf{Debiased} & \textbf{Target} &
\textbf{Full} & \textbf{Fused} & \textbf{Debiased} \\
\midrule
CD8$^+$ T TIGIT$^+$ CD45RA$^+$ & 158 & 0.84 (0.01) & \textbf{0.75 (0.02)} & 0.78 (0.02) & 0.76 (0.02) & 0.83 (0.01) & 0.79 (0.02) \\
T reg & 170 & 0.85 (0.01) & 0.81 (0.03) & 0.85 (0.02) & \textbf{0.79 (0.02)} & 0.84 (0.01) & 0.83 (0.01) \\
CD8$^+$ T TIGIT$^+$ CD45RO$^+$ & 182 & 0.88 (0.01) & \textbf{0.82 (0.02)} & 0.86 (0.03) & 0.83 (0.02) & 0.87 (0.01) & 0.86 (0.02) \\
CD8$^+$ T CD69$^+$ CD45RA$^+$ & 223 & 0.90 (0.01) & \textbf{0.78 (0.02)} & 0.81 (0.02) & 0.81 (0.02) & 0.88 (0.01) & 0.81 (0.02) \\
CD4$^+$ T activated integrinB7$^+$ & 224 & 0.82 (0.01) & 0.74 (0.01) & 0.79 (0.01) & \textbf{0.73 (0.02)} & 0.80 (0.01) & 0.75 (0.01) \\
CD8$^+$ T CD49f$^+$ & 241 & 0.81 (0.02) & 0.75 (0.02) & 0.80 (0.02) & \textbf{0.73 (0.02)} & 0.80 (0.01) & 0.77 (0.01) \\
MAIT & 264 & 0.87 (0.02) & \textbf{0.78 (0.02)} & 0.81 (0.02) & 0.81 (0.02) & 0.86 (0.02) & 0.80 (0.02) \\
CD8$^+$ T CD69$^+$ CD45RO$^+$ & 331 & 0.85 (0.01) & 0.81 (0.01) & 0.84 (0.02) & \textbf{0.78 (0.02)} & 0.83 (0.01) & 0.81 (0.02) \\
CD8$^+$ T CD57$^+$ CD45RA$^+$ & 343 & 0.88 (0.01) & 0.74 (0.02) & 0.76 (0.02) & 0.78 (0.02) & 0.85 (0.01) & \textbf{0.74 (0.02)} \\
gdT CD158b$^+$ & 400 & 0.93 (0.01) & \textbf{0.78 (0.03)} & 0.78 (0.03) & 0.88 (0.02) & 0.92 (0.01) & 0.78 (0.03) \\
CD8$^+$ T naive & 747 & 0.84 (0.01) & \textbf{0.71 (0.01)} & 0.73 (0.01) & 0.76 (0.01) & 0.81 (0.01) & 0.73 (0.01) \\
CD8$^+$ T CD57$^+$ CD45RO$^+$ & 1002 & 0.85 (0.01) & \textbf{0.72 (0.01)} & 0.73 (0.01) & 0.77 (0.01) & 0.82 (0.01) & 0.73 (0.01) \\
CD4$^+$ T naive & 1110 & 0.85 (0.01) & 0.78 (0.01) & 0.80 (0.01) & 0.79 (0.01) & 0.82 (0.01) & \textbf{0.78 (0.01)} \\
CD4$^+$ T activated & 1328 & 0.84 (0.01) & \textbf{0.70 (0.01)} & 0.71 (0.01) & 0.74 (0.01) & 0.78 (0.00) & 0.72 (0.01) \\
\bottomrule
\end{tabular}

\end{table}